\documentclass[10pt,leqno,final]{amsart}
\usepackage[]{graphicx}
\usepackage[]{xcolor}
\makeatletter
\def\maxwidth{ %
  \ifdim\Gin@nat@width>\linewidth
    \linewidth
  \else
    \Gin@nat@width
  \fi
}
\makeatother

\definecolor{fgcolor}{rgb}{0.345, 0.345, 0.345}

\usepackage{framed}
\makeatletter
 {\par\unskip\endMakeFramed%
 \at@end@of@kframe}
\makeatother

\definecolor{shadecolor}{rgb}{.97, .97, .97}
\definecolor{messagecolor}{rgb}{0, 0, 0}
\definecolor{warningcolor}{rgb}{1, 0, 1}
\definecolor{errorcolor}{rgb}{1, 0, 0}
\newenvironment{knitrout}{}{} % an empty environment to be redefined in TeX

\usepackage{alltt}

\usepackage[round,elide]{natbib}
\usepackage{times}
\usepackage[neverdecrease]{paralist}
\usepackage[colorlinks=true]{hyperref}
\usepackage{multirow,bigstrut}
\usepackage[T1]{fontenc}

\usepackage{tikz}
\usetikzlibrary{arrows,arrows.meta,positioning,decorations,calc,math}

\usepackage{mathabx,mathrsfs,eucal,bbm,amssymb,mathtools,accents}

\definecolor{maize}{rgb}{0.836,0.660,0.016}
\definecolor{grey}{rgb}{0.484,0.492,0.496}
\definecolor{purple}{rgb}{0.518,0.082,0.765}
\definecolor{darkgreen}{rgb}{0,0.392,0}
\definecolor{royalblue}{rgb}{0.263,0.431,0.933}

\newenvironment{mathsize}[2]{\fontsize{#1}{#2}\selectfont}{\normalfont}

\newcommand\Set[1]{\left\{{#1}\right\}}
\newcommand\CondSet[2]{\left\{{#1}\;\middle\vert\;{#2}\right\}}

\newcommand\Prob[1]{\mathsf{Prob}\left[{#1}\right]}
\newcommand\CondProb[2]{\mathsf{Prob}\left[{#1}\;\middle\vert\;{#2}\right]}
\newcommand\Expect[1]{\mathbbm{E}\left[{#1}\right]}
\newcommand\CondExpect[2]{\mathbbm{E}\left[{#1}\;\middle\vert\;{#2}\right]}
\newcommand\indicator[1]{\mathbbm{1}_{{#1}}}
\newcommand\Indicator[1]{\mathbbm{1}\{{#1}\}}
\newcommand\lik{\mathcal{L}}

\newcommand{\pkg}[1]{\texttt{\textbf{#1}}}
\newcommand{\dd}{\mathrm{d}}
\newcommand{\halfopen}[1]{\left[{#1}\right)}
\newcommand{\halfclosed}[1]{\left({#1}\right]}

\newcommand\leftlim[1]{\widetilde{#1}}
\newcommand\rightlim[1]{\undertilde{#1}}
\newcommand\rdm[1]{\boldsymbol{#1}}
\newcommand\spc[1]{\mathbb{#1}}
\newcommand\deme[1]{\mathrm{#1}}
\newcommand\jump[1]{\mathsf{#1}}
\newcommand\func[1]{\mathsf{#1}}
\newcommand\lab[1]{\mathsf{#1}}
\newcommand\col[1]{\func{d}({#1})}
\newcommand\ctr[1]{\func{m}({#1})}
\newcommand\cols[2]{\func{col}_{#1}({#2})}
\newcommand\fork[3]{{\kappa}_{#1}^{#2}\mkern-1mu{#3}}
\newcommand\chop[3]{{\chi}_{#1}^{#2}\mkern-1mu{#3}}
\newcommand\swap[3]{{\sigma}_{#1}^{#2}\mkern-1mu{#3}}
\newcommand\backfork[3]{{\leftlim{\kappa}}_{#1}^{#2}\mkern-1mu{#3}}
\newcommand\backchop[3]{{\leftlim{\chi}}_{#1}^{#2}\mkern-1mu{#3}}
\newcommand\backswap[3]{{\leftlim{\sigma}}_{{#1}}^{#2}\mkern-1mu{#3}}
\newcommand\event[1]{\func{ev}({#1})}

\def\emptyset{\varnothing}
\def\preceq{\preccurlyeq}
\def\cadlag{{c\`adl\`ag}}
\def\caglad{{c\`agl\`ad}}
\def\Xspace{\spc{X}}
\def\Jumps{\spc{U}}
\def\Demes{\spc{D}}

\def\Zp{\spc{Z}_{+}}
\def\R{\spc{R}}
\def\Rp{\spc{R}_{+}}

\def\ie{{i.e.,} }
\def\eg{{e.g.,} }
\def\cf{{cf.}~}
\def\st{\text{s.t.}}

\def\prune{\func{prune}}
\def\obs{\func{obs}}
\def\leaves{\mathbb{L}}
\def\partit{\func{partit}}
\def\time{\func{t}}
\def\ColorSet{\Demes\times\Set{0,1}}

\def\X{X}
\def\Xr{\rdm{\X}}
\def\Xt{\leftlim{\X}}
\def\Xrt{\rdm{\Xt}}
\def\Xh{\hat{\X}}
\def\Xrh{\rdm{\hat{\X}}}
\def\U{U}
\def\Ur{\rdm{\U}}

\def\Uh{\hat{\U}}
\def\Urh{\rdm{\hat{\U}}}
\def\Y{Y}
\def\Yr{\rdm{\Y}}
\def\Yt{\leftlim{\Y}}
\def\Yrt{\rdm{\Yt}}

\def\fY{\func{Y}}
\def\Z{Z}
\def\Zr{\rdm{\Z}}
\def\Zt{\leftlim{\Z}}

\def\T{T}
\def\Th{\hat{\T}}
\def\Trh{\rdm{\hat{\T}}}
\def\K{K}
\def\Kr{\rdm{\K}}
\def\H{H}
\def\Hr{\rdm{\H}}
\def\G{G}
\def\Gr{\rdm{\G}}
\def\Gh{\hat{\G}}
\def\Grh{\rdm{\Gh}}
\def\y{y}
\def\yr{\rdm{\y}}
\def\yrh{\rdm{\hat{\y}}}
\def\yt{\leftlim{\y}}
\def\yrt{\rdm{\yt}}
\def\P{P}
\def\Ph{\hat{P}}
\def\Pr{\rdm{\P}}
\def\Phir{\rdm{\Phi}}
\def\Phit{\leftlim{\Phi}}
\def\Phirt{\rdm{\Phit}}
\def\Philt{\rdm{\rightlim{\Phi}}}
\def\Phih{\hat{\Phi}}
\def\Phirh{\rdm{\Phih}}
\def\Vr{\rdm{V}}
\def\ut{\leftlim{u}}
\def\vt{\leftlim{v}}
\def\wt{\leftlim{w}}

\renewcommand\binom[2]{{\scriptscriptstyle\begin{pmatrix}{#1}\\{#2}\end{pmatrix}}}
\newcommand\binratio[4]{{\scriptscriptstyle\begin{pmatrix}{#1}&{#2}\\{#3}&{#4}\end{pmatrix}}}
\newcommand\tbinratio[4]{\begin{psmallmatrix}{#1}&{#2}\\{#3}&{#4}\end{psmallmatrix}}

\usepackage{chngcntr}
\counterwithout{equation}{section}
\renewcommand\thesection{\arabic{section}}
\renewcommand\thesubsection{\thesection.\arabic{subsection}}
\renewcommand\thefigure{\arabic{figure}}
\renewcommand\thetable{\arabic{table}}

\usepackage[small,compact]{titlesec}
\newcommand\periodafter[1]{{#1}.}
\titleformat{\section}[hang]{\large\bfseries}{\periodafter\thesection}{2ex}{}{}
\titleformat{\subsection}[runin]{\normalsize\bfseries}{\periodafter\thesubsection}{1ex}{\periodafter}{}
\titleformat{\subsubsection}[runin]{\normalsize\bfseries}{\periodafter\thesubsubsection}{1ex}{\periodafter}{}
\titleformat{\paragraph}[runin]{\normalsize\bfseries}{\theparagraph}{0em}{\periodafter}{}
\titlespacing*{\section}{0em}{*1}{*1}
\titlespacing{\subsection}{0em}{*0}{*1}
\titlespacing{\subsubsection}{0em}{*0}{*1}
\titlespacing{\paragraph}{0em}{*0}{*1}

\usepackage{zref-clever}
\AddToHook{env/thm/begin}{\zcsetup{countertype={thm=theorem}}}
\AddToHook{env/corol/begin}{\zcsetup{countertype={thm=corollary}}}
\zcsetup{cap=true,sort=true,comp=true,}
\zcRefTypeSetup{theorem}{
Name-sg = Theorem,
Name-pl = Theorems,
}
\zcRefTypeSetup{corollary}{
Name-sg = Corollary,
Name-pl = Corollaries,
}
\zcRefTypeSetup{lemma}{
Name-sg = Lemma,
Name-pl = Lemmas,
}
\zcRefTypeSetup{figure}{
Name-sg = Fig.,
Name-pl = Figs.,
}
\zcRefTypeSetup{equation}{
Name-sg = Eq.,
Name-pl = Eqs.,
refbounds = {,,,},
}
\zcRefTypeSetup{pluralequation}{
Name-sg = Eqs.,
Name-pl = Eqs.,
refbounds = {,,,},
}
\zcRefTypeSetup{section}{
Name-sg = \S\!,
Name-pl = \S\S\!,
}

\theoremstyle{plain}
\newtheorem{thm}{Theorem}
\newtheorem{corol}[thm]{Corollary}

\newtheorem{lemma}{Lemma}
\theoremstyle{definition}
\newtheorem*{defn}{Definition}

\theoremstyle{remark}
\newtheorem{remark}{Remark}

\usepackage{algorithm,caption}
\usepackage[noend]{algpseudocode}

\title[Exact Phylodynamic Likelihood]{Exact Phylodynamic Likelihood\\ via Structured Markov Genealogy Processes}
\author[King]{Aaron~A.~King}
\address{
  A.~A.~King,
  Department of Ecology \& Evolutionary~Biology,
  Center for the Study of Complex~Systems, \&
  Department of Mathematics,
  University of Michigan,
  Ann~Arbor, MI~48109~USA
  {and}
  Santa~Fe~Institute,
  1399 Hyde~Park~Road,
  Santa~Fe, NM~87501~USA
}
\email{kingaa@umich.edu}
\urladdr{\href{https://kinglab.eeb.lsa.umich.edu/}{https://kinglab.eeb.lsa.umich.edu/}}
\author[Lin]{Qianying~Lin}
\address{
  Q.-Y.~Lin,
  Department of Biostatistics,
  College of Public Health,
  Ohio State University,
  Columbus, OH~43210~USA
}
\author[Ionides]{Edward~L.~Ionides}
\address{
  E.~L.~Ionides,
  Department of Statistics,
  University of Michigan,
  Ann~Arbor, MI~48109~USA
}
\date{\today}

\hypersetup{pdftitle={Exact Phylodynamic Likelihood via Structured Markov Genealogy Processes}}
\hypersetup{pdfauthor={A.A. King, Q.-Y. Lin, E.L. Ionides}}
\hypersetup{urlcolor=black,citecolor=black,linkcolor=black,filecolor=black}

\IfFileExists{upquote.sty}{\usepackage{upquote}}{}
\begin{document}

\begin{abstract}
  We show that each member of a broad class of Markovian population models induces a unique stochastic process on the space of genealogies.
  We construct this genealogy process and derive exact expressions for the likelihood of an observed genealogy in terms of a filter equation, the structure of which is completely determined by the population model.
  We show that existing phylodynamic methods based on the coalescent and linear birth-death processes are special cases.
  We derive some properties of filter equations and describe a class of algorithms that can be used to numerically solve them.
  Importantly, because these algorithms rely only on simulation of the population model, they retain the plug-and-play property upon which simulation-based inference depends.
  Our results open the door to statistically efficient likelihood-based phylodynamic inference for a much wider class of models than has been possible.
\end{abstract}

\maketitle

\section{Introduction}

When the genome of an infectious agent accumulates mutations on timescales similar to those of transmission and infection progression, the resulting pattern of differences among genomes contains information on the history of the pathogen's passage through individual hosts and the host population.
As \citet{Grenfell2004} observed, one can extract this information to gain insight into the structure and dynamics of the host-pathogen system.
In particular, one can formalize mathematical models of transmission, estimate their parameters, and compare their ability to explain data, following standard statistical paradigms.
This collection of tasks is known as \emph{phylodynamic} inference;
\citet{Alizon2024} provides a recent review.

A common approach to phylodynamic inference involves building a mathematical linkage between the tree-like \emph{genealogy} or \emph{phylogeny} that expresses the relationships of shared ancestry among sampled genomes and a model of the dynamics of the transmission system.
Various linkages are possible, but, to attain maximal statistical efficiency (\ie lose the least information), it is desirable to be able to compute the likelihood function for models of interest.
The likelihood function is the probability density of a given set of data, conditional on a given model, viewed as a function of the parameters of that model.
While various kinds of data may be available, we focus here on the case where the data are a set of genomic sequences.
Notationally, if the data, $S$, is a set of genome sequences, $\Phi$ a genealogical tree relating these sequences, $E$ a model of sequence evolution, and $D$ a dynamic transmission model, then the full likelihood is
\begin{equation*}
  \lik(D,E) = f(S|D,E) = \int{f(S|\Phi,E)\,f(\Phi|D)\,\dd{\Phi}},
\end{equation*}
where the integral is taken over all possible genealogies and we somewhat loosely use the symbol $f$ for the various distinct probability densities, the nature of each of which is clear from its arguments.
In this expression, $f(S|\Phi,E)$ is typically the \citet{Felsenstein2004} phylogenetic likelihood.
The function $f(\Phi|D)$, which links the phylogeny to the dynamic model, is the \emph{phylodynamic likelihood}.
%% Although we suppress the dependencies here, the likelihoods $\lik$, $f(S|\Phi,E)$, and $f(\Phi|D)$ are to be viewed as functions of the parameters that specify the models $D$ and $E$.
In the Bayesian context, the phylodynamic likelihood is sometimes referred to as a \emph{tree prior} \citep{Moeller2018,Volz2018}.
Its computation has remained out of reach, except in several special cases.
This paper presents theory that facilitates its computation for a very broad range of dynamic models.

%% Phylodynamic inference rests upon mathematical linkage between mathematical models of disease transmission at the population level and genome-sequence data collected from individual infected hosts.
%% The tightest possible such linkage is effected by a computable likelihood function.
%% The tree-like \emph{genealogy} or \emph{phylogeny} that expresses the relationships of shared ancestry among the sequences.
%% In particular, if one can determine the likelihood of any given genealogy given a model and also the likelihood of the data given a genealogy, then the likelihood of the data given the model is obtained by summing over the space of genealogies.

With few exceptions, existing approaches to the phylodynamic likelihood have been based on one of two mathematical idealizations.
The first is the \citet{Kingman1982a} coalescent, by which the likelihood of a given genealogy is computed using a reverse-time argument.
This computation provides the exact likelihood for a genealogy resulting from a particular, constant population-size, dynamic model \citep[the Moran model, \eg][]{Moran1958,Kingman1982b,Moehle2000,King2020,King2022}.
Extensions of this approach develop approximate likelihoods for the case when the population size varies as a function of time \citep{Griffiths1994,Drummond2005} or according to an SIR process \citep{Volz2009a,Rasmussen2011}, as long as the population size is large and the sample-fraction remains negligible.
The second idealization is the linear birth-death process, for which the likelihood is available in closed form \citep{Stadler2010}.
Linearity in this context amounts to the assumption that distinct lineages do not interact:
it is the resulting self-similarity of genealogies that renders the likelihood analytically tractable.
Extensions of this approach develop approximations via linearization of nonlinear processes or restriction to scenarios in which population growth is nearly linear \citep[\eg][]{MacPherson2021}.
Although the tractability of these approaches is appealing, concern naturally arises as to validity of the approximations in specific cases, the biases introduced by them, and the amount of information in data left uncaptured by these approximate methods.
For this reason, there is interest in improved phylodynamic inference techniques.

What would an ideal phylodynamic inference method look like?
First, it would allow us to ask the questions we wish to ask.
Since the set of scientifically interesting hypotheses is not contained within the set of statistically convenient models, an ideal phylodynamic inference methodology would put no restrictions on the form of the models that it can accommodate.
In particular, since nonlinearity, nonstationarity, noise, and measurement error are prominent and ubiquitous in epidemiology, such a method would accommodate nonlinear, time-inhomogeneous, stochastic transmission models.
Moreover, because many of the most scientifically and practically important uncertainties concern heterogeneities in transmission rates and the susceptibility, behavior, age, and location of hosts, it would accommodate host populations structured by these and other factors.
Second, it would be both statistically efficient and robust to model misspecification.
In particular, it would achieve maximal statistical efficiency by being based on the exact phylodynamic likelihood, which would also facilitate objective comparison among parameterizations and models.
At the same time, model predictions would be robust to small deviations from the model assumptions.
Finally, an ideal inference method would be computationally efficient.
If numerical computations were absolutely required for its use, these would scale well with model complexity and data volume.

Of course, there is no hope for the realization of such an ideal method.
In particular, in selecting methodology, one must navigate trade-offs between statistical efficiency and robustness, and between computational expense and fidelity of model to question.
With respect to these trade-offs, far more attention has been paid to date on relatively inexpensive approaches that are restricted to models that can only approximate the motivating questions to a greater or lesser extent.
The present paper offers a complementary perspective:
we develop mathematical expressions, and corresponding numerical algorithms, for the phylodynamic likelihood of almost arbitrarily complex models.
Our focus on the exact likelihood offers maximum information capture, even as the computations it involves are somewhat delicate.
Numerical examples are postponed to \zcref{sec:worked-examples}, and computational scalability is discussed in \zcref{sec:discussion}.

While the theory presented here greatly expands the class of models for which exact likelihood is computable, it is limited to models with \emph{discrete structure}.
Thus, while structuring factors such as age, stage of infection progression, and spatial location are most naturally expressed in terms of continuous variables, to apply the results presented here, one would have to discretize these factors.
We suspect that this limitation may often be relatively painless in practice, since discretely structured models have repeatedly proved their value in epidemiology.
In particular, compartmental models offer great flexibility and have often been used as approximations when continuous structure leads to uncomfortably high model dimension.

To connect a model at the level of a population with genealogies based on samples taken from individual hosts, it is necessary to make assumptions about the individuals in the population.
The simplest such assumption is that the individuals that are identical with respect to the population dynamics are fully statistically identical.
That is, that they are \emph{exchangeable}.
In a compartmental model, this is tantamount to the assumption that the residence times of the individuals within each compartment are identically distributed, though not independent.
Although exchangeability is indeed an additional assumption, it is so natural that it is frequently unrecognized as such, and one often reads statements to the effect that exchangeability of individuals is a consequence of the Markovian assumption.
Nonetheless, since it adds minimal additional structure, it is the natural assumption, the one adopted by the coalescent- and birth-death-process-based approaches described above, and the one we will make in this paper.

From the mathematical point of view, a fundamental difficulty in determining expressions for the likelihood of a genealogy generated by a given population process is the fact that a genealogy's structure is determined by events occurring at potentially widely separated times.
For example, the appearance of a branch in an observed genealogy depends both on there having been a divergence of lineages at that time, but also on two of the emerging lineages having been sampled, perhaps much later.
As discussed above, both coalescent and linear birth-death approaches to phylodynamic inference employ restrictive assumptions to get around this non-locality.
The approaches of \citet{Vaughan2019} and \citet{King2022} avoid the need for such restrictions by exploiting the observation that, conditional on population history, there is a factorization of the likelihood with time-localized factors.
The main contribution of this paper is the extension of these results to accommodate population structure.
Additionally, the proofs given here are both simpler and more general than the ones given in the earlier papers.

In the following, we take as our starting point a population model in the form of a discretely structured continuous-time Markov process.
We show how such a process uniquely induces several stochastic processes in the space of genealogies.
We go on to show that exact likelihoods for these genealogies can be computed by solving a certain \emph{filter equation}.
In particular, for each given population-level model there is a specific filter equation.
We then show that our approach generalizes both the coalescent- and linear birth-death-process-based approaches.
Indeed, the filter equations for the models on which these approaches are based can be solved to obtain the familiar likelihood formulae.
In appendices, we develop the properties of filter equations generally, provide a sequential Monte Carlo algorithm that can be used when numerical integration is needed for their solution, and work out the filter equations and their solutions in two examples.
At the end of the main text, \zcref{tab:symbols} is an index of the mathematical notation used in the paper.

The codes used to generate the figures in this paper are archived on Zenodo (DOI:~\href{https://doi.org/10.5281/zenodo.21461034}{10.5281/zenodo.21461034}).
The open-source \pkg{R} package \pkg{phylopomp} (\url{https://github.com/kingaa/phylopomp}) implements the simulation and likelihood-computation algorithms employed here.

\section{Mathematical preliminaries}

\subsection{Notation}
\label{sec:notation}

Throughout the paper, we will adopt the convention that a boldface symbol (\eg $\Xr$), denotes a random element.
We will be concerned with a variety of stochastic processes, in both discrete and continuous time.
In both cases, we will use a subscript to indicate the time parameter: \eg $\Xr_t$ or $\Gr_k$, where $t$ takes values in the non-negative reals $\Rp$ and $k$ in the non-negative integers $\Zp$.
In the case of continuous-time processes, we will find ourselves needing to refer to their left- or right-limits.
Accordingly, if $\Phir_t$ is any process (random or not), let
\begin{equation*}
  \begin{gathered}
    \Phirt_t\coloneq
    \displaystyle\lim_{s\,\uparrow\,{t}}\;\Phir_{s}
    \qquad\text{and}\qquad
    \Philt_t\coloneq
    \displaystyle\lim_{s\,\downarrow\,{t}}\;\Phir_{s}.
  \end{gathered}
\end{equation*}
Most of the processes treated here will either be right-continuous with left limits (\ie \cadlag, so that $\Phir_t=\Philt_t$) or left-continuous with right limits (\caglad, \ie $\Phirt_t=\Phir_t$).
Note that, if $\Phir_t$ is \cadlag, then $\Phirt_t$ is \caglad\ and if $\Phir_t$ is \caglad, then $\Philt_t$ is \cadlag.

If $\Phir_t$, $t\in\Rp$ is a pure jump process, knowledge of its sample path is equivalent to knowledge of the number, $\Kr_t$, of jumps it has taken as of time $t$, the jump times $\Trh_k$, and the embedded chain $\Phirh^{}_k\coloneq{\Phir_{\Trh^{}_k}}$, $k=0,\dots,\Kr^{}_t$.
In particular, $\Phir_t$ has \cadlag\ sample paths and we adopt the convention that $\Trh_0=0$ and $\Trh_{\Kr_t+1}=t$, then
$\Phir_s=\Phirh_k$ for $s\in\halfopen{\Trh_k,\Trh_{k+1}}$, $k=0,\dots,\Kr_t$.

We will use the symbol $\indicator{P}$ or $\Indicator{P}$ to denote an indicator, which is $1$ when $P$ is true and $0$ when it is false.
A complete index to the symbols and acronyms used in the paper is to be found in \zcref{tab:symbols}.

\subsection{Population process}
\label{sec:population-process}

We are motivated by the desire for exact phylodynamic inference methods for as wide a class of epidemiological models as possible.
In particular, we would like to be able to formulate and parameterize an arbitrary compartmental model and to quantify its ability to explain data using likelihood.
\zcref[S]{fig:example-models} depicts a few of the simplest such models in order to give a sense of the kinds of complexities that can arise.
Of course, with the ability to entertain models with countably many compartments, much greater complexity is possible.
In particular, one can model not only complex infection progression, but also strain structure, behavioral structure, age structure, and spatial structure using compartmental models.
As is well known, one can discretize continuous structure-variables and employ the linear chain trick to accommodate non-exponential residence times \citep[e.g.,][]{Hurtado2021,Keeling2008}.
While the utility of these approximations will vary, a very wide range of model assumptions lie within the scope of the theory presented here.

\begin{figure}
  \begin{center}
  \resizebox{0.9\linewidth}{!}{
    \begin{tikzpicture}[scale=1]
      \tikzstyle{box}=[draw=black, text=black, fill=white, very thick, minimum size=3em]
      \tikzstyle{ibox}=[draw=darkgreen]
      \tikzstyle{deme}=[fill=black!20!white]
      \tikzstyle{label}=[font=\Large]
      \tikzstyle{coordinate}=[inner sep=0pt,outer sep=0pt]
      \tikzstyle{flow}=[draw=black, very thick, >=stealth]
      \tikzstyle{modulate}=[draw=darkgreen, >=Circle]

      %% SEIRS model
      \coordinate (origin) at (0,0);
      \node[label] (lab) at (origin) {\textbf{A}};
      \node [box] (S) at ($(origin)+(1,-1)$) {$\deme{S}$};
      \node [box,deme] (E) at ($(S)+(2,0)$) {$\deme{E}$};
      \node [box,ibox,deme] (I) at ($(E)+(2,0)$) {$\deme{I}$};
      \node [box] (R) at ($(I)+(2,0)$) {$\deme{R}$};
      \coordinate (overR) at ($(R)+(0,1)$);
      \coordinate (midSE) at ($(S)!0.5!(E)$);
      \draw [flow,->] (S) -- (E);
      \draw [flow,->] (E) -- (I);
      \draw [flow,->] (I) -- (R);
      \draw [flow,->] (R) -- (overR) -- (S |- overR) -- (S);
      \draw [modulate,->] (I.north west) .. controls ($(I)+(-1,1)$) and ($(midSE)+(0,1)$) .. (midSE);
      \node[font=\normalsize] (demes) at ($(E)!0.5!(I)+(0,-1)$) {$\Demes=\{\deme{E},\deme{I}\}$};

      %% SIIR (two-strain) model
      \coordinate (origin) at (0,-3);
      \node[label] (lab) at (origin) {\textbf{B}};
      \node [box] (S) at ($(origin)+(1,-2)$) {$\deme{S}$};
      \node [box] (R1) at ($(S)+(6,1)$) {$\deme{R}_1$};
      \node [box] (R2) at ($(S)+(6,-1)$) {$\deme{R}_2$};
      \coordinate (overR) at ($(R1)+(0,1)$);
      \coordinate (underR) at ($(R2)-(0,1)$);
      \node [box,deme] (E1) at ($(S)+(2,1)$) {$\deme{E}_1$};
      \node [box,deme] (E2) at ($(S)+(2,-1)$) {$\deme{E}_2$};
      \node [box,ibox,deme] (I1) at ($(E1)+(2,0)$) {$\deme{I}_1$};
      \node [box,ibox,deme] (I2) at ($(E2)+(2,0)$) {$\deme{I}_2$};
      \draw [flow,->] (S) -- (E1);
      \draw [flow,->] (E1) -- (I1);
      \draw [flow,->] (I1) -- (R1);
      \draw [flow,->] (S) -- (E2);
      \draw [flow,->] (E2) -- (I2);
      \draw [flow,->] (I2) -- (R2);
      \draw [flow,->] (R1) -- (overR) -- (S |- overR) -- (S);
      \draw [flow,->] (R2) -- (underR) -- (S |- underR) -- (S);
      \coordinate (midSE1) at ($(S)!0.5!(E1)$);
      \coordinate (midSE2) at ($(S)!0.5!(E2)$);
      \draw [modulate,->] (I1.north west) .. controls ($(I1)+(-1,1)$) and ($(midSE1)+(0,1.8)$) .. (midSE1);
      \draw [modulate,->] (I2.south west) .. controls ($(I2)+(-1,-1)$) and ($(midSE2)+(0,-1.8)$) .. (midSE2);
      \node[font=\normalsize] (demes) at ($(E2)!0.5!(I2)+(0,-2)$) {$\Demes=\{\deme{E}_1,\deme{E}_2,\deme{I}_1,\deme{I}_2\}$};

      %% COVID model
      \coordinate (origin) at (9,0);
      \node[label] (lab) at (origin) {\textbf{C}};
      \node [box] (S) at ($(origin)+(1,-2)$) {$\deme{S}$};
      \node [box] (R) at ($(S)+(7,1)$) {$\deme{R}$};
      \coordinate (overR) at ($(R)+(0,1)$);
      \node [box,deme] (E) at ($(S)+(2,0)$) {$\deme{E}$};
      \node [box,ibox,deme] (Ia) at ($(E)+(2,1)$) {$\deme{I_A}$};
      \node [box,ibox,deme] (Is) at ($(E)+(2,-1)$) {$\deme{I_S}$};
      \node [box] (H) at ($(Is)+(2,0)$) {$\deme{H}$};
      \node [box] (D) at ($(H)+(2,0)$) {$\deme{D}$};
      \draw [flow,->] (S) -- (E);
      \draw [flow,->] (E) -- (Ia.west);
      \draw [flow,->] (E) -- (Is.west);
      \draw [flow,->] (Ia) -- (R);
      \draw [flow,->] (Is) -- (H);
      \draw [flow,->] (H) -- (R);
      \draw [flow,->] (Is) -- (R);
      \draw [flow,->] (H) -- (D);
      \draw [flow,->] (R) -- (overR) -- (S |- overR) -- (S);
      \coordinate (midSE) at ($(S)!0.5!(E)$);
      \draw [modulate,->] (Ia.north west) .. controls ($(Ia.north west)+(165:1)$) and ($(midSE)+(0,2)$) .. (midSE);
      \draw [modulate,->] (Is.south west) .. controls ($(Is.south west)+(195:1)$) and ($(midSE)+(0,-2)$) .. (midSE);
      \node[font=\normalsize] (demes) at ($(Is)+(0,-1)$) {$\Demes=\{\deme{E},\deme{I_A},\deme{I_S}\}$};

      %% SI2R (super-spreader) model
      \coordinate (origin) at (9,-5);
      \node[label] (lab) at (origin) {\textbf{D}};
      \node [box] (S) at ($(origin)+(2,-1)$) {$\deme{S}$};
      \node [box,deme] (E) at ($(S)+(2,0)$) {$\deme{E}$};
      \node [box,ibox,deme] (Il) at ($(E)+(2,0)$) {$\deme{I_L}$};
      \node [box,ibox,deme] (Ih) at ($(Il)+(0,-2)$) {$\deme{I_H}$};
      \node [box] (R) at ($(Il)+(2,0)$) {$\deme{R}$};
      \coordinate (overR) at ($(R)+(0,1)$);
      \coordinate (midSE) at ($(S)!0.5!(E)$);
      \draw [flow,->] (S) -- (E);
      \draw [flow,->] (E) -- (Il);
      \draw [flow,<->] (Il) -- (Ih);
      \draw [flow,->] (Il) -- (R);
      \draw [flow,->] (Ih.east) -- (R |- Ih.east) -- (R);
      \draw [flow,->] (R) -- (overR) -- (S |- overR) -- (S);
      \draw [modulate,->] (Il.north west) .. controls ($(Il.north west)+(135:0.8)$) and ($(midSE)+(0,1)$) .. (midSE);
      \draw [modulate,->] (Ih.west) .. controls ($(Ih.west)+(180:0.8)$) and ($(midSE)+(0,-1)$) .. (midSE);
      \node[font=\normalsize] (demes) at (midSE |- Ih) {$\Demes=\{\deme{E},\deme{I_L},\deme{I_H}\}$};
    \end{tikzpicture}
  }
\end{center}
  \caption{
    Examples of discretely-structured population models.
    Demes are shaded.
    Compartments containing infectious hosts are outlined in green.
    Curved green lines connect transmission rates with the compartments whose occupancies control their modulation;
    each such connection gives rise to a nonlinearity in the model.
    \textbf{(A)} An SEIRS model.
    Susceptible individuals ($\deme{S}$), once infected, enter a transient incubation phase ($\deme{E}$) before they become infectious ($\deme{I}$).
    Upon recovery ($\deme{R}$), individuals experience immunity from reinfection.
    If this immunity wanes, they re-enter the susceptible compartment.
    Pathogen lineages are to be found in hosts within the $\deme{E}$ and $\deme{I}$ compartments only.
    Accordingly, there are two demes: $\Demes=\Set{\deme{E},\deme{I}}$.
    If there is exactly one lineage per host, then the occupancy, $n(\Xr_t)=(n_{\deme{E}}(\Xr_t),n_{\deme{I}}(\Xr_t))$, is the integer 2-vector giving the numbers of hosts in the respective compartments.
    See \zcref{sec:demes} for definition and discussion of demes and deme occupancy.
    \textbf{(B)} In this four-deme model, two distinct pathogen strains compete for susceptibles.
    \textbf{(C)} A three-deme model according to which, after an incubation period, hosts may develop asymptomatic infection ($\deme{I_A}$).
    If they do not recover, symptomatically infected hosts ($\deme{I_S}$) can progress to hospitalization ($\deme{H}$) and death ($\deme{D}$).
    \textbf{(D)} A three-deme model with heterogeneity in transmission behavior.
    Contagious individuals move randomly between low-transmission ($\deme{I_L}$) and high-transmission ($\deme{I_H}$) behaviors.
    \label{fig:example-models}
  }
\end{figure}

We will refer to the stochastic process defined by the given model as the \emph{population process}.
We will assume that the population process is a time-inhomogeneous Markov jump process, $\Xr_t$, $t\in\Rp$, taking values in some space $\Xspace$.
In earlier work \citep{King2022}, we limited ourselves to the case $\Xspace=\mathbb{Z}^d$, but here we assume only that $\Xspace$ is a complete metric measure space with a countable dense subset.
The population process is completely specified by its initial-state density, $p_0$, and its transition rates $\alpha$.
In particular, we suppose that
\begin{equation}
  \label{eq:ic}
  \Prob{\Xr_0\in\mathcal{E}}=\int_{\mathcal{E}}{p_0(x)\,\dd{x}}
\end{equation}
for all measurable sets $\mathcal{E}\subseteq\Xspace$.
For any $t\in\Rp$, $x,x'\in\Xspace$, we think of the quantity $\alpha(t,x,x')$ as the instantaneous hazard of a jump from $x$ to $x'$.
More precisely, the transition rates have the following properties:
\begin{equation*}
  \begin{gathered}
    \alpha(t,x,x')\ge{0}, \qquad \int_{\Xspace}{\alpha(t,x,x')\,\dd{x'}}<\infty,\\
  \end{gathered}
\end{equation*}
for all $t\in\Rp$ and $x,x'\in\Xspace$ and that, as a function of time, $\alpha$ is \cadlag\ and continuous almost everywhere.
Henceforth, we understand that integrals are taken over all of $\Xspace$ unless otherwise specified.
Let $\Kr_t$ be the number of jumps that $\Xr$ has taken by time $t$.
We assume that $\Kr_t$ is a simple counting process so that
\begin{equation*}
  \begin{gathered}
    \CondProb{\Kr_{t+\Delta}=n+1}{\Kr_{t}=n,\Xr_{t}=x}=\Delta\,\int{\alpha(t,x,x')\,\dd{x'}}+o(\Delta),\\
    \CondProb{\Kr_{t+\Delta}>n+1}{\Kr_{t}=n}=o(\Delta),\\
    \CondProb{\Xr_{t+\Delta}\in\mathcal{E}\setminus\Set{x}}{\Xr_{t}=x}=\Delta\,\int_{\mathcal{E}\setminus\Set{x}}{\alpha(t,x,x')\,\dd{x'}}+o(\Delta),\\
    %% \CondProb{\Xr_{t+\Delta}\in\mathcal{E}}{\Xr_{t}=x, \Kr_{t+\Delta}-\Kr_{t}=1}=\frac{\int_{\mathcal{E}}{\alpha(t,x,x')\,\dd{x'}}}{\int{\alpha(t,x,x')\,\dd{x'}}}+o(\Delta).
    \CondProb{\Xr_{t+\Delta}=x}{\Xr_{t}=x}=1-\Delta\,\int{\alpha(t,x,x')\,\dd{x'}}+o(\Delta).
  \end{gathered}
\end{equation*}
We further assume that $\Xr_t$ has \cadlag\ sample paths and is \emph{non-explosive}, \ie $\Prob{\Kr_t<\infty}=1$ for all $t<\infty$.

\subsection{Kolmogorov equations}
\label{sec:kolmogorov-eqns}

The above may be compactly summarized by stating that if $v(t,x)$ satisfies the Kolmogorov forward equation (KFE),
\begin{equation}
  \label{eq:kfe}
  \begin{gathered}
    \frac{\partial{v}}{\partial{t}}(t,x)
    =\int\!{v(t,x')\,\alpha(t,x',x)\,\dd{x'}}
    -\int\!{v(t,x)\,\alpha(t,x,x')\,\dd{x'}},
    \qquad
    v(0,x) = p_0(x),
  \end{gathered}
\end{equation}
for $t\in[0,T]$, then $\int_{\mathcal{E}}\!{v(t,x)\,\dd{x}}=\Prob{\Xr_t\in\mathcal{E}}$ for every $t\in[0,T]$ and measurable $\mathcal{E}\subseteq{\Xspace}$.
\zcref[S]{eq:kfe} is sometimes called the \emph{master equation} for $\Xr_t$.
The adjoint form of the KFE is the Kolmogorov backward equation,
\begin{equation}
  \label{eq:kbe}
  \begin{gathered}
    -\frac{\partial{F}}{\partial{s}}(s,x)
    =\int\!{\alpha(t,x,x')\,\left[F(s,x')-F(s,x)\right]\dd{x'}},
    \qquad
    F(T,x) = f(x).
  \end{gathered}
\end{equation}
If $F$ satisfies \zcref{eq:kbe} for $s\in[0,T]$ and $x\in\Xspace$, then $F(s,x)=\CondExpect{f(\Xr_T)}{\Xr_s=x}$ for all $t\in[0,T]$ and $x\in\Xspace$.

\subsection{Inclusion of jumps at deterministic times}
\label{sec:deltas}

For modeling purposes, it is sometimes desirable to insist that certain events occur at specified times.
For example, if samples are collected at specific times in such a way that the timing itself conveys no information about the process, one might wish to condition on the sampling time.
We can expand the class of population models to allow for this as follows.
Suppose that $S=\Set{s_1,s_2,\dots,}\subset\Rp$ is a sequence of times.
Let us postulate that, at each of these times, an event occurs at which $\Xr_t$ jumps according to a given probability kernel $\pi$.
In particular, for any state $x\in\Xspace$ and measurable $\mathcal{E}\subset\Xspace$,
$\pi(s_i,x,\mathcal{E})$ is the probability that the jump at time $s_i$ is to $\mathcal{E}$, conditional on $\Xrt_{s_i}=x$.
With this notation, the KFE for the process becomes
\begin{align}
  \label{eq:kfe-reg}
  \frac{\partial{v}}{\partial{t}}(t,x)
  &=\int\!{v(t,x')\,\alpha(t,x',x)\,\dd{x'}}
  -\int\!{v(t,x)\,\alpha(t,x,x')\,\dd{x'}},
  &t\notin{S},\\
  \label{eq:kfe-sing}
  v(t,x)\,\dd{x}
  &=\int\!{\leftlim{v}(t,x')\,\pi(t,x',\dd{x})\,\dd{x'}},
  &t\in{S}.
\end{align}
Note that \zcref{eq:kfe-reg} is identical to \zcref{eq:kfe};
we call this the \emph{regular part} of the KFE.
We refer to \zcref{eq:kfe-sing} as the \emph{singular part} of the KFE.
In this notation, $\pi(t,x',\dd{x})/\dd{x}$ is the density (\ie Radon-Nikodym derivative) of $\pi$ with respect to the base measure on $\Xspace$.

As a matter of notation, one can represent \zcref{eq:kfe-reg,eq:kfe-sing} as a single equation in the form of \zcref{eq:kfe}.
In particular, if in \zcref{eq:kfe} we make the substitution
\begin{equation*}
  \alpha(t,x,x')\,\dd{x'}\mapsto\alpha(t,x,x')\,\dd{x'}+\sum_{s\in{S}}{\delta(s,t)\,\pi(t,x,\dd{x'})},
\end{equation*}
we obtain an equation which we can view as shorthand for \zcref{eq:kfe-reg,eq:kfe-sing}.
Here, $\delta(s,t)$ is the one-sided Dirac delta function satisfying $\delta(s,t)=0$ for $s\ne{t}$ and
\begin{equation*}
  \int_a^b{f(t)\,\delta(s,t)\,\dd{t}}=f(s)\,\Indicator{s\in\halfopen{a,b}},
\end{equation*}
whenever $f$ is \cadlag\ and $-\infty\le{a}<{b}\le{\infty}$.
%% To be even more parsimonious, we can refer to \zcref{eq:kfe-reg,eq:kfe-sing} as the $(\alpha,\pi,S)$ KFE.

\subsection{Jump marks}
\label{sec:jump-marks}

\begin{figure}
  \begin{center}
  \resizebox{0.9\linewidth}{!}{
    \begin{tikzpicture}[scale=1]
      \usetikzlibrary{shapes,arrows,positioning}
      \tikzstyle{coordinate}=[inner sep=0pt,outer sep=0pt]
      \tikzstyle{state}=[shape=ellipse, color=black, draw, font=\LARGE, fill=white, thick, minimum height=8em]
      \tikzstyle{trans}=[color=black, font=\Large, thick, >=stealth]
      \node[state] (base) at (0,0) {$x=(S,E,I,R)$};
      \node[state] (trans) at (165:8) {$x'=(S-1,E+1,I,R)$};
      \node[state] (prog) at (205:7) {$x'=(S,E-1,I+1,R)$};
      \node[state] (recov) at (305:5) {$x'=(S,E,I-1,R+1)$};
      \node[state] (wane) at (350:8) {$x'=(S+1,E,I,R-1)$};
      \node[state] (sample) at (18:7.5) {$x'=(S,E,I,R)$};
      \draw[trans,->] (base) -- (trans) node[midway,above,sloped] {$\lab{Trans}$};
      \draw[trans,->] (base) -- (prog) node[midway,above,sloped] {$\lab{Prog}$};
      \draw[trans,->] (base) -- (recov) node[midway,above,sloped] {$\lab{Recov}$};
      \draw[trans,->] (base) -- (wane) node[midway,above,sloped] {$\lab{Wane}$};
      \draw[trans,->] (base) -- (sample) node[midway,above,sloped] {$\lab{Sample}$};
      \node[font=\LARGE] (U) at (95:3.5) {$\Jumps=\Set{\lab{Trans},\lab{Prog},\lab{Recov},\lab{Wane},\lab{Sample}}$};
    \end{tikzpicture}
  }
\end{center}
  \caption{
    \label{fig:markov-state}
    Markov state transition diagram for the SEIRS model depicted in \zcref{fig:example-models}A.
    The state, $x$, is characterized by four numbers, $S$, $E$, $I$, and $R$.
    From a given state $x$, there are five possible kinds of jumps $x\mapsto{x'}$.
    Accordingly, the set, $\Jumps$, of jump marks has five elements.
    In the terminology of \zcref{sec:event-types}, each of these is of a different type:
    $\lab{Trans}$ (transmission) is of birth type,
    $\lab{Prog}$ (progression) is of migration type,
    $\lab{Recov}$ (recovery) is of death type,
    $\lab{Sample}$ (sampling) is of sample type,
    and $\lab{Wane}$ (loss or waning of immunity) is of neutral type.
    Note that, in this formulation, when a sampling event occurs, the state does not change.
  }
\end{figure}

It will be useful to divide the jumps of the population process $\Xr_t$ into distinct categories which differ with respect to the changes they induce in a genealogy.
For this purpose, we let $\Jumps$ be a countable set of jump \emph{marks} such that
\begin{equation*}
  \alpha(t,x,x')=\sum_{u\in\Jumps}{\alpha_u(t,x,x')}.
\end{equation*}
\zcref[S]{fig:markov-state} shows an example for which $\Jumps$ has five elements.
In the following, sums over $u$ are to be taken over the whole of $\Jumps$ unless otherwise indicated.

Let us define the \emph{jump mark process}, $\Ur_t$, to be the mark of the latest jump as of time $t$.
As usual, we take the sample paths of $\Ur_t$ to be \cadlag.
Observe that $\Ur_t$ is not Markov, though $(\Xr_t,\Ur_t)$ is.
In particular,
\begin{equation*}
  \frac{\partial{v}}{\partial{t}}(t,x,u)
  =\sum_{u'}{\int\!{v(t,x',u')\,\alpha_{u}(t,x',x)\,\dd{x'}}}
  -\sum_{u'}{\int\!{v(t,x,u)\,\alpha_{u'}(t,x,x')\,\dd{x'}}}
\end{equation*}
is the KFE for the $(\Xr_t,\Ur_t)$ process.

\subsection{Demes and deme occupancy}
\label{sec:demes}

Our first goal in this paper is to show how a given population process induces a stochastic process on the space of genealogies.
At each time, this genealogy will represent the relationships of shared ancestry among a population of lineages extant at that time.
To accommodate the structure of the population, this population of lineages will itself be subdivided into discrete categories.
In particular, we suppose that there are a countable set of subpopulations, within each of which individual lineages are exchangeable.
We call these subpopulations \emph{demes} and let $\Demes$ be an index set for them.
\zcref[S]{fig:example-models} illustrates this concept in the context of several specific models.
We note that other authors have used different terminology for the same concept, including ``colony'', ``type'', ``state'', and ``population'' \citep{Takahata1988,Notohara1990,Boskova2014,Volz2018,BaridoSottani2020,Seidel2024,Vaughan2025}.

We define the \emph{deme occupancy} function $n:\Demes\times\Xspace\to\Zp$ so that
for $i\in\Demes$, $x\in\Xspace$, $n_i(x)$ is the number of lineages in deme $i$ when the population is in state $x$.

\subsection{Examples}
\label{sec:examples}

The class of population models to which the theory presented here applies is very broad.
In particular, it encompasses the entire class of compartmental models with time-dependent hazards.
Here, to give a sense of this breadth, we briefly describe a few models of interest.

\paragraph{SIRS model}

\citet{King2022} developed formulae for the exact likelihood of a genealogy induced by an SIRS model.
The theory developed in this paper applies, but since there is only one deme in this model, it is a simple case.
In \zcref{sec:sirs-example}, we illustrate the theory by working it out in detail for this model.

\paragraph{SEIRS model}

Modifying the SIRS model by incorporating an incubation period results in the SEIRS model, which has two demes (\zcref{fig:example-models}A):
$\Demes=\Set{\deme{E},\deme{I}}$.
We can take the state space to be $\Zp^4$:
the state $x=(S,E,I,R)$ is defined by the numbers of hosts in each of the four compartments.
The deme occupancy function in this case is $n(x)=(E,I)$.
The genealogies depicted in \zcref{fig:geneal,fig:upo} pertain to this model and, in \zcref{sec:seirs-example}, we work out the theory for the SEIRS model in detail and display results of some numerical computations.

\paragraph{Two-strain competition model}

A simple model for the competition of two pathogen strains is depicted in \zcref{fig:example-models}B.
In this model, the state vector consists of seven numbers: $x=(S,E_1,E_2,I_1,I_2,R_1,R_2)$.
There are four demes ($\Demes=\Set{\deme{E}_1,\deme{E}_2,\deme{I}_1,\deme{I}_2}$) and the occupancy function is $n(x)=(E_1,E_2,I_1,I_2)$.

\paragraph{Complex infection progression}

\zcref[S]{fig:example-models}C depicts a model in which infections are heterogeneous with respect to the severity of the disease they engender.
There are three demes ($\Demes=\Set{\deme{E},\deme{I_A},\deme{I_S}}$).
Asymptomatic infections ($\deme{I_A}$) are unobservable, while symptomatic infections ($\deme{I_S}$) may develop into hospitalized cases ($\deme{H}$) and deaths ($\deme{D}$).

\paragraph{Superspreading model}

Compartmental models can be used to capture heterogeneities of various kinds.
\zcref[S]{fig:example-models}D depicts a model with behavioral heterogeneity.
There are three demes ($\Demes=\Set{\deme{E},\deme{I_L},\deme{I_H}}$).
Hosts in the low-transmitting $\deme{I_L}$ compartment have fewer contacts than those in the high-transmitting $\deme{I_H}$ compartment, and individuals move between these compartments as they engage in bouts of risky or cautious behavior.

\paragraph{Kingman coalescent and Moran model}

The \citet{Kingman1982a} coalescent is one of the two foundations upon which most existing phylodynamic approaches have been constructed.
It is the ancestral process for the Moran model, in which a fixed population of $n$ lineages experiences events at times distributed according to a rate-$\mu$ Poisson process.
At each such event, an individual lineage selected uniformly at random dies and is replaced by the offspring of a second randomly selected lineage.
In \zcref{sec:kingman}, we show that the Kingman coalescent likelihood is a special case of the theory presented in this paper.

\paragraph{Linear birth-death model}

The linear birth-death process, the second basis for widely-used phylodynamic methods, is also a special case of the theory presented here.
For this process, we have $\Xspace=\Zp$ and there is a single deme.
$\Xr_t$ represents the size of a population and $n(x)=x$.
In \zcref{sec:lbdp}, we derive exact expressions for the likelihood under this model.

\subsection{History process}
\label{sec:history-process}

Consider the Markov process $(\Xr_t,\Ur_t)$.
We define its \emph{history process}, $\Hr_t$, to be the restriction of the random function $s\mapsto(\Xr_s,\Ur_s)$ to the interval $[0,t]$.
Note that $\Hr_t$ is itself trivially a Markov process, since it contains its own history.
Alternatively, one can identify $\Hr_t$ with the sequence
$\left(\left(\Trh_k,\Xrh_k,\Urh_k\right)\right)_{k=0}^{\Kr_t}$.
In particular, conditional on $\Hr_t$, both $\Xr_t$ and $\Ur_t$ are deterministic, as are $\Kr_t$, the embedded chains, $\Xrh_k$, $\Urh_k$, and the point process of event times $\Trh_k$.
The probability measure on the space of histories can be expressed in terms of these:
\begin{mathsize}{9pt}{10pt}
  \begin{equation}
    \label{eq:Hdens}
    \Prob{\dd{\H_t}}
    =p_{0}(\Xh_{0})\,\dd{\Xh_0}\,
    \prod_{k=1}^{K_{t}}{\alpha_{\Uh_k}\!\!\left(\Th_k,\Xh_{k-1},\Xh_{k}\right)\,\dd{\Xh_k}\,\dd{\Th_k}}
    \,\exp{\left(-\sum_{k=0}^{K_t}{\int_{\Th_{k}}^{\Th_{k+1}}{\sum_{u}{\int{\alpha_{u}(t',\Xh_{k},x')\,\dd{x'}}}}\,\dd{t'}}\right)},
  \end{equation}
\end{mathsize}%
where again, by convention, $\Th^{}_0=0$ and $\Th^{}_{K^{}_t+1}=t$.

If $\H$ is a potential value of a history process, $\Hr_t$, we define $\time(\H)$ to be the right endpoint of its domain and use the notation $\event{\H}\coloneq\Set{\Th^{}_1,\dots,\Th^{}_{\K_t}}\subset{[0,\time(\H)]}$ to denote the set of its jump times.

\section{Genealogy processes}

\subsection{Genealogies}
\label{sec:genealogy}

A \emph{genealogy}, $G$, encapsulates the relationships of shared ancestry among a set of lineages that are extant at some time $\time(G)\in\Rp$ and perhaps a set of samples collected at earlier times (\zcref{fig:geneal}).
A genealogy has a tree- or forest-like structure, with four distinct kinds of nodes:
\begin{inparaenum}[(i)]
\item \emph{tip nodes}, which represent labeled extant lineages;
\item \emph{internal nodes}, which represent events at which lineages diverged and/or moved from one deme to another;
\item \emph{sample nodes}, which represent labeled samples; and
\item \emph{root nodes}, at the base of each tree.
\end{inparaenum}
Each node $a$ is associated with a specific time, $\time(a)$.
In particular, if $a$ is a tip node in $G$, then $\time(a)=\time(G)$;
if $a$ is a sample node, then $\time(a)\le{\time(G)}$ is the time at which the sample was taken.
Moreover, if node $a$ is ancestral to node $a'$, then $\time(a)\le{\time(a')}$ and $\time(a')-\time(a)$ is the distance between $a$ and $a'$ along the genealogy (\ie the patristic distance).
Without loss of generality we assume that $\time(a)=0$ for all root nodes $a$.
We let $\event{G}$ denote the set of all internal and sample node-times of the genealogy $G$;
we refer to these as \emph{genealogical event times}.

Importantly, a genealogy informs us not only about the shared ancestry of any pair of lineages, but also about where in the set of demes any given lineage was at all times.
Accordingly, we can visualize a genealogy as a graph, the nodes and edges of which are painted with a distinct color for each deme (\zcref{fig:geneal}).
Note that a genealogy will in general have \emph{branch-point nodes}, \ie internal nodes with more than one descendant, but may also have internal nodes with only one descendant.
We refer to such nodes as \emph{inline nodes}.
These occur whenever the color changes along a branch, but can also occur without a color-change.

\begin{figure}[t]
  \begin{center}
\begin{knitrout}\small
\definecolor{shadecolor}{rgb}{0.969, 0.969, 0.969}\color{fgcolor}

{\centering \includegraphics[width=0.5\linewidth]{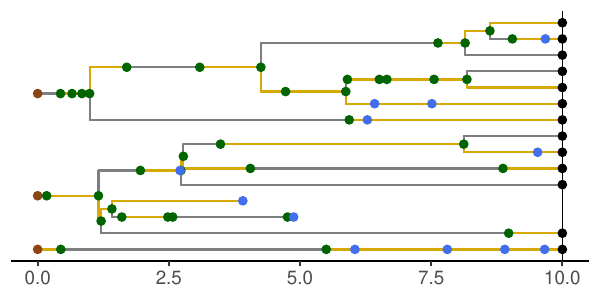} }

\end{knitrout}
  \end{center}
  \caption{
    A genealogy, $G$, specifies the relationships of shared ancestry (via its tree-structure) and deme occupancy histories (via the coloring of its branches) of a set of lineages extant at some time $\time(G)$, as well as some samples gathered at earlier times.
    Here, the genealogy is drawn from the genealogy process for the SEIR model.
    We have $\time(G)=10$ and there are two demes, $\Demes=\Set{\deme{E},\deme{I}}$ shown in grey and yellow, respectively.
    Tip nodes, denoting extant lineages, are shown as black dots;
    sample nodes are shown as blue dots;
    internal nodes are indicated in green.
    Note that internal nodes occur not only at branch-points, but also inline (\ie along branches).
    Wherever a lineage moves from one deme (color) to another, an internal node occurs;
    the converse does not necessarily hold.
    \label{fig:geneal}
  }
\end{figure}

Formally, we define a genealogy by extending the original construction of \citet{Kingman1982a}.
Specifically, we define a genealogy $G$, to be a triple, $(T,Z,Y)$, where $T=\time(G)\in\Rp$ is the \emph{genealogy time}, $Z$ specifies the genealogy's \emph{tree structure}, and $Y$ gives the \emph{coloring}.
In particular, let $\leaves$ be a countable set of labels---referring to samples and/or extant lineages---and let $\partit(\leaves)$ be the collection of all partitions of subsets of $\leaves$.
That is, an element $z\in{\partit(\leaves)}$ is a collection of nonempty, mutually disjoint subsets of $\leaves$.
Partition \emph{fineness} defines a partial order on $\partit(\leaves)$.
Specifically, for $z,z'\in{\partit(\leaves)}$, we say $z\preceq{z'}$ if and only if for every $b'\in{z'}$ there is $b\in{z}$ such that $b\supseteq{b'}$.
The tree structure of $G$ is a \cadlag\ map $Z:[0,T]\to\partit(\leaves)$ that is monotone in the sense that $t_1\le{t_2}$ implies $Z_{t_1}\preceq{Z_{t_2}}$.
An element $b\in{Z_t}$ is a set of labels;
it represents the branch of the tree that bears the corresponding lineages.
We use the notation $\event{Z}$ to denote the set of times at which $Z$ is discontinuous.
Note that $\event{Z}$ includes the times of all tip, sample, and branch-point nodes, but excludes root and non-sample inline nodes.
Therefore, $\event{Z}\subseteq{\event{G}}$.

The third element of $G$ specifies the coloring of branches and locations of tip, sample, and internal nodes (including inline nodes).
Mathematically, if $G=(T,Z,Y)$, then $Y$ is a function that maps each point on the genealogy to a deme and an indicator.
In particular, if $t\in[0,T]$ and $a$ is the label of any tip or sample node,
$Y_t(a)=(\col{Y_t(a)},\ctr{Y_t(a)})\in\ColorSet$, where $\col{Y_t(a)}$ is the deme in which the lineage of $a$ is located at time $t$ and $\ctr{Y_t(a)}=1$ if there is a node at $t$ and $0$ otherwise.
We call $\col{Y_t(a)}$ the \emph{color} of the branch bearing $a$ at $t$ and $\ctr{Y_t(a)}$, the \emph{event indicator}.

Since $a,a'\in{b}\in{Z_t}$ implies $Y_t(a)=Y_t(a')$, one can equally well think of $Y_t$ as a map $Z_t\to\ColorSet$ and we will find it useful to think of it in this way most of the time.
Likewise, in an abuse of notation that risks little confusion, one sees that, given $Y$, $\col{Y}:Z\to\Demes$ and $\ctr{Y}:Z\to\Set{0,1}$ are well defined in the obvious way.
In particular, $\ctr{Y}\circ{Z}$ has finite support and we take $\col{Y}\circ{Z}$ to be \cadlag.

Given a tree $Z$, we let $\fY(Z)$ denote the set of colorings $Y$ that are compatible with $Z$.
We moreover define $\fY_t(Z)\coloneq\CondSet{Y_t}{Y\in{\fY(Z)}}$.
Formally speaking, $\fY(Z)$ is a fiber bundle over $Z$, each $\fY_t(Z)$ being a fiber.
In particular the canonical map $Y_t\mapsto{Z_t}:\fY(Z)\to{Z}$ always exists and is unique.
Finally, we will sometimes use the notation $G^{\lab{Z}}$ to refer to the tree-structure portion of $G$, and $G^{\lab{Y}}$ to refer to the coloring, so that $G=(\time(G),G^{\lab{Z}},G^{\lab{Y}})$.

\subsection{Event types}
\label{sec:event-types}

To see how events in the population process leave a trace in the genealogy, we assume that, at each jump in the population process, a corresponding change occurs in the genealogy, according to whether lineages branch, die, move between demes, or are sampled.
For this purpose, there are five distinct \emph{pure types} of events:
\begin{compactenum}[(a)]
\item \emph{Birth-type events} result in the branching of one or more new lineages from an existing lineage.
  Examples of birth-type events include transmission events, speciations, and actual births.
  Importantly, we assume that all new lineages arising from a birth event share the same parent and that at most one birth event occurs at a time, almost surely.
  While it is possible to relax this assumption, the resulting notational complexities are nontrivial, so we postpone consideration of this more general case.
\item \emph{Death-type events} result in the extinction of one or more lineages.
  Examples include recovery from infection, death of a host, and species extinctions.
  We allow for the possibility that multiple lineages, potentially in multiple demes, die simultaneously.
\item \emph{Migration-type events} result in the movement of a lineage from one deme to another.
  Spatial movements, changes in host age or behavior, and progression of an infection can all be represented as migration-type events.
  In a migration-type event, one or more lineages within exactly one deme may move simultaneously, and they may move to different demes.
  Again, while it possible to relax the assumption that exactly one deme is the source of all migrating lineages within any one event, we postpone consideration of this case for the present.
\item \emph{Sample-type events} result in the collection of a sample from a lineage.
  We allow for the possibility that multiple samples are collected simultaneously, though we require that, in this case, each extant lineage is sampled at most once and that all sampled lineages be in the same deme.
  It is possible to allow for simultaneous sampling of multiple demes but, again, we postpone consideration of this case to ease the exposition.
\item \emph{Neutral-type events} result in no change to any of the lineages.
\end{compactenum}
\zcref[S]{fig:markov-state} depicts an example with jumps of all five pure types.
It is not necessary that an event be of a pure type;
\emph{compound events} partake of more than one type.
For example, a sample/death-type event, in which a lineage is simultaneously sampled and removed, has been employed \citep{Leventhal2014}, as have birth/death events in which one lineage reproduces at the same moment that another dies (\eg the \citet{Moran1958} process).
%% There are restrictions on the parentage within an event.
%% In particular, in Theorem 1, we assume that, given $u$, the parent deme is deterministic.
In combining the pure types, the main restriction is that, at any compound event, all parent lineages must be drawn from the same deme.
Beyond this, the theory presented here places few restrictions on the complexity of the events that can occur by combining events of the various pure types.

\subsection{Genealogy process}
\label{sec:genealogy-process}

We now show how a given population process induces a stochastic process, $\Gr_t$, on the space of genealogies.
In the case of unstructured population processes (\ie those having a single deme), \citet{King2022} gave a related construction that is equivalent to the one presented here.

\begin{figure}
  \input{event_types}
  \caption{
    Event types differ by their effects on the genealogy.
    This can be seen by examining the local structure of the genealogy in the neighborhood of a jump.
    \textbf{(A)} A birth-type jump results in the branching of one or more child lineages from the parent.
    There can be only one parent, though the demes of the child lineages may differ from that of their parent.
    Here, a parent of the grey deme sires one child lineage in each of the grey and yellow demes.
    The \emph{production} of an event is an integer vector, with one entry for each deme.
    The production of this event is therefore $r=(r_{\deme{grey}},r_{\deme{yellow}})=(2,1)$.
    The \emph{deme occupancy} of an event is the number of lineages in each deme just to the right of the event.
    The deme occupancy at this event is therefore $n=(n_{\deme{grey}},n_{\deme{yellow}})=(3,5)$.
    \textbf{(B)} A death-type event causes the extinction of a lineage.
    Since internal nodes without children are recursively removed, the affected branch is dropped.
    The production of this event is $r=(0,0)$ and the deme occupancy is $n=(3,4)$.
    \textbf{(C)} A migration-type event results in the movement of one or more lineages from one deme to another.
    Here, one lineage moves from the yellow to the grey deme.
    The production of this event is $r=(1,0)$, \ie the production is 1 for the grey deme and 0 for the yellow.
    The deme occupancy is $n=(6,2)$.
    \textbf{(D)} In a sample-type event, one or more sample nodes (blue circles) are inserted.
    Here, there are two samples, one in each of the grey and yellow demes.
    Accordingly, $r=(1,1)$ and $n=(2,6)$.
    \textbf{(E)} A neutral-type event has no effect on the genealogy and zero production in all demes: $r=(0,0)$, $n=(5,3)$.
    \textbf{(F)} The theory presented here allows for compound events.
    As an example, here a birth/death-type event occurs, wherein one yellow lineage is extinguished and a grey lineage simultaneously sires a grey child.
    For this event, we have $r=(2,0)$ and $n=(6,2)$.
    \textbf{(G)} Here, a compound sample/death-type event with $r=(0,0)$ and $n=(2,5)$ occurs.
    A grey lineage is sampled and simultaneously extinguished.
    Note that recursive removal does not occur, since sample nodes are never removed.
    \textbf{(H)} A compound birth/migration-type event with $r=(4,0)$ and $n=(6,2)$.
    \label{fig:event-types}
  }
\end{figure}

At each jump in the population process, a change is made to the genealogy, according to the mark, $u$, of the jump (\zcref{fig:event-types}).
In particular:
\begin{compactenum}[(a)]
\item
  If $u$ is of birth-type (\zcref{fig:event-types}A), it results in the creation of one new internal node, call it $b$.
  A tip node, $a$, of the appropriate deme is chosen with uniform probability from among those present and $b$ is inserted so that its ancestor is that of $a$, while $a$ takes $b$ as its ancestor.
  One new tip node, of the appropriate deme, is created for each of the children, all of which take $b$ as their immediate ancestor.
\item
  If $u$ is of death-type (\zcref{fig:event-types}B), one or more tip nodes of the appropriate demes are selected with uniform probability from among those present.
  These are deleted.
  Next, non-sample internal nodes without children are recursively removed.
  Sample nodes are never removed.
\item
  At a migration-type event (\zcref{fig:event-types}C), the appropriate number of migrating lineages are selected at random with uniform probability, from among those present in the appropriate deme.
  For each selected lineage, one new branch node is inserted between the selected tip node and its ancestor.
  The color of the descendant branch changes accordingly.
\item
  At a sample-type event (\zcref{fig:event-types}D), the appropriate number of sampled lineages are selected at random from among the tip nodes, with uniform probability according to deme.
  One new sample node is introduced for each selected lineage:
  each is inserted between a selected tip node and its ancestor.
\item
  At a neutral-type event (\zcref{fig:event-types}E), no change is made to the genealogy.
\item
  Finally, events of compound type (\eg \zcref{fig:event-types}F--H) are accommodated by combining the foregoing rules.
\end{compactenum}
In each of these events, the new node or nodes that are introduced have node-times equal to the time of the jump.

\subsubsection{Emergent lineages and production}
\label{sec:production}

At each event, the lineages which descend from a node inserted at that event are said to \emph{emerge} from the event.
Thus, after a birth-type event, the emerging lineages include all the new offspring as well as the parent.
Likewise, at pure migration- or sample-type events, each migrating or sampled lineage emerges from the event.
At pure death-type events, no lineages emerge.
In general, at an event of mark $u$, there are $r^u_i$ emergent lineages in deme $i$.
We require that $r^u_i$ be a constant, for each $u$ and $i$.
Thus there is a function $r:\Jumps\times\Demes\to\Zp$, such that $r^u_i$ lineages of deme $i$ emerge from each event of mark $u$.
Since, in applications, one is free to expand the set of jump-marks $\Jumps$ as needed, this is not a restriction on the models that the theory can accommodate.
We say $r^u\coloneq{(r^u_i)_{i\in\Demes}}$ is the \emph{production} of an event of mark $u$.
Note that the lineages that die as a result of an event do not count in the production but that a parent lineage that survives the event does count.

\subsubsection{Conditional independence and exchangeability}

Application of these rules at each jump of $\Xr_t$ constructs a chain of genealogies $\Grh_k$.
In particular, at each jump-time $\Trh_k$, the genealogy $\Grh_{k-1}$ is modified according to the jump-mark $\Urh_k$ to yield $\Grh_k$.
We view $\Grh_k$, $k=0,1,2,\dots$, as the embedded chain of the continuous-time genealogy process $\Gr_t$.
It is very important to note that, conditional on $(\Xrh_k,\Urh_k)$, the number of parents and number of offspring in each deme is determined and the random choice of which lineages die, migrate, are sampled, or sire offspring is independent of these choices at any other times and independent of $(\Xrh_j,\Urh_j)$ for all $j\ne{k}$.
Moreover, by assumption, the lineages within each deme are exchangeable:
any lineage within a deme is as likely as any other lineage in that deme to be selected as a parent or for death, sampling, or migration.
Finally, note that $\Gr_t$ does not have the Markov property, though $(\Xr_t,\Ur_t,\Gr_t)$ and $(\Xr_t,\Gr_t)$ do.
If, instead of dropping tip nodes at death events we were to retain them as we do samples, the resulting genealogy---which we can call the \emph{complete genealogy}---would have the Markov property.

\subsection{Pruned and obscured genealogies}

\begin{figure}
\begin{knitrout}\small
\definecolor{shadecolor}{rgb}{0.969, 0.969, 0.969}\color{fgcolor}

{\centering \includegraphics[width=0.5\linewidth]{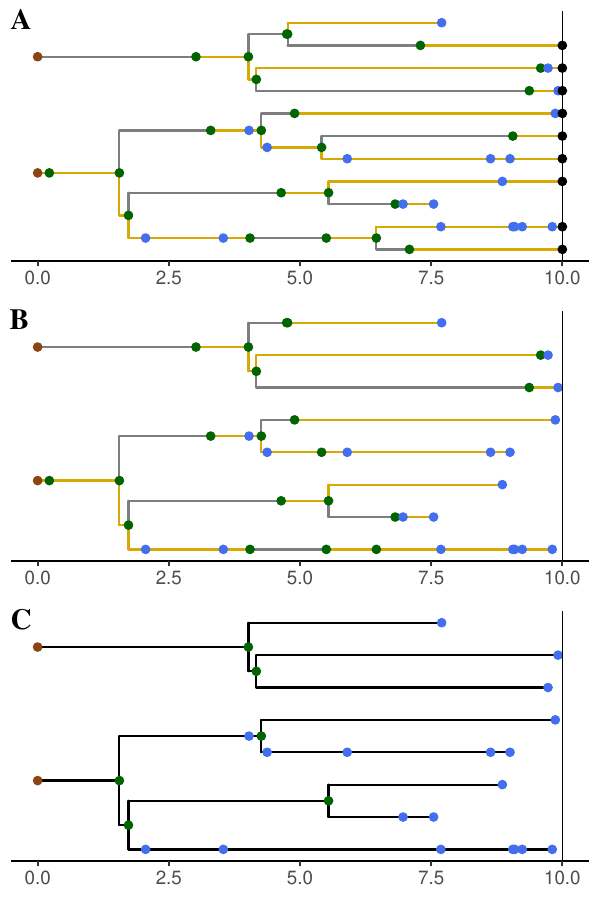} }

\end{knitrout}
  \caption{
    \label{fig:upo}
    Unpruned, pruned, and obscured genealogies from a single realization of the genealogy process induced by the SEIRS model depicted in \zcref{fig:example-models,fig:markov-state}.
    \textbf{(A)} A realization of the unpruned genealogy process $\Gr_t$ is shown at $t=10$.
    Tip nodes, corresponding to lineages alive at time $t=10$ are indicated with black points.
    Blue points represent samples;
    green points, internal nodes.
    Branches are colored according to the deme in which the corresponding lineage resided at that point in time:
    grey denotes $\deme{E}$ and yellow, $\deme{I}$.
    \textbf{(B)} The genealogy is \emph{pruned} by deleting all tip nodes and then recursively pruning away childless internal nodes.
    Sample nodes are never removed.
    \textbf{(C)} A pruned genealogy is \emph{obscured} by effacing all deme information from lineage histories:
    the colors are erased, as are all non-sample inline nodes.
    See the text (\zcref{sec:genealogy,sec:pruning,sec:obscuration}) for more detail.
  }
\end{figure}

The process just described yields a genealogy that relates all extant members of the population, and all samples.
Moreover, it details each lineage's complete history of movement through the various demes.
However, the data we ultimately wish to analyze will be based only on samples.
Nor, in general, will the histories of deme occupancy be observable.
A generative model must account for this loss of information.
We therefore now describe how genealogies are \emph{pruned} to yield sample-only genealogies and then \emph{obscured} via the erasure of color from their branches (\zcref{fig:upo}).

\subsubsection{Pruned genealogy}
\label{sec:pruning}

Given a genealogy $G$, one obtains the \emph{pruned genealogy}, $P=\prune(G)$ by first dropping every tip node and then recursively dropping every childless internal node (\zcref{fig:upo}A--B).
In a pruned genealogy only internal and sample nodes remain, and sample nodes are found at all of the leaves and possibly some of the interior nodes of the genealogy.
Observe that a pruned genealogy is a colored genealogy:
it retains information about where among the demes each of its lineages was through time (\zcref{fig:upo}B).
Note also that a pruned genealogy $P$ is a genealogy, characterized by its time, $\time(P)$ and the functions $P^{\lab{Y}}$ and $P^{\lab{Z}}$ just as an unpruned genealogy is.
Crucially, since it contains within itself all of its past history, the pruned genealogy process $\Pr_t=\prune(\Gr_t)$ is Markov.

\subsubsection{Lineage count and saturation}
\label{sec:ells}

In the following, we will find that we need to count the deme-specific numbers of lineages present in a given pruned genealogy at a given time.
Accordingly, suppose $P=(T,Z,Y)$ is a pruned genealogy and suppose $t\in[0,T]$.
Let $\ell_i$ denote the number of lineages in deme $i$ at time $t$ and $\ell\coloneq(\ell_i)_{i\in\Demes}\in\Zp^{\Demes}$.
Clearly, $\ell$ depends only on $\Y_t$.
Therefore, we can define $\ell$ as a function such that, whenever $P=(T,Z,Y)$ is a pruned genealogy, $\ell(\Y_t)$ is the vector of deme-specific lineage counts at time $t$.
We refer to $\ell$ as the \emph{lineage-count} function (cf.~\zcref{fig:ells}).

We will also have occasion to refer to the deme-specific number of lineages emerging from a given event.
In particular, given a node time $t$ in a pruned genealogy $P=(T,Z,Y)$, the number $s_i$ of lineages of deme $i$ emerging from all nodes with time $t$ is well defined and we can write $s\coloneq\left(s_i\right)_{i\in\Demes}$.
Like the lineage-count function, $s$ depends only on $Y_t$.
Thus, we can define the \emph{saturation} function such that, whenever $P=(T,Z,Y)$ is a pruned genealogy, $s(\Y_t)$ is the integer vector of deme-specific numbers of emerging lineages at time $t$.
\zcref[S]{fig:ells} illustrates.

\begin{figure}
  \begin{center}
  \resizebox{0.5\linewidth}{!}{
    \begin{tikzpicture}[scale=2]
      \usetikzlibrary{shapes,arrows,positioning}
      \usetikzlibrary{arrows.meta,positioning,decorations}
      \tikzstyle{label}=[font=\Large]
      \tikzstyle{coordinate}=[inner sep=0pt,outer sep=0pt]
      \tikzstyle{deme1}=[color=grey, line width=3pt]
      \tikzstyle{deme1H}=[color=grey, line width=3pt, densely dotted]
      \tikzstyle{deme2}=[color=maize, line width=3pt]
      \tikzstyle{deme2H}=[color=maize, line width=3pt, densely dotted]
      \tikzstyle{axis}=[color=black, thick, >=stealth]
      \tikzstyle{timepoint}=[color=red!70!black]

      \def\ysp{0.2} %% spacing between adjacent lineages

      %% Birth-type event, unpruned
      \begin{scope}[shift={(0,0)}]
        \coordinate (origin) at (0,0);
        \node[label] (lab) at (origin) {\textbf{A}};
        \coordinate (left) at ($(origin)+(0.1,-0.1)$);
        \coordinate (right) at ($(left)+(1,0)$);
        \coordinate (mid) at ($(left)!0.5!(right)$);
        \coordinate (rlab) at ($(right)+(0.1,0)$);
        \coordinate (bleft) at ($(left)+(0,-8*\ysp)$);
        \coordinate (bright) at (rlab |- bleft);
        \coordinate (temp) at ($(origin)+(0,-3*\ysp)$);
        \coordinate (temp1) at ($(temp)+(0,-\ysp)$);
        \draw[deme2] (mid |- temp) -- (mid |- temp1) -- (right |- temp1);
        \draw[deme1] (left |- temp) -- (right |- temp);
        \node[text=darkgreen,font=\Huge] at (mid |- temp) {$\bullet$};
        \foreach \y in {1,2,7} {
          \coordinate (temp) at ($(origin)+(0,-\y*\ysp)$);
          \draw[deme1] (left |- temp) -- (right |- temp);
        }
        \foreach \y in {5,6,8} {
          \coordinate (temp) at ($(origin)+(0,-\y*\ysp)$);
          \draw[deme2] (left |- temp) -- (right |- temp);
        }
        \draw[timepoint] (mid) -- (mid |- bleft);
        \draw[axis,->] (bleft) -- (bright);
        \node[xshift=1ex] (tlab) at (bright.east) {$t$};
      \end{scope}

      %% Birth-type event, pruned
      \begin{scope}[shift={(0,-2)}]
        \coordinate (origin) at (0,0);
        \node[label] (lab) at (origin) {\textbf{B}};
        \coordinate (left) at ($(origin)+(0.1,-0.1)$);
        \coordinate (right) at ($(left)+(1,0)$);
        \coordinate (mid) at ($(left)!0.5!(right)$);
        \coordinate (rlab) at ($(right)+(0.1,0)$);
        \coordinate (bleft) at ($(left)+(0,-8*\ysp)$);
        \coordinate (bright) at (rlab |- bleft);
        \coordinate (temp) at ($(origin)+(0,-3*\ysp)$);
        \coordinate (temp1) at ($(temp)+(0,-\ysp)$);
        \draw[deme1] (left |- temp) -- (mid |- temp);
        \draw[deme1H] (mid |- temp) -- (right |- temp);
        \draw[deme2] (mid |- temp) -- (mid |- temp1) -- (right |- temp1);
        \node[text=darkgreen,font=\Huge] at (mid |- temp) {$\bullet$};
        \foreach \y in {7} {
          \coordinate (temp) at ($(origin)+(0,-\y*\ysp)$);
          \draw[deme1H] (left |- temp) -- (right |- temp);
        }
        \foreach \y in {1,2} {
          \coordinate (temp) at ($(origin)+(0,-\y*\ysp)$);
          \draw[deme1] (left |- temp) -- (right |- temp);
        }
        \foreach \y in {5,6} {
          \coordinate (temp) at ($(origin)+(0,-\y*\ysp)$);
          \draw[deme2H] (left |- temp) -- (right |- temp);
        }
        \foreach \y in {8} {
          \coordinate (temp) at ($(origin)+(0,-\y*\ysp)$);
          \draw[deme2] (left |- temp) -- (right |- temp);
        }
        \draw[timepoint] (mid) -- (mid |- bleft);
        \draw[axis,->] (bleft) -- (bright);
        \node[xshift=1ex] (tlab) at (bright.east) {$t$};
      \end{scope}

      %% Migration-type event, unpruned
      \begin{scope}[shift={(2,0)}]
        \coordinate (origin) at (0,0);
        \node[label] (lab) at (origin) {\textbf{C}};
        \coordinate (left) at ($(origin)+(0.1,-0.1)$);
        \coordinate (right) at ($(left)+(1,0)$);
        \coordinate (mid) at ($(left)!0.5!(right)$);
        \coordinate (rlab) at ($(right)+(0.1,0)$);
        \coordinate (bleft) at ($(left)+(0,-8*\ysp)$);
        \coordinate (bright) at (rlab |- bleft);
        \foreach \y in {2,3,8} {
          \coordinate (temp) at ($(origin)+(0,-\y*\ysp)$);
          \draw[deme1] (left |- temp) -- (right |- temp);
        }
        \coordinate (temp) at ($(origin)+(0,-4*\ysp)$);
        \draw[deme1] (left |- temp) -- (mid |- temp);
        \draw[deme2] (mid |- temp) -- (right |- temp);
        \node[text=darkgreen,font=\Huge] at (mid |- temp) {$\bullet$};
        \foreach \y in {1,5} {
          \coordinate (temp) at ($(origin)+(0,-\y*\ysp)$);
          \draw[deme2] (left |- temp) -- (right |- temp);
        }
        \foreach \y in {6,7} {
          \coordinate (temp) at ($(origin)+(0,-\y*\ysp)$);
          \draw[deme1] (left |- temp) -- (right |- temp);
        }
        \draw[timepoint] (mid) -- (mid |- bleft);
        \draw[axis,->] (bleft) -- (bright);
        \node[xshift=1ex] (tlab) at (bright.east) {$t$};
      \end{scope}

      %% Migration-type event, pruned
      \begin{scope}[shift={(2,-2)}]
        \coordinate (origin) at (0,0);
        \node[label] (lab) at (origin) {\textbf{D}};
        \coordinate (left) at ($(origin)+(0.1,-0.1)$);
        \coordinate (right) at ($(left)+(1,0)$);
        \coordinate (mid) at ($(left)!0.5!(right)$);
        \coordinate (rlab) at ($(right)+(0.1,0)$);
        \coordinate (bleft) at ($(left)+(0,-8*\ysp)$);
        \coordinate (bright) at (rlab |- bleft);
        \coordinate (temp) at ($(origin)+(0,-4*\ysp)$);
        \draw[deme1] (left |- temp) -- (mid |- temp);
        \draw[deme2] (mid |- temp) -- (right |- temp);
        \node[text=darkgreen,font=\Huge] at (mid |- temp) {$\bullet$};
        \foreach \y in {2,3,8} {
          \coordinate (temp) at ($(origin)+(0,-\y*\ysp)$);
          \draw[deme1H] (left |- temp) -- (right |- temp);
        }
        \foreach \y in {6,7} {
          \coordinate (temp) at ($(origin)+(0,-\y*\ysp)$);
          \draw[deme1] (left |- temp) -- (right |- temp);
        }
        \foreach \y in {1} {
          \coordinate (temp) at ($(origin)+(0,-\y*\ysp)$);
          \draw[deme2H] (left |- temp) -- (right |- temp);
        }
        \foreach \y in {5} {
          \coordinate (temp) at ($(origin)+(0,-\y*\ysp)$);
          \draw[deme2] (left |- temp) -- (right |- temp);
        }
        \draw[timepoint] (mid) -- (mid |- bleft);
        \draw[axis,->] (bleft) -- (bright);
        \node[xshift=1ex] (tlab) at (bright.east) {$t$};
      \end{scope}

      %% Sample event, unpruned
      \begin{scope}[shift={(4,0)}]
        \coordinate (origin) at (0,0);
        \node[label] (lab) at (origin) {\textbf{E}};
        \coordinate (left) at ($(origin)+(0.1,-0.1)$);
        \coordinate (right) at ($(left)+(1,0)$);
        \coordinate (mid) at ($(left)!0.5!(right)$);
        \coordinate (rlab) at ($(right)+(0.1,0)$);
        \coordinate (bleft) at ($(left)+(0,-8*\ysp)$);
        \coordinate (bright) at (rlab |- bleft);
        \foreach \y in {2,5} {
          \coordinate (temp) at ($(origin)+(0,-\y*\ysp)$);
          \draw[deme1] (left |- temp) -- (right |- temp);
        }
        \foreach \y in {1,4,7} {
          \coordinate (temp) at ($(origin)+(0,-\y*\ysp)$);
          \draw[deme2] (left |- temp) -- (right |- temp);
        }
        \foreach \y in {3,6,8} {
          \coordinate (temp) at ($(origin)+(0,-\y*\ysp)$);
          \draw[deme1] (left |- temp) -- (right |- temp);
        }
        \coordinate (temp) at ($(origin)+(0,-2*\ysp)$);
        \node[text=royalblue,font=\Huge] at (mid |- temp) {$\bullet$};
        \coordinate (temp) at ($(origin)+(0,-6*\ysp)$);
        \node[text=royalblue,font=\Huge] at (mid |- temp) {$\bullet$};
        \draw[timepoint] (mid) -- (mid |- bleft);
        \draw[axis,->] (bleft) -- (bright);
        \node[xshift=1ex] (tlab) at (bright.east) {$t$};
      \end{scope}

      %% Sample event, pruned
      \begin{scope}[shift={(4,-2)}]
        \coordinate (origin) at (0,0);
        \node[label] (lab) at (origin) {\textbf{F}};
        \coordinate (left) at ($(origin)+(0.1,-0.1)$);
        \coordinate (right) at ($(left)+(1,0)$);
        \coordinate (mid) at ($(left)!0.5!(right)$);
        \coordinate (rlab) at ($(right)+(0.1,0)$);
        \coordinate (bleft) at ($(left)+(0,-8*\ysp)$);
        \coordinate (bright) at (rlab |- bleft);
        \coordinate (temp) at ($(origin)+(0,-2*\ysp)$);
        \draw[deme1] (left |- temp) -- (mid |- temp);
        \draw[deme1H] (mid |- temp) -- (right |- temp);
        \foreach \y in {3,8} {
          \coordinate (temp) at ($(origin)+(0,-\y*\ysp)$);
          \draw[deme1H] (left |- temp) -- (right |- temp);
        }
        \foreach \y in {5,6} {
          \coordinate (temp) at ($(origin)+(0,-\y*\ysp)$);
          \draw[deme1] (left |- temp) -- (right |- temp);
        }
        \foreach \y in {1,4} {
          \coordinate (temp) at ($(origin)+(0,-\y*\ysp)$);
          \draw[deme2] (left |- temp) -- (right |- temp);
        }
        \foreach \y in {7} {
          \coordinate (temp) at ($(origin)+(0,-\y*\ysp)$);
          \draw[deme2H] (left |- temp) -- (right |- temp);
        }
        \coordinate (temp) at ($(origin)+(0,-2*\ysp)$);
        \node[text=royalblue,font=\Huge] at (mid |- temp) {$\bullet$};
        \coordinate (temp) at ($(origin)+(0,-6*\ysp)$);
        \node[text=royalblue,font=\Huge] at (mid |- temp) {$\bullet$};
        \draw[timepoint] (mid) -- (mid |- bleft);
        \draw[axis,->] (bleft) -- (bright);
        \node[xshift=1ex] (tlab) at (bright.east) {$t$};
      \end{scope}
    \end{tikzpicture}
  }
\end{center}
  \caption{
    \textbf{Lineage count and saturation.}
    Each panel shows the neighborhood of a single event in the unpruned genealogy (top row) and the corresponding pruned genealogy (bottom row).
    Pruning consists of the removal of all branches that are not ancestral to some sample.
    In the bottom row of panels, pruned branches are indicated using broken lines.
    \textbf{(A)} A birth-type event with production $r=(r_{\deme{grey}},r_{\deme{yellow}})=(1,1)$ occurs.
    \textbf{(B)} Suppose that pruning results in the removal of the dashed lineages.
    Then the lineage count at this event-time is $\ell=(\ell_{\deme{grey}},\ell_{\deme{yellow}})=(2,2)$.
    The saturation is $s=(0,1)$ since only a single, yellow lineage emerges from the event.
    \textbf{(C)} A migration-type event with production $r=(0,1)$ occurs.
    \textbf{(D)} After pruning, $\ell=(2,2)$ and $s=(0,1)$.
    \textbf{(E)} A sample-type event occurs in which two grey lineages are sampled (production $r=(2,0)$).
    \textbf{(F)} After pruning, $\ell=(2,2)$ and $s=(1,0)$.
    Observe that in panels B and D, the local structures of the pruned genealogies are identical, though they arise from events of different type.
    \label{fig:ells}
  }
\end{figure}

\subsubsection{Compatibility}
\label{sec:compatibility}

Suppose $P$ is a pruned genealogy, with $\time(P)=T$ and $t\in{[0,T]}$.
The local structure of $P$ at $t$ is, in general, compatible with only a subset of the possible jumps $\Jumps$.
For example, if there is a branch node or a sample node at $t$, then it is compatible only with birth-type or sample-type jumps, respectively.
Similarly, if a lineage moves from deme $i$ to deme $i'$ at $t$, then $u$ must be either of $i\to{i'}$ migration type or of a birth type with parent in $i$ and $r^u_{i'}>0$.
To succinctly accommodate all possibilities, let us introduce the indicator function $Q$ such that $Q=1$ if the local genealogy structure---which is captured by the values of $P^{\lab{Y}}$ at and just before $t$---is compatible with an event of type $u$ and $Q=0$ otherwise.
That is, $Q_u(y,y')=1$ if and only if
there is a feasible genealogy, $\G=(\T,\Z,\Y)$, and history, $\H$,
and a $t\in{[0,\T]}$ such that,
given $\Gr_\T=\G$ and $\Hr_\T=\H$,
we have $\U_t=u$, $\Yt_t=y$, and $\Y_t=y'$.
We refer to $Q$ as the \emph{compatibility indicator}.

\subsubsection{Obscured genealogy}
\label{sec:obscuration}

The \emph{obscured genealogy} is obtained by discarding all information about demes and events not visible from the topology of the tree alone (\zcref{fig:upo}B--C).
In particular, if $P=(T,Z,Y)$ is a pruned genealogy, we write $\obs(P)=(T,Z)$ to denote the obscured genealogy.

\subsection{Binomial ratio}
\label{sec:binomial-ratio}

The statement of the theorems to come is made easier with the following definition.
For $n,r,\ell,s\in{\Zp^\Demes}$,
we define the \emph{binomial ratio}
\begin{equation}
  \label{eq:binratio}
  \binratio{n}{\ell}{r}{s}\coloneq
  \begin{cases}
    \frac{\displaystyle\prod_{i\in\Demes}{\tbinom{n_i-\ell_i}{r_i-s_i}}}%
         {\displaystyle\prod_{i\in\Demes}{\tbinom{n_i}{r_i}}},
         & \text{if}\ \forall i\ n_i\ge{\Set{\ell_i,r_i}}\ge{s_i}\ge{0},\\
         0, & \text{otherwise}.
  \end{cases}
\end{equation}
where $\tbinom{a}{b}=\tfrac{a!}{(a-b)!\,b!}$ is the binomial coefficient.
Observe that $\tbinratio{n}{\ell}{r}{s}\in{[0,1]}$.
Moreover, in consequence of the Chu-Vandermonde identity, we have
\begin{equation}
  \label{eq:chu-vandermonde}
  \sum_{s\in\Zp^{\Demes}}\binratio{n}{\ell}{r}{s}\binom{\ell}{s}=1,
\end{equation}
whenever $n_i\ge{\Set{\ell_i,r_i}}\ge{0}$ for all $i$.

\section{Results}
\label{sec:results}

In the following, \zcref{thm:pruned-lik} establishes the likelihood of a given pruned genealogy conditional on the history of the population process.
\zcref[S]{thm:pruned-filter} then shows how the history can be integrated out.
\zcref[S]{thm:obsc-lik} gives the likelihood of an obscured genealogy conditional on the population history via an auxiliary coloring process.
The main result, \zcref{thm:obsc-filter}, integrates out the history and coloring processes jointly to obtain the likelihood.
\zcref[S]{corol:pruned-filter,corol:obsc-filter} provide the adjoint (backward) formulations of the filter equations derived in \zcref{thm:pruned-filter,thm:obsc-filter}, respectively.
In \zcref{sec:special-cases}, we demonstrate that the closed-form likelihood formulas for the \citet{Kingman1982a} coalescent and linear birth-death-sampling processes fall out as special cases of these results.

\subsection{Likelihood for pruned genealogies}

\begin{thm}
  \label{thm:pruned-lik}
  Suppose $\P=(\T,\Z,\Y)$ is a given pruned genealogy.
  Define
  \begin{equation}\label{eq:phidef}
    \phi^{}_u(x,y',y)\coloneq\binratio{n(x)}{\ell(y)}{r^{u}}{s(y)}\,Q_{u}(y',y),
  \end{equation}
  where $n$ is the deme occupancy (\zcref{sec:demes}), $r^u$ is the production (\zcref{sec:production}), $\ell$ and $s$ are the lineage-count and saturation functions, respectively (\zcref{sec:ells}), $Q$ is the compatibility indicator (\zcref{sec:compatibility}), and the binomial ratio is as defined in \zcref{sec:binomial-ratio}.
  Then
  \begin{equation*}
    \CondProb{\Pr_\T=\P}{\Hr_\T=\H}=\Indicator{\event{\H}\supseteq{\event{\P}}}\,\prod_{t\in\event{\H}}{\phi^{}_{\U_t}\!(\X_t,\Yt_t,\Y_t)}.
  \end{equation*}
\end{thm}

The proof is given in \zcref{sec:proof}.

Next, we show how the likelihood of a pruned genealogy, unconditional on the history, can be computed.
For this, we use the filter equation technology developed in \zcref{sec:filter-eqns}.
In particular, the following theorem follows immediately from \zcref{thm:pruned-lik,lemma:sing-filt}.

\begin{thm}
  \label{thm:pruned-filter}
  Suppose that $\P=(\T,\Z,\Y)$ is a given pruned genealogy.
  Suppose that $w=w(t,x)$ is \cadlag\ in $t$ and satisfies the initial condition $w(0,x)=p_0(x)$ and the filter equation
  \begin{equation}
    \label{eq:pruned-filter-eqn}
    \begin{aligned}
      \frac{\partial{w}}{\partial{t}}(t,x)=
      &\sum_u{\int{w(t,x')\,\alpha_u(t,x',x)\,\phi^{}_u(x,\Yt_t,\Y_t)\,\dd{x'}}}
      -\sum_u{\int{w(t,x)\,\alpha_u(t,x,x')\,\dd{x'}}},
      &t\notin{\event{\P}},\\
      w(t,x)=&\sum_u{\int{\wt(t,x')\,\alpha_u(t,x',x)\,\phi^{}_u(x,\Yt_t,\Y_t)\,\dd{x'}}},
      &t\in{\event{\P}},
    \end{aligned}
  \end{equation}
  where $\phi$ is defined in \zcref{eq:phidef}.
  Then the likelihood of $\P$ is
  \begin{equation*}
    \lik=\int{w(T,x)\,\dd{x}}.
  \end{equation*}
\end{thm}

\begin{remark}
  \label{rem:w-interpretation}
  The quantity $w(t,x)$ appearing in \zcref{eq:pruned-filter-eqn} is not itself obviously a probability density, except at $t=T$.
  Rather, as $t$ increases, $w(t,x)$ accumulates the contributions to the likelihood associated with progressively greater portions of the genealogy.
  Specifically, $w(t,x)$ represents an expectation of partial versions of the products appearing in \zcref{thm:pruned-lik} (containing only the factors for times $\le{t}$) over population trajectories having terminal value $x$ at time $t$.
\end{remark}

One can alternatively express the likelihood of a pruned genealogy using the adjoint form of the filter equation (see \zcref{sec:adjoint-filter}).

\begin{corol}
  \label{corol:pruned-filter}
  If $\P=(\T,\Z,\Y)$ is a given pruned genealogy and $F=F(s,x)$ is \caglad\ in $s$ and satisfies the final condition $\rightlim{F}(\T,x)=1$ and the adjoint filter equation
  \begin{equation}
    \label{eq:pruned-adjoint-eqn}
    \begin{aligned}
      -\frac{\partial{F}}{\partial{s}}(s,x)=
      &\sum_u{\int{\leftlim{\alpha}_u(s,x,x')\,\left[\phi^{}_u(x',\Yt_s,\Y_s)\,F(s,x')-F(s,x)\right]\,\dd{x'}}},
      &s\notin{\event{\P}},\\
      F(s,x)=&\sum_u{\int{\alpha_u(s,x,x')\,\phi^{}_u(x',\Yt_s,\Y_s)\,\rightlim{F}(s,x')\dd{x'}}},
      &s\in{\event{\P}},
    \end{aligned}
  \end{equation}
  where $0\le{s}\le{T}$ and $\phi$ is defined in \zcref{eq:phidef},
  then the likelihood of $\P$ is
  \begin{equation*}
    \lik=\int{F(0,x)\,p_0(x)\,\dd{x}}.
  \end{equation*}
\end{corol}

\begin{remark}
  \label{rem:ordering}
  Note that the likelihood given by the foregoing expressions is the likelihood of a given genealogy \emph{in which the samples are ordered} and that, while this order should be compatible with time-ordering, it may not be completely determined by it.
  In particular, if $n$ samples are taken at one time, there are $n!$ possible orderings of the samples, each of which has an identical likelihood.
  In consequence, the likelihood given above will differ by a factor of $n!$ from the likelihood computed ignoring ordering of samples.
  From the standpoint of parameter estimation, this is immaterial, since factors that depend only on the data drop out of the score function.
  This consideration is relevant, however, in the context of model comparison.
\end{remark}

\subsection{Likelihood for obscured genealogies}

Our next result concerns the likelihood of a given obscured genealogy conditional on the history.

\begin{thm}
  \label{thm:obsc-lik}
  Suppose that $(\T,\Z)$ is a given obscured genealogy.
  Let $q$ and $\pi$ be probability kernels, such that
  for all $x\in\Xspace$ and $y\in\fY_0(\Z)$,
  \begin{equation*}
    \begin{gathered}
      q(x,y)\ge{0},\qquad
      \sum_{y\in\fY_0(\Z)}{q(x,y)}=1,
    \end{gathered}
  \end{equation*}
  and, for all $u\in\Jumps$, $t\in\Rp$, $x,x'\in\Xspace$, $y,y'\in\fY_t(\Z)$,
  \begin{equation*}
    \begin{gathered}
      \pi_u(t,x,x',y,y')\ge{0},\qquad
      \sum_{y'\in\fY_t(\Z)}{\pi_u(t,x,x',y,y')}=1.
    \end{gathered}
  \end{equation*}
  Suppose moreover that $\pi_u(t,x,x',y,y')>0$ whenever $\alpha_u(t,x,x')\,Q_u(y,y')>0$
  and that $q(x,y)>0$ whenever $\CondProb{\Pr_0^{\lab{Y}}=y}{\Xr_0=x}>0$.
  Then there is a stochastic jump process $\yr_t$ with sample paths in $\fY(\Z)$ such that $(\Xr_t,\Ur_t,\yr_t)$ is Markov and
  \begin{equation*}
    \CondProb{\Pr^{\lab{Z}}_\T=\Z}{\Hr_\T=\H}=
    \Indicator{\event{\H}\supseteq\event{\Z}}\,\Expect{\frac{1}{q(X_0,\yr_0)}\,\prod_{t\in\event{\H}}{\frac{\phi^{}_{\U_t}(\X_t,\yrt_t,\yr_t)}{\pi_{\U_t}(t,\Xt_t,\X_t,\yrt_t,\yr_t)}}},
  \end{equation*}
  where $\phi$ is defined in \zcref{eq:phidef} and
  the expectation is taken over the sample paths of $\yr_t$.
\end{thm}
\begin{proof}
  First, observe that, since $\obs$ is a deterministic operator,
  \begin{equation}\label{eq:IS1}
    \CondProb{\Pr^{\lab{Z}}_\T=\Z}{\Hr_\T=\H}=\CondExpect{\Indicator{\Pr^{\lab{Z}}_\T=\Z}}{\Hr_\T=\H}.
  \end{equation}
  Our strategy will be to evaluate \zcref{eq:IS1} using a change of measure:
  we will propose pruned genealogies compatible with $\Z$ as sample paths from a stochastic process driven by $\Xr_t$ and
  evaluate the the expectation in \zcref{eq:IS1} by a suitably weighted expectation over these paths.
  Conditional on $\Hr_\T=\H$, the initial distribution $q$ and probability kernel $\pi$ generate a Markov chain, $\yrh_k$ such that
  \begin{equation*}
    \begin{gathered}
      \CondProb{\yrh_0}{\Hr_\T=\H}=q(\X_0,\yrh_0),
      \qquad
      \CondProb{\yrh_k}{\yrh_{k-1},\Hr_\T=\H}=\pi_{\Uh_k}(\Th_k,\Xh_{k-1},\Xh_{k},\yrh_{k-1},\yrh_k).
    \end{gathered}
  \end{equation*}
  The required process $\yr_t$ is the unique \caglad\ process with event times $\Th_k$ and $\yrh_k$ as its embedded chain.
  This construction of $\yr_t$ obviously guarantees that $\event{\H}\supseteq\event{\yr}\supseteq\event{\Z}$ and that $(\Xr_t,\Ur_t,\yr_t)$ is Markov.

  Now, for $\y\in\fY(\Z)$, let us define $C(\y)=(\T,\Z,\y)$.
  Then, by construction, $\obs(C(\y))=(\T,\Z)$ and,
  conversely, for every pruned genealogy $\P$ satisfying $\time(\P)=\T$ and $\P^{\lab{Z}}=\Z$, $C(\P^{\lab{Y}})=\P$.
  Moreover, the conditions on the kernels $q$ and $\pi$ guarantee that, if $\CondProb{\Pr_\T=\P}{\Hr_\T=\H}>0$ and $\P^{\lab{Z}}=\Z$, then $\CondProb{\yr=\P^{\lab{Y}}}{\Hr_\T=\H}>0$.
  We therefore have that
  \begin{equation*}
    \CondProb{\Pr^{\lab{Z}}_\T=\Z}{\Hr_\T=\H}=
    \Expect{\frac{\CondProb{\Pr_\T=C(\yr)}{\Hr_\T=\H}}{\pi(\yr\vert\H)}},
  \end{equation*}
  the expectation being taken with respect to the random process $\yr$.
  Here, by definition,
  \begin{equation*}
    \pi(\yr\vert\H)=q(\X_0,\yr_0)\,\prod_{t\in\event{\H}}{\pi_{\U_{t}}(t,\Xt_{t},\X_{t},\yrt_{t},\yr_{t})}.
  \end{equation*}
  The result then follows from \zcref{thm:pruned-lik}.
\end{proof}

%% Note that, since $\fY_t(\Z)$ is finite, it is permissible, for example, to choose $q$ and $\pi$ to be uniform.

The next result is the central result of the paper.
It shows how to compute the likelihood of an obscured genealogy.
It is an immediate consequence of \zcref{thm:obsc-lik,lemma:sing-filt}.

\begin{thm}
  \label{thm:obsc-filter}
  Let $(\T,\Z)$ be a given obscured genealogy.
  Then there are probability kernels $q$ and $\pi$ as in \zcref{thm:obsc-lik} such that if
  \begin{equation*}
    \begin{gathered}
      \beta_u(t,x,x',y,y')=\alpha_u(t,x,x')\,\pi_u(t,x,x',y,y'),\qquad
      \varPsi_u(t,x,x',y,y')=\frac{\phi^{}_u(x',y,y')}{\pi_u(t,x,x',y,y')},
    \end{gathered}
  \end{equation*}
  and if $w=w(t,x,y)$ satisfies
  the initial condition $w(0,x,y)=p_0(x)\,\Indicator{q(x,y)>0}$
  and the filter equation
  \begin{mathsize}{9pt}{10pt}
    \begin{align}
      \label{eq:obsc-filter-reg}
      &\frac{\partial{w}}{\partial{t}}=
      \sum_{uy'}{\int{w(t,x',y')\,\beta_u(t,x',x,y',y)\,\varPsi_u(t,x',x,y',y)\,\dd{x'}}}
      -\sum_{uy'}{\int{w(t,x,y)\,\beta_u(t,x,x',y,y')\,\dd{x'}}},
      &t\notin{\event{\Z}},\\
      \label{eq:obsc-filter-sing}
      &w(t,x,y)=\sum_{uy'}{\int{\wt(t,x',y')\,\beta_u(t,x',x,y',y)\,\varPsi_u(t,x',x,y',y)\,\dd{x'}}},
      &t\in{\event{\Z}},
    \end{align}
  \end{mathsize}%
  then the likelihood of $(\T,\Z)$ is
  \begin{equation*}
    \lik=\sum_{y\in\fY_{\T}(\Z)}{\int{w(\T,x,y)\,\dd{x}}}.
  \end{equation*}
\end{thm}

\begin{remark}
  \label{rem:part-obsc}
  Although, as defined, an obscured genealogy does not include the coloring of lineages, it is straightforward to accommodate sample-associated metadata.
  In particular, if one models the metadata dependence on the population state and lineage color, the likelihood of the metadata enters $\phi$ multiplicatively.
\end{remark}

\zcref[S]{lemma:monte-carlo} shows how the likelihood can be computed via sequential Monte Carlo.
As before, there is also an adjoint form of \zcref{thm:obsc-filter}, to wit:

\begin{corol}
  \label{corol:obsc-filter}
  Let $(\T,\Z)$ be a given obscured genealogy and let $\beta$ and $\varPsi$ be defined as in \zcref{thm:obsc-filter}.
  Suppose $F$ is \caglad\ and satisfies the final condition
  $\rightlim{F}(\T,x,y)=1$ for all $x\in\Xspace$, $y\in\fY_{\T}(\Z)$ and also the adjoint filter equation
  \begin{align}
    \label{eq:obsc-adj-reg}
    -&\frac{\partial{F}}{\partial{s}}=
    \sum_{uy'}{\int{\leftlim{\beta}_u(s,x,x',y,y')\,\left[\leftlim{\varPsi}_u(s,x,x',y,y')\,F(s,x',y')-F(s,x,y)\right]\,\dd{x'}}},
    &s\notin{\event{\Z}},\\
    \label{eq:obsc-adj-sing}
    &F(s,x,y)=\sum_{uy'}{\int{\beta_u(s,x,x',y,y')\,\varPsi_u(s,x,x',y,y')\,\rightlim{F}(s,x',y')\,\dd{x'}}},
    &s\in{\event{\Z}},
  \end{align}
  for $0\le{s}\le{\T}$.
  Then the likelihood of $(\T,\Z)$ is
  \begin{equation*}
    \lik=\sum_{y\in\fY_0(\Z)}{\int{F(0,x,y)\,p_0(x)\,\dd{x}}}.
  \end{equation*}
\end{corol}

\subsection{Special cases}
\label{sec:special-cases}

\subsubsection{Moran model and the Kingman coalescent}
\label{sec:kingman}

In the Moran model \citep{Moran1958,King2020}, events occur according to a rate-$\mu$ Poisson process.
At each event, a compound birth-death jump occurs so that the population size, $n$, remains constant (cf.~\zcref{fig:event-types}F).
If we let $\Xr_t$ be the number of events that have occurred by time $t$, then $\Xr_t$ is a simple counting process, which we can use to define the state of the population process.
Its KFE is then
\begin{mathsize}{9pt}{10pt}
  \begin{equation*}
    \begin{gathered}
      \frac{\partial{v}}{\partial{t}}=
      \mu\,v(t,x-1)
      -\mu\,v(t,x),
      \qquad
      v(0,x) =
      \begin{cases}
        1, & x=0,\\
        0, & x>0.
      \end{cases}
    \end{gathered}
  \end{equation*}
\end{mathsize}%

In the classical case \citep{Kingman1982a}, $m$ samples are taken simultaneously at a single time, $T$.
Since nothing depends on the state, in writing the corresponding filter equation (\zcref{eq:obsc-filter-reg}), we can sum over both $x$ and $y$.
Then, if $B$ is the set of branch-times and $\ell(t)$ is the number of lineages in the genealogy at time $t$, upon this summation, \zcref{eq:obsc-filter-reg,eq:obsc-filter-sing} become
\begin{mathsize}{9pt}{10pt}
  \begin{equation}
    \zcsetup{reftype=pluralequation}
    \label{eq:moran-filter}
    \begin{gathered}
      w(0) = 1,
      \qquad
      \frac{\partial{w}}{\partial{t}}=
      \mu\,w(t)\,\left(1-\frac{\tbinom{\ell(t)}{2}}{\tbinom{n}{2}}\right)
      -\mu\,w(t),\quad t\notin{B},
      \qquad
      w(t) = \frac{\mu}{\tbinom{n}{2}}\,\wt(t),\quad t\in{B}.
    \end{gathered}
  \end{equation}
\end{mathsize}%
Integrating \zcref{eq:moran-filter} and taking logarithms yields
\begin{equation}
  \label{eq:kingman}
  \log{\lik}=|B|\,\log{\frac{\mu}{\tbinom{n}{2}}}-\frac{\mu}{\tbinom{n}{2}}\,\sum_{i=1}^{m}\!{\tbinom{i}{2}\,s^{}_i},
\end{equation}
where the $s^{}_{i}\coloneq\int{\Indicator{\ell(t)=i}\,\dd{t}}$ are the durations of the \emph{coalescent intervals}, \ie intervals between successive branch-points.
We recognize \zcref{eq:kingman} as the expression for the \citet{Kingman1982a} coalescent likelihood \citep[\eg][]{Wakeley2009}.

More generally, if in addition samples are taken according to a rate-$\psi$ Poisson process such that the set of sample-times in the genealogy is $S=S_0\cup{S_1}$, where $S_0$, $S_1$ are the sets of times of terminal and inline samples, respectively, then \zcref{eq:obsc-filter-reg,eq:obsc-filter-sing} reduce to
\begin{mathsize}{9pt}{10pt}
  \begin{equation}
    \zcsetup{reftype=pluralequation}
    \label{eq:moran-filter2}
    \begin{gathered}
      w(0) = 1,
      \qquad
      \frac{\partial{w}}{\partial{t}}=
      -\mu\,\frac{\tbinom{\ell(t)}{2}}{\tbinom{n}{2}}\,w(t),
      \quad t\notin{S\cup{B}},
      \qquad
      \frac{w(t)}{\wt(t)} =
      \begin{cases}
        \frac{\mu}{\tbinom{n}{2}}, & t\in{B},\\[8pt]
        \psi\,\left(1-\frac{\ell(t)}{n}\right), & t\in{S_0},\\
        \frac{\psi}{n}, & t\in{S_1}.
      \end{cases}
    \end{gathered}
  \end{equation}
\end{mathsize}%
Integrating \zcref{eq:moran-filter2} yields
\begin{equation}
  \label{eq:mgp}
  \log{\lik}-|S|\,\log{\psi}
  =\sum_{t\in{S_0}}{\log{\left(1-\frac{\ell(t)}{n}\right)}}
  -|S_1|\,\log{n}
  +|B|\,\log{\frac{\mu}{\tbinom{n}{2}}}
  -\frac{\mu}{\tbinom{n}{2}}\,\sum_{i=1}^{\infty}\!{\tbinom{i}{2}\,s^{}_i},
\end{equation}
in agreement with the result of the very different derivation of \citet{King2020}.

\subsubsection{Linear birth-death model} %% LBDP
\label{sec:lbdp}

In this model, the state variable, $\Xr_t$, is the size of a population at time $t$.
All individuals face the same per-capita birth and death rates, which are $\lambda$ and $\mu$, respectively.
Accordingly, the Kolmogorov backward equation for this model is
\begin{mathsize}{9pt}{10pt}
  \begin{equation*}
    -\frac{\partial{F}}{\partial{s}}=
    \lambda\,x\,\left[F(s,x+1)-F(s,x)\right]
    +\mu\,x\,\left[F(s,x-1)-F(s,x)\right].
  \end{equation*}
\end{mathsize}%
\citet{Stadler2010} considered the case where samples are taken through time at a uniform per-capita rate $\psi$ and lineages extant at time $T$ are sampled with probability $\rho$.
Let $B$ denote the set of branch-times in an observed genealogy, let $S_0$ and $S_1$ be, respectively, the sets of terminal and inline sample times prior to time $T$, and suppose $n$ lineages are sampled at time $T$.
\zcref[S]{eq:obsc-adj-reg} for this process reads
\begin{mathsize}{9pt}{10pt}
  \begin{equation*}
    \label{eq:filter-unsummed-lbdp-reg}
    \begin{aligned}
      -\frac{\partial{F}}{\partial{s}}(s,x,y)=
      &\lambda\,x\,\frac{\tbinom{x+1-\ell(s)}{2}}{\tbinom{x+1}{2}}\,F(s,x+1,y)
      +\sum_{y'\in{\fY_s(\Z)\setminus{y}}}{\lambda\,x\,\frac{x+1-\ell(s)}{x+1}\,F(s,x+1,y')}\\
      &+\mu\,x\,F(s,x-1,y)
      -\left(\lambda+\mu+\psi\right)\,x\,F(s,x,y),
    \end{aligned}
  \end{equation*}
\end{mathsize}%
As before, we can take $F$ to be independent of $y$.
We then sum over $y$ and apply the Chu-Vandermonde identity (\zcref{eq:chu-vandermonde}) to obtain
\begin{mathsize}{9pt}{10pt}
  \begin{equation}
    \label{eq:filter-lbdp-reg}
    -\frac{\partial{F}}{\partial{s}}(s,x)=
    \lambda\,x\,\left(1-\frac{\tbinom{\ell(s)}{2}}{\tbinom{x+1}{2}}\right)\,F(s,x+1)
    +\mu\,x\,F(s,x-1)\\
    -\left(\lambda+\mu+\psi\right)\,x\,F(s,x),
  \end{equation}
\end{mathsize}%
which holds for $x\ge\ell(s)$ and $s\notin{B\cup{S_0}\cup{S_1}\cup\Set{T}}$.
The singular part (\zcref{eq:obsc-adj-sing}) reads
\begin{mathsize}{9pt}{10pt}
  \begin{equation}
    \label{eq:filter-lbdp-sing}
    F(s,x)=\begin{cases}
    \frac{2\lambda}{x+1}\,\rightlim{F}(s,x+1),
    &s\in{B},
    \\[2ex]
    \psi\left(x-\ell(s)\right)\,\rightlim{F}(s,x),
    &s\in{S_0},
    \\[2ex]
    \psi\,\rightlim{F}(s,x),
    &s\in{S_1},
    \\[2ex]
    \tbinom{x}{n}\,(1-\rho)^{x-n}\,\rho^n,
    &s=T.
    \end{cases}
  \end{equation}
\end{mathsize}%
We note that $F(s,x)=0$ when $x<\ell(s)$.

The linearity assumption implies self-similarity of the genealogies.
In particular, at each time $s$, all of the subtrees descending from the lineages present at $s$ are i.i.d.
This suggests the ansatz
\begin{equation}
  F(s,x)=C(s)\,G(s)^{x-\leftlim{\ell}(s)}\,H(s)^{\leftlim{\ell}(s)}\,K(x,\leftlim{\ell}(s)),
\end{equation}
where $G$, $H$ are smooth functions and $C$ is \caglad\ and piecewise constant, with discontinuities only at genealogical events, so that
\begin{equation*}
  \rightlim{F}(s,x)=\rightlim{C}(s)\,G(s)^{x-{\ell}(s)}\,H(s)^{{\ell}(s)}\,K(x,{\ell}(s)).
\end{equation*}
Putting these into the first three cases of \zcref{eq:filter-lbdp-sing} yields conditions on $C$ and $K$ that are satisfied when
\begin{equation*}
  \begin{gathered}
    K(x,\ell)=\frac{x!}{(x-\ell)!}
    \qquad\text{and}\qquad
    \frac{C(s)}{\rightlim{C}(s)}=\begin{cases}
         {2\lambda}\,H(s),
         &s\in{B},
         \\[2ex]
         \psi\,G(s)\,H(s)^{-1},
         &s\in{S_0},
         \\[2ex]
         \psi,
         &s\in{S_1}.
    \end{cases}
  \end{gathered}
\end{equation*}
To satisfy the last case of \zcref{eq:filter-lbdp-sing}, we note that $\leftlim{\ell}(T)=n$ and therefore impose the conditions
\begin{equation}
  \zcsetup{reftype=pluralequation}
  \label{eq:FGfc}
  G(T)=1-\rho, \qquad H(T)=1, \qquad C(T)=\frac{\rho^n}{n!}.
\end{equation}
Finally, upon substituting the ansatz into \zcref{eq:filter-lbdp-reg}, we obtain for $s\notin{B\cup{S_0}\cup{S_1}\cup{\Set{T}}}$,
\begin{equation*}
  \begin{split}
    -\frac{1}{F}\frac{\partial{F}}{\partial{s}}
    &=
    -(x-\ell)\,\frac{G'}{G}
    -\ell\,\frac{H'}{H}\\
    &=\lambda\,x\,\left(1-\frac{\tbinom{\ell}{2}}{\tbinom{x+1}{2}}\right)\,\frac{x+1}{x+1-\ell}\,G
    +\mu\,x\,\frac{x-\ell}{x}\,G^{-1}
    -\left(\lambda+\mu+\psi\right)\,x\\
    &=\lambda\,(x+\ell)\,G
    +\mu\,(x-\ell)\,G^{-1}
    -\left(\lambda+\mu+\psi\right)\,x.\\
  \end{split}
\end{equation*}
If this is to hold for all $x$ and $\ell$, we must have
\begin{equation}
  \zcsetup{reftype=pluralequation}
  \label{eq:FGeq}
  \begin{gathered}
    -\frac{G'}{G}
    =\lambda\,G+\mu\,G^{-1}
    -\left(\lambda+\mu+\psi\right),
    \qquad
    -\frac{H'}{H}
    =2\,\lambda\,G-(\lambda+\mu+\psi).
  \end{gathered}
\end{equation}
In passing, we note that \zcref{eq:FGeq} appear prominently in the very different derivation of \citet{Stadler2010}.
\zcref[S]{eq:FGfc,eq:FGeq} form a final-value problem that can be solved by elementary means to obtain closed-form expressions for $G(s)$ and $H(s)$.
%% \begin{equation*}
%%   \begin{gathered}
%%     G = \frac{\delta\,(1-\rho)\,\cosh\left(\tfrac{\delta}{2}\,(s-T)\right)+\left(\lambda-\mu+\psi-\rho\,(\lambda+\mu+\psi)\right)\,\sinh\left(\tfrac{\delta}{2}\,(s-T)\right)}{\delta\,\cosh\left(\tfrac{\delta}{2}\,(s-T)\right)+(\lambda-\mu-\psi-2\,\lambda\,\rho)\,\sinh\left(\tfrac{\delta}{2}\,(s-T)\right)}\\
%%     H = \left(\cosh\left(\tfrac{\delta}{2}\,(s-T)\right)+\frac{\lambda-\mu-\psi-2\,\lambda\,\rho}{\delta}\,\sinh\left(\tfrac{\delta}{2}\,(s-T)\right)\right)^{-2}\\
%%     \delta=\sqrt{(\lambda+\mu+\psi)^2-4\,\lambda\,\mu}
%%   \end{gathered}
%% \end{equation*}
\zcref[S]{thm:obsc-filter} then states that $\lik=\int{F(0,x)\,p_0(x)\,\dd{x}}$.
In particular, conditional on $\Xr_0=x$, we have
\begin{equation}
  \lik=\frac{x!}{(x-\ell(0))!}\,\frac{\rho^n}{n!}\,G(0)^{x-{\ell}(0)}\,H(0)^{{\ell}(0)}\,\psi^{|S_1|}\,\prod_{\mathclap{t\in{B}}}{\left(2\,\lambda\,H(t)\right)}\,\prod_{\mathclap{t\in{S_0}}}{\left(\psi\,{G(t)}{H(t)^{-1}}\right)}.
\end{equation}
\citet[Thm.~3.5]{Stadler2010} derives an expression for this quantity in the special case $x=\ell(0)=1$, which is equivalent up to the factor of $n!$ explained in \zcref{rem:ordering}.
%% (our $G$, $H$ are related to the $p_0$, $q$ of that paper by the relations $G(s)=p_0(T-s)$, $H(s)=4/q(T-s)$).

The foregoing represents an independent verification of the correctness of the theory in the case of the simple linear birth-death process with constant sampling through time and one bulk-sampling event.
It should be clear, however, that even within the context of the unstructured, linear, birth-death process, the theory presented here readily accommodates complexities such as time-varying sampling rates, multiple bulk-sampling events at specified or random times, and time-varying birth and death rates.

\section{Worked examples}
\label{sec:worked-examples}

Let us now illustrate the theory by working out the filter equations in the case of two familiar nonlinear models.
In particular, we will show how the general filter equation introduced in \zcref{thm:obsc-filter} takes on a specific form that depends only on the model structure.
In each case, we will also provide results of some elementary numerical computations based on these filters, using methods described in \zcref{sec:filter-eqns}.

In this section, the mathematical expressions are complex and, ultimately, the algorithms involved in actually computing the likelihood can be expressed more economically in computer code.
However, the interested reader may find it instructive to see how the general expressions of \zcref{thm:obsc-filter} are translated into concrete equations relative to a particular model.
In particular, the following examples are a reliable guide to this process, since the sequence of evolutions performed in each of the two examples below is entirely mechanical and the fairly complex notation needed is a faithful reflection of the algorithmic structures involved in carrying out the numerical computations.
Before turning to the examples, therefore, we first develop the necessary notation.

\subsection{Chop, swap, and fork operators}
\label{sec:swapnchop}

\begin{figure}
  \begin{center}
  \resizebox{0.8\linewidth}{!}{
    \begin{tikzpicture}[scale=1.8]
      \usetikzlibrary{shapes,arrows,positioning}
      \usetikzlibrary{arrows.meta,positioning,decorations}
      \tikzstyle{label}=[font=\Large]
      \tikzstyle{coordinate}=[inner sep=0pt,outer sep=0pt]
      \tikzstyle{deme1}=[color=grey, line width=3pt]
      \tikzstyle{deme1H}=[color=grey, line width=3pt, densely dotted]
      \tikzstyle{deme2}=[color=maize, line width=3pt]
      \tikzstyle{deme2H}=[color=maize, line width=3pt, densely dotted]
      \tikzstyle{axis}=[color=black, thick, >=stealth]
      \tikzstyle{timepoint}=[color=red!70!black]

      \def\ysp{0.2} %% spacing between adjacent lineages

      %% Chop to emptyset
      \begin{scope}[shift={(0,0)}]
        \coordinate (origin) at (0,0);
        \node[label] (lab) at (origin) {\textbf{A}};
        \coordinate (left) at ($(origin)+(0.1,-0.1)$);
        \coordinate (right) at ($(left)+(1,0)$);
        \coordinate (before) at ($(left)!0.3!(right)$);
        \coordinate (mid) at ($(left)!0.5!(right)$);
        \coordinate (at) at ($(left)!0.6!(right)$);
        \coordinate (uright) at ($(right)+(0.1,0)$);
        \coordinate (bleft) at ($(left)+(0,-7*\ysp)$);
        \coordinate (bright) at (uright |- bleft);
        \coordinate (text) at ($(bleft)+(0,-\ysp)$);
        \coordinate (temp) at ($(origin)+(0,-3*\ysp)$);
        \draw[timepoint] (at) -- (at |- bleft);
        \draw[timepoint] (before) -- (before |- bleft);
        \draw[deme2] (left |- temp) -- (at |- temp);
        \node[font=\footnotesize,anchor=east] at (left |- temp) {$\Set{a}$};
        \node[text=royalblue,font=\Huge] at (at |- temp) {$\bullet$};
        \node[text=white,font=\footnotesize] at (at |- temp) {$a$};
        \foreach \y in {1,2,6} {
          \coordinate (temp) at ($(origin)+(0,-\y*\ysp)$);
          \draw[deme1] (left |- temp) -- (right |- temp);
        }
        \foreach \y in {4,5,7} {
          \coordinate (temp) at ($(origin)+(0,-\y*\ysp)$);
          \draw[deme2] (left |- temp) -- (right |- temp);
        }
        \draw[axis,->] (bleft) -- (bright);
        \node[xshift=1ex] (tlab) at (bright.east) {$t$};
        \node (y) at (before |- lab) {$y^{\phantom{\prime}}$};
        \node (yprime) at (at |- lab) {$y^{\prime}$};
        \node at (mid |- text) {$y'=\chop{\deme{I}}{a\emptyset}{y}$};
      \end{scope}

      %% Chop to non-empty set
      \begin{scope}[shift={(1.8,0)}]
        \coordinate (origin) at (0,0);
        \node[label] (lab) at (origin) {\textbf{B}};
        \coordinate (left) at ($(origin)+(0.1,-0.1)$);
        \coordinate (right) at ($(left)+(1,0)$);
        \coordinate (before) at ($(left)!0.3!(right)$);
        \coordinate (mid) at ($(left)!0.5!(right)$);
        \coordinate (at) at ($(left)!0.6!(right)$);
        \coordinate (uright) at ($(right)+(0.1,0)$);
        \coordinate (bleft) at ($(left)+(0,-7*\ysp)$);
        \coordinate (bright) at (uright |- bleft);
        \coordinate (text) at ($(bleft)+(0,-\ysp)$);
        \coordinate (temp) at ($(origin)+(0,-5*\ysp)$);
        \draw[timepoint] (at) -- (at |- bleft);
        \draw[timepoint] (before) -- (before |- bleft);
        \draw[deme2] (left |- temp) -- (right |- temp);
        \node[font=\footnotesize,anchor=east] at (left |- temp) {$\Set{a}\cup{b}$};
        \node[font=\footnotesize,anchor=west] at (right |- temp) {$b$};
        \node[text=royalblue,font=\Huge] at (at |- temp) {$\bullet$};
        \node[text=white,font=\footnotesize] at (at |- temp) {$a$};
        \foreach \y in {2,4,7} {
          \coordinate (temp) at ($(origin)+(0,-\y*\ysp)$);
          \draw[deme1] (left |- temp) -- (right |- temp);
        }
        \foreach \y in {1,3,6} {
          \coordinate (temp) at ($(origin)+(0,-\y*\ysp)$);
          \draw[deme2] (left |- temp) -- (right |- temp);
        }
        \draw[axis,->] (bleft) -- (bright);
        \node[xshift=1ex] (tlab) at (bright.east) {$t$};
        \node (y) at (before |- lab) {$y^{\phantom{\prime}}$};
        \node (yprime) at (at |- lab) {$y^{\prime}$};
        \node at (mid |- text) {$y'=\chop{\deme{II}}{ab}{y}$};
      \end{scope}

      %% Swap
      \begin{scope}[shift={(3.6,0)}]
        \coordinate (origin) at (0,0);
        \node[label] (lab) at (origin) {\textbf{C}};
        \coordinate (left) at ($(origin)+(0.1,-0.1)$);
        \coordinate (right) at ($(left)+(1,0)$);
        \coordinate (before) at ($(left)!0.3!(right)$);
        \coordinate (mid) at ($(left)!0.5!(right)$);
        \coordinate (at) at ($(left)!0.6!(right)$);
        \coordinate (uright) at ($(right)+(0.1,0)$);
        \coordinate (bleft) at ($(left)+(0,-7*\ysp)$);
        \coordinate (bright) at (uright |- bleft);
        \coordinate (text) at ($(bleft)+(0,-\ysp)$);
        \coordinate (temp) at ($(origin)+(0,-2*\ysp)$);
        \draw[timepoint] (at) -- (at |- bleft);
        \draw[timepoint] (before) -- (before |- bleft);
        \node[font=\footnotesize,anchor=east] at (left |- temp) {$b$};
        \node[font=\footnotesize,anchor=west] at (right |- temp) {$b$};
        \draw[deme1] (left |- temp) -- (at |- temp);
        \draw[deme2] (at |- temp) -- (right |- temp);
        \node[text=darkgreen,font=\Huge] at (at |- temp) {$\bullet$};
        \foreach \y in {1,4,6,7} {
          \coordinate (temp) at ($(origin)+(0,-\y*\ysp)$);
          \draw[deme1] (left |- temp) -- (right |- temp);
        }
        \foreach \y in {3,5} {
          \coordinate (temp) at ($(origin)+(0,-\y*\ysp)$);
          \draw[deme2] (left |- temp) -- (right |- temp);
        }
        \draw[axis,->] (bleft) -- (bright);
        \node[xshift=1ex] (tlab) at (bright.east) {$t$};
        \node (y) at (before |- lab) {$y^{\phantom{\prime}}$};
        \node (yprime) at (at |- lab) {$y^{\prime}$};
        \node at (mid |- text) {$y'=\swap{\deme{EI}}{b}{y}$};
      \end{scope}

      %% Swap (no color change)
      \begin{scope}[shift={(5.4,0)}]
        \coordinate (origin) at (0,0);
        \node[label] (lab) at (origin) {\textbf{D}};
        \coordinate (left) at ($(origin)+(0.1,-0.1)$);
        \coordinate (right) at ($(left)+(1,0)$);
        \coordinate (before) at ($(left)!0.3!(right)$);
        \coordinate (mid) at ($(left)!0.5!(right)$);
        \coordinate (at) at ($(left)!0.6!(right)$);
        \coordinate (uright) at ($(right)+(0.1,0)$);
        \coordinate (bleft) at ($(left)+(0,-7*\ysp)$);
        \coordinate (bright) at (uright |- bleft);
        \coordinate (text) at ($(bleft)+(0,-\ysp)$);
        \coordinate (temp) at ($(origin)+(0,-4*\ysp)$);
        \draw[timepoint] (at) -- (at |- bleft);
        \draw[timepoint] (before) -- (before |- bleft);
        \node[font=\footnotesize,anchor=east] at (left |- temp) {$b$};
        \node[font=\footnotesize,anchor=west] at (right |- temp) {$b$};
        \draw[deme2] (left |- temp) -- (at |- temp);
        \draw[deme2] (at |- temp) -- (right |- temp);
        \node[text=darkgreen,font=\Huge] at (at |- temp) {$\bullet$};
        \foreach \y in {3,5,6} {
          \coordinate (temp) at ($(origin)+(0,-\y*\ysp)$);
          \draw[deme1] (left |- temp) -- (right |- temp);
        }
        \foreach \y in {1,2,7} {
          \coordinate (temp) at ($(origin)+(0,-\y*\ysp)$);
          \draw[deme2] (left |- temp) -- (right |- temp);
        }
        \draw[axis,->] (bleft) -- (bright);
        \node[xshift=1ex] (tlab) at (bright.east) {$t$};
        \node (y) at (before |- lab) {$y^{\phantom{\prime}}$};
        \node (yprime) at (at |- lab) {$y^{\prime}$};
        \node at (mid |- text) {$y'=\swap{\deme{II}}{b}{y}$};
      \end{scope}

      %% Fork
      \begin{scope}[shift={(7.2,0)}]
        \coordinate (origin) at (0,0);
        \node[label] (lab) at (origin) {\textbf{E}};
        \coordinate (left) at ($(origin)+(0.1,-0.1)$);
        \coordinate (right) at ($(left)+(1,0)$);
        \coordinate (before) at ($(left)!0.3!(right)$);
        \coordinate (mid) at ($(left)!0.5!(right)$);
        \coordinate (at) at ($(left)!0.6!(right)$);
        \coordinate (uright) at ($(right)+(0.1,0)$);
        \coordinate (bleft) at ($(left)+(0,-7*\ysp)$);
        \coordinate (bright) at (uright |- bleft);
        \coordinate (text) at ($(bleft)+(0,-\ysp)$);
        \coordinate (temp) at ($(origin)+(0,-5*\ysp)$);
        \coordinate (temp1) at ($(temp)+(0,-\ysp)$);
        \draw[timepoint] (at) -- (at |- bleft);
        \draw[timepoint] (before) -- (before |- bleft);
        \node[font=\footnotesize,anchor=east] at (left |- temp) {$b\cup{b'}$};
        \node[font=\footnotesize,anchor=west] at (right |- temp) {$b$};
        \node[font=\footnotesize,anchor=west] at (right |- temp1) {$b'$};
        \draw[deme2] (left |- temp) -- (at |- temp);
        \draw[deme2] (at |- temp) -- (right |- temp);
        \draw[deme1] (at |- temp) -- (at |- temp1) -- (right |- temp1);
        \node[text=darkgreen,font=\Huge] at (at |- temp) {$\bullet$};
        \foreach \y in {2,3} {
          \coordinate (temp) at ($(origin)+(0,-\y*\ysp)$);
          \draw[deme1] (left |- temp) -- (right |- temp);
        }
        \foreach \y in {1,4,7} {
          \coordinate (temp) at ($(origin)+(0,-\y*\ysp)$);
          \draw[deme2] (left |- temp) -- (right |- temp);
        }
        \draw[axis,->] (bleft) -- (bright);
        \node[xshift=1ex] (tlab) at (bright.east) {$t$};
        \node (y) at (before |- lab) {$y^{\phantom{\prime}}$};
        \node (yprime) at (at |- lab) {$y^{\prime}$};
        \node at (mid |- text) {$y'=\fork{\deme{IIE}}{bb'}{y}$};
      \end{scope}

    \end{tikzpicture}
  }
\end{center}
  \caption{
    \label{fig:swapnchop}
    \textbf{Chop, swap, and fork operators.}
    Each panel shows the local structure in the vicinity of an event.
    There are two demes, $\Set{\deme{E,I}}$, denoted by grey and yellow, respectively.
    In each case, $y$ denotes the coloring (mapping of branches to demes) just before the event;
    $y'$ denotes the coloring at the event.
    In \textbf{(A)}, the sample labeled $a$ (blue dot) terminates the branch, $\Set{a}$, that bears only that lineage.
    Symbolically, we write $y'=\chop{\deme{I}}{a\emptyset}{y}$.
    In \textbf{(B)}, the sample labeled $a$ is inline:
    the branch labeled $b$ descends from it.
    The branch immediately ancestral to $a$ is therefore $\Set{a}\cup{b}$.
    Here, the color does not change at the sample event.
    Therefore $y'=\chop{\deme{II}}{ab}{y}$.
    In \textbf{(C)}, an event (green dot) occurs on branch $b$ and the color changes from $\deme{E}$ to $\deme{I}$.
    Symbolically, then, $y'=\swap{\deme{EI}}{b}{y}$.
    Panel \textbf{(D)} shows an event without an associated change in color.
    We denote this change with the notation $y'=\swap{\deme{II}}{b}{y}$.
    In this case, observe that while $\col{y'}=\col{y}$, $\ctr{y'}=1\ne{0}=\ctr{y}$, whence $y'\ne{y}$.
    In \textbf{(E)}, a branch event occurs at the green dot, splitting the branches $b$, $b'$ from their parent branch $b\cup{b'}$.
    In this case, $b\cup{b'}$ and $b$ are both in $\deme{I}$ and $b'$ is in $\deme{E}$, whence we write $y'=\fork{\deme{IIE}}{bb'}{y}$.
  }
\end{figure}

To specialize the general results to particular models, one needs to be able to speak about events along specific branches of the genealogy, for which the following notation is useful.
Recall from \zcref{sec:genealogy} the formal definition of a pruned genealogy as a triple $(\T,\Z,\Y)$ where $\T$ is the genealogy time; $\Z$, the tree structure; and $\Y$ the coloring.
In particular, recall that $\Z$ is defined as a monotone map $\Z:[0,\T]\to\partit(\leaves)$, where $\leaves$ is the set of labels attached to samples.
Recall also that $\Y$ is a map from $\Z$ into the set $\ColorSet$.
In the following, suppose that $z\in\partit(\leaves)$ and $y:z\to\ColorSet$.
Thus $z$ is the set of branches present at a particular instant, and $y$ is a possible coloring of $z$.
Recall also that $y$ has two components:
the \emph{color} $\col{y}$ is the mapping $z\to\Demes$ and the \emph{indicator} $\ctr{y}:z\to\Set{0,1}$ indicates where, if anywhere, genealogical events are present in $y$.
Finally, for $i\in\Demes$, let $\cols{i}{y}\coloneq\CondSet{b}{\col{y}(b)=i}$ be the set of all branches in $y$ having color $i$.

Suppose $a\in\leaves$ is a label, $b\in{z}$ is a branch, $a\in{b}$, and $i,j\in\Demes$.
Let $b'=b\setminus\Set{a}$.
If $b'\ne\emptyset$, we define $\chop{}{ab'}{z}={z\cup\Set{b'}\setminus\Set{b}}$ and we let $\chop{}{a\emptyset}{z}={z\setminus\Set{\Set{a}}}$.
This means that the sets of branches $z$ and $\chop{}{ab'}{z}$ are alike except that, in the latter, the label $a$ has been deleted (\zcref{fig:swapnchop}A--B).
Moreover, if $b\in\cols{i}{y}$ and $z'=\chop{}{ab'}{z}$,
define $\chop{ij}{ab'}{y}:z'\to\ColorSet$ by
$\chop{ij}{ab'}{y}(b')=(j,1)$ and $\chop{ij}{ab'}{y}(c)=y(c)$ for $c\ne{b'}$.
Note that if $b=\Set{a}$, then $b'=\emptyset$ and $\chop{ij}{ab'}{y}$ is independent of $j$.
In this case, we can use $\chop{i}{a\emptyset}{y}$ to denote the common value.
In plain language, $\chop{ij}{ab'}{y}$ is the coloring on the branches of $z'$ under which $b'$ is in deme $j$ and all other branches are colored as they are in $y$.
Thus, for example, if $(\T,\Z,\Y)$ is a genealogy and there is a sample event at $t$, the label of the sample being $a\in\leaves$, then $\Z_t=\chop{}{ab'}{\Zt_t}$.
If $\Yt_t(b)=i$ and it is a terminal sample event, then $\Y_t=\chop{i}{a\emptyset}{\Yt_t}$;
if it is an inline sample, the descending lineage being in deme $j$, then $\Y_t=\chop{ij}{ab'}{\Yt_t}$.
It is natural also to write $\chop{ij}{ab}{\col{y}}=\col{\chop{ij}{ab}{y}}$ and $\chop{ij}{ab}{\ctr{y}}=\ctr{\chop{ij}{ab}{y}}$.
We call $\chop{}{}{}$ the \emph{chop operator}.

Again suppose $b\in{z}$ and $i,j\in\Demes$.
If $b\in\cols{i}{y}$, then
define $\swap{ij}{b}{y}:z\to\ColorSet$ to be another coloring of $z$ that agrees with $y$ everywhere but on branch $b$.
On branch $b$, the color is swapped from $i$ to $j$ and the indicator is set to 1 (\zcref{fig:swapnchop}C--D).
That is, $\swap{ij}{b}{y}(b)=(j,1)$ and $\swap{ij}{b}{y}(c)=y(c)$ for all $c\ne{b}$.
Again, we make the natural definitions $\swap{ij}{b}{\col{y}}=\col{\swap{ij}{b}{y}}$ and $\swap{}{b}{\ctr{y}}=\ctr{\swap{ij}{b}{y}}$,
the latter holding for any $j$ as long as $y(b)=i$.
Thus, the result of applying $\swap{ij}{b}{}$ is to change the coloring of branch $b$ from $i$ to $j$.
We refer to $\swap{}{}{}$ as the \emph{swap operator}.

Now let us suppose that $c,c'$ are disjoint, nonempty subsets of $\leaves$, such that $c\cup{c'}=b\in{z}$.
Let $\fork{}{cc'}{z}=z\cup\Set{c,c'}\setminus\Set{b}$.
Thus $\fork{}{cc'}{z}$ is identical to $z$ except in that $b$ has been forked into two branches, $c$ and $c'$ (\zcref{fig:swapnchop}E).
Similarly, if $i,j,k\in\Demes$ and $b\in\cols{i}{y}$, we define
$\fork{ijk}{cc'}{y}$ to be the coloring of $\fork{}{cc'}{z}$ such that when $y'=\fork{ijk}{cc'}{y}$, then $\col{y'(c)}=j$, $\col{y'(c')}=k$, and $\ctr{y'(c)}=\ctr{y'(c')}=1$.
Therefore, the result of applying $\fork{ijk}{cc'}{}$ to a coloring $y$ is that the $i$-colored branch $c\cup{c'}$ is forked into a $j$-colored branch $c$ and a $k$-colored branch $c'$.
Again, it is natural to define $\fork{ijk}{cc'}{\col{y}}=\col{\fork{ijk}{cc'}{y}}$ and $\fork{ijk}{cc'}{\ctr{y}}=\ctr{\fork{ijk}{cc'}{y}}$.
We call $\fork{}{}{}$ the \emph{fork operator}.
Note that the notation for the fork operator accommodates only binary forks, which will be sufficient for our purposes below.
However, it would be fairly straightforward to extend the notation to deal with the general case of $n$-ary forks.

Finally, let us define reversals of these operators.
It is easy to verify that there are well defined operators $\backchop{}{}{}$, $\backswap{}{}{}$, and $\backfork{}{}{}$ such that
\begin{equation}
  \label{eq:swapnchop}
  \begin{aligned}
    z'=\chop{}{a\emptyset}{z},\ y'=\chop{i}{a\emptyset}{y},\ \Set{a}\in\cols{i}{y}
    &\Longleftrightarrow
    z=\backchop{}{a\emptyset}{z'},\ y=\backchop{i}{a\emptyset}{y'},\ a\notin\bigcup{z'},\\
    z'=\chop{}{ab}{z},\ y'=\chop{ij}{ab}{y},\ b\cup\Set{a}\in\cols{i}{y}
    &\Longleftrightarrow
    z=\backchop{}{ab}{z'},\ y=\backchop{ij}{ab}{y'},\ b\in\cols{j}{y'},\ a\notin\bigcup{z'},\\
    y'=\swap{ij}{b}{y},\ b\in\cols{i}{y}
    &\Longleftrightarrow
    y=\backswap{ij}{b}{y'},\ b\in\cols{j}{y'},\\
    y'=\fork{ijk}{cc'}{y},\ c\cup{c'}\in\cols{i}{y}
    &\Longleftrightarrow
    y=\backfork{ijk}{cc'}{y'},\ c\in\cols{j}{y'},\ c'\in\cols{k}{y'}.\\
  \end{aligned}
\end{equation}

\subsection{SIRS model} %% SIRS
\label{sec:sirs-example}

Despite its relative simplicity, it will be instructive to consider, as a first illustration of the application of the theory, the SIRS model discussed in \zcref{sec:examples}.
Let us deduce the filter equation (\zcref{eq:obsc-filter-reg,eq:obsc-filter-sing}) corresponding to this one-deme model.
Here, the state vector is $x=(S,I,R)$ and the KFE is
\begin{mathsize}{9pt}{10pt}
  \begin{equation*}
    \begin{split}
      \frac{\partial{v}}{\partial{t}}(S,I,R)=
      &\frac{\beta\,(S+1)\,(I-1)}{N}\,v(t,S+1,I-1,R)
      -\frac{\beta\,S\,I}{N}\,v(t,S,I,R)\\
      &+\gamma\,(I+1)\,v(t,S,I+1,R-1)
      -\gamma\,I\,v(t,S,I,R)\\
      &+\omega\,(R+1)\,v(t,S-1,I,R+1)
      -\omega\,R\,v(t,S,I,R),
    \end{split}
  \end{equation*}
\end{mathsize}%
where $N=S+I+R$ is the host population size.
Note that, though the theory allows for time-dependent event rates, this example is time-homogeneous.
This model has one deme and occupancy function $n(x)=I$.
There are four jump marks:
$\Jumps=\Set{\jump{Trans},\jump{Recov},\jump{Wane},\jump{Sample}}$,
corresponding to transmission, recovery, waning of immunity, and sampling, respectively.
\zcref[S]{tab:sirs-model-elements} gives $\alpha_u$, $r^u$, and the event type for each of these marks.
The relevant binomial ratios are
\begin{mathsize}{9pt}{10pt}
  \begin{equation}
    \label{eq:sirs-binratio}
    \binratio{n}{\ell}{r}{s}=
    \begin{cases}
      \tbinratio{I}{\ell}{2}{0}=\frac{(I-\ell)\,(I-\ell-1)}{I\,(I-1)}\,\indicator{I\ge\ell}, & u=\jump{Trans}, s=0, \\[1ex]
      \tbinratio{I}{\ell}{2}{1}=\frac{2\,(I-\ell)}{I\,(I-1)}\,\indicator{I\ge\ell}, & u=\jump{Trans}, s=1, \\[1ex]
      \tbinratio{I}{\ell}{2}{2}=\frac{2}{I\,(I-1)}\,\indicator{I\ge\ell}, & u=\jump{Trans}, s=2, \\[1ex]
      \tbinratio{I}{\ell}{0}{0}=\indicator{I\ge\ell}, & u=\jump{Recov}, s=0, \\[1ex]
      \tbinratio{I}{\ell}{0}{0}=\indicator{I\ge\ell}, & u=\jump{Wane}, s=0, \\[1ex]
      \tbinratio{I}{\ell}{1}{0}=\frac{I-\ell}{I}\,\indicator{I\ge\ell}, & u=\jump{Sample}, s=0, \\[1ex]
      \tbinratio{I}{\ell}{1}{1}=\frac{1}{I}\,\indicator{I\ge\ell}, & u=\jump{Sample}, s=1. \\[1ex]
    \end{cases}
  \end{equation}
\end{mathsize}
The compatibility indicator is given by
\begin{mathsize}{9pt}{10pt}
  \begin{equation}
    \label{eq:sirs-Qs}
    Q_u(y,y')=\begin{cases}
    1, & u=\jump{Sample}, \exists{a}\ \st\ y'=\chop{}{a\emptyset}{y},\\
    1, & u=\jump{Sample}, \exists{a,b}\ \st\ y'=\chop{}{ab}{y},\\
    1, & u=\jump{Trans}, \exists{b,b'}\ \st\ y'=\fork{}{bb'}{y},\\
    1, & u=\jump{Trans}, y'=y,\\
    1, & u=\jump{Trans}, \exists{b}\ \st\ y'=\swap{}{b}{y},\\
    1, & u=\jump{Recov}, y'=y,\\
    1, & u=\jump{Wane}, y'=y,\\
    0, & \text{otherwise}.
    \end{cases}
  \end{equation}
\end{mathsize}%
%% (Recall that the swap, chop, and fork operators $\swap{}{}{}$, $\chop{}{}{}$, $\fork{}{}{}$ are defined in \zcref{sec:swapnchop}).

\begin{table}
  \caption{
    \label{tab:sirs-model-elements}
    Elements of the SIRS model pertinent to the genealogy process.
    The population state is $x=(S,I,R)$ and there are four types of jumps, $x\mapsto{x'}=(S',I',R')$.
    For each of these, the table shows the rate ($\alpha_u$), production ($r^u$), and event type (\cf~\zcref{sec:event-types}).
  }
  \begin{tabular}{cccl}
    \hline\hline
    $u$ & $\alpha_u(t,x,x')$ \bigstrut & $r^u$ \bigstrut & Event type \\
    \hline
    $\jump{Trans}$  & $\frac{\beta S I}{N}\,\indicator{S'=S-1,I'=I+1}$ \bigstrut & 2 & pure birth \\
    $\jump{Recov}$  & $\gamma\,I\,\indicator{I'=I-1,R'=R+1}$           \bigstrut & 0 & pure death \\
    $\jump{Wane}$   & $\omega\,R\,\indicator{S'=S+1,R'=R-1}$           \bigstrut & 0 & neutral \\
    $\jump{Sample}$ & $\psi\,I\,\indicator{x'=x}$                      \bigstrut & 1 & pure sample \\
    \hline\hline
  \end{tabular}
\end{table}

Now, suppose we are given an obscured genealogy $(T,Z)$ and wish to compute its likelihood under the SIRS model.
The general form of the singular part of the filter equation (\zcref{eq:obsc-filter-sing}) is
\begin{equation*}
  w(t,x,y)=\sum_{uy'}{\int{\wt(t,x',y')\,\alpha_u(t,x',x)\,\phi_u(t,x',x,y',y)\,\dd{x'}}},\qquad t\in\event{Z}.
\end{equation*}
Recalling that a coloring has two components, \ie $y=(d,m)$, where $d=\col{y}$ is the deme, or color, and $m=\ctr{y}$ is the event indicator (\cf~\zcref{sec:genealogy}), we can make these explicit in writing $w=w(t,x,d,m)$.
In the case of the SIRS model, however, there being only one deme, the dependence of $w$ on $d$ is trivial.
Let us therefore rewrite the general equation as
\begin{equation}
  \label{eq:sirs5}
  w(t,x,m)=\sum_{um'}{\int{\wt(t,x',m')\,\alpha_u(t,x',x)\,\phi_u(t,x',x,m',m)\,\dd{x'}}},\qquad t\in\event{Z}.
\end{equation}

Only two kinds of genealogical events are possible in $Z$ under the SIRS model: samples and branch-points.
Samples, in turn, may be inline or terminal.
Let $B$ denote the set of branch times;
$S_0$, the set of times of terminal samples;
and $S_1$, the times of inline samples.
Then $\event{Z}=B\cup{S_0}\cup{S_1}$.
Applying the relevant cases of \zcref{eq:sirs-Qs} (\ie the first three), and using \zcref{eq:phidef}, we can rewrite \zcref{eq:sirs5} as
\begin{equation*}
  w(t,x,m)=\begin{cases}
  \displaystyle\int{\wt(t,x',m')\,\alpha_{\jump{Sample}}(t,x',x)\,\tbinratio{I}{\ell}{1}{0}\,\dd{x'}} & t\in{S_0}.\\[2ex]
  \displaystyle\int{\wt(t,x',m')\,\alpha_{\jump{Sample}}(t,x',x)\,\tbinratio{I}{\ell}{1}{1}\,\dd{x'}} & t\in{S_1},\\[2ex]
  \displaystyle\int{\wt(t,x',m')\,\alpha_{\jump{Trans}}(t,x',x)\,\tbinratio{I}{\ell}{2}{2}\,\dd{x'}} & t\in{B}.\\[2ex]
  \end{cases}\qquad
\end{equation*}
Defining $w(t,x)=\sum_m{w(t,x,m)}$, we have
\begin{equation}
  \label{eq:sirs6}
  w(t,x)=\begin{cases}
  \displaystyle\int{\wt(t,x')\,\alpha_{\jump{Sample}}(t,x',x)\,\tbinratio{I}{\ell}{1}{0}\,\dd{x'}} & t\in{S_0}.\\[2ex]
  \displaystyle\int{\wt(t,x')\,\alpha_{\jump{Sample}}(t,x',x)\,\tbinratio{I}{\ell}{1}{1}\,\dd{x'}} & t\in{S_1},\\[2ex]
  \displaystyle\int{\wt(t,x')\,\alpha_{\jump{Trans}}(t,x',x)\,\tbinratio{I}{\ell}{2}{2}\,\dd{x'}} & t\in{B}.\\[2ex]
  \end{cases}\qquad
\end{equation}
Finally, putting in the expressions from \zcref{eq:sirs-binratio,tab:sirs-model-elements} gives
\begin{mathsize}{9pt}{10pt}
  \begin{equation}
    \label{eq:sirs7}
    w(t,S,I,R)=
    \begin{cases}
      \wt(t,S,I,R)\,\psi\,(I-\ell),                   & t\in{S_0},\ I\ge\ell,\\[1ex]
      \wt(t,S,I,R)\,\psi,                             & t\in{S_1},\ I\ge\ell,\\[1ex]
      \wt(t,S+1,I-1,R)\,\frac{2\,\beta\,(S+1)}{I\,N}, & t\in{B},\ I\ge{\ell},\\[1ex]
      0                                               & I<\ell.
    \end{cases}
  \end{equation}
\end{mathsize}%

Let us now turn to the regular part of the filter equation (\zcref{eq:obsc-filter-reg}), which holds for $t\notin{B\cup{S_0}\cup{S_1}}$,
\ie in the intervals between events of $Z$.
Again, since there is only one deme for this model, we have
\begin{mathsize}{9pt}{10pt}
  \begin{equation}
    \label{eq:sirs1}
    \begin{split}
      \frac{\partial{w}}{\partial{t}}(t,x,m) =&
      \sum_{um'}{\int{w(t,x',m')\,\alpha_u(t,x',x)\,\phi_u(t,x',x,m',m)\,\dd{x'}}}\\
      &-\sum_{um'}{\int{w(t,x,m)\,\alpha_u(t,x,x')\,\pi_u(t,x,x',m,m')\,\dd{x'}}}.
    \end{split}
  \end{equation}
\end{mathsize}%
Using \zcref{eq:phidef,eq:sirs-binratio,eq:sirs-Qs}, we can evaluate the sums in \zcref{eq:sirs1} to obtain
\begin{mathsize}{9pt}{10pt}
  \begin{equation}
    \label{eq:sirs2}
    \begin{split}
      \frac{\partial{w}}{\partial{t}}(t,x,m)=&
      \int{w(t,x',m)\,\alpha_{\jump{Trans}}(t,x',x)\,\tbinratio{I}{\ell}{2}{0}\,\dd{x'}}
      +\sum_{b\in{Z_t}}{\int{w(t,x',\backswap{}{b}{m})\,\alpha_{\jump{Trans}}(t,x',x)\,\tbinratio{I}{\ell}{2}{1}\,\dd{x'}}}\\
      &+\int{w(t,x',m)\,\alpha_{\jump{Recov}}(t,x',x)\,\tbinratio{I}{\ell}{0}{0}\,\dd{x'}}
      +\int{w(t,x',m)\,\alpha_{\jump{Wane}}(t,x',x)\,\tbinratio{I}{\ell}{0}{0}\,\dd{x'}}\\
      &-\sum_{u}{\int{w(t,x,m)\,\alpha_u(t,x,x')\,\dd{x'}}}.
    \end{split}
  \end{equation}
\end{mathsize}%
Again, we sum over $m$, obtaining
\begin{mathsize}{9pt}{10pt}
  \begin{equation}
    \label{eq:sirs3}
    \begin{split}
      \frac{\partial{w}}{\partial{t}}(t,x)=&
      \int{w(t,x')\,\alpha_{\jump{Trans}}(t,x',x)\,\left[\tbinratio{I}{\ell}{2}{0}+\ell\,\tbinratio{I}{\ell}{2}{1}\right]\,\dd{x'}}\\
      &+\int{w(t,x')\,\alpha_{\jump{Recov}}(t,x',x)\,\tbinratio{I}{\ell}{0}{0}\,\dd{x'}}
      +\int{w(t,x')\,\alpha_{\jump{Wane}}(t,x',x)\,\tbinratio{I}{\ell}{0}{0}\,\dd{x'}}\\
      &-\sum_{u}{\int{w(t,x)\,\alpha_u(t,x,x')\,\dd{x'}}},
    \end{split}
  \end{equation}
\end{mathsize}%
Finally, inserting the expressions from \zcref{tab:sirs-model-elements} and \zcref{eq:sirs-binratio} leaves us with the following regular filter equation
\begin{mathsize}{9pt}{10pt}
  \begin{equation}
    \label{eq:sirs4}
    \begin{split}
      \frac{\partial{w}}{\partial{t}}(t,S,I,R) =
      &\frac{\beta\,(S+1)\,(I-1)}{N}\,\left(1-\frac{\tbinom{\ell(t)}{2}}{\tbinom{I}{2}}\right)\,w(t,S+1,I-1,R)
      -\frac{\beta\,S\,I}{N}\,w(t,S,I,R)\\
      &+\gamma\,(I+1)\,w(t,S,I+1,R-1)
      -\gamma\,I\,w(t,S,I,R)\\
      &+\omega\,(R+1)\,w(t,S-1,I,R+1)
      -\omega\,R\,w(t,S,I,R)\\
      &-\psi\,I\,w(t,S,I,R),
    \end{split}
  \end{equation}
\end{mathsize}%
which holds only for $I\ge\ell$.
For $I<\ell$, it follows from \zcref{eq:sirs7} that $w=0$.
Here we have used the Chu-Vandermonde identity (\zcref{eq:chu-vandermonde}) to rewrite the first term of \zcref{eq:sirs3}.
Note the presence in \zcref{eq:sirs4} of the decay term proportional to $\psi$.

\zcref[S]{eq:sirs4,eq:sirs7} are identical to those obtained using the more specialized approach of \citet{King2022}.
\zcref[S]{fig:sirs-example} shows some numerical results obtained by integrating this filter equation using the scheme of \zcref{sec:smc}.

\begin{figure}
  \begin{center}
\begin{knitrout}\small
\definecolor{shadecolor}{rgb}{0.969, 0.969, 0.969}\color{fgcolor}

{\centering \includegraphics[width=1\linewidth]{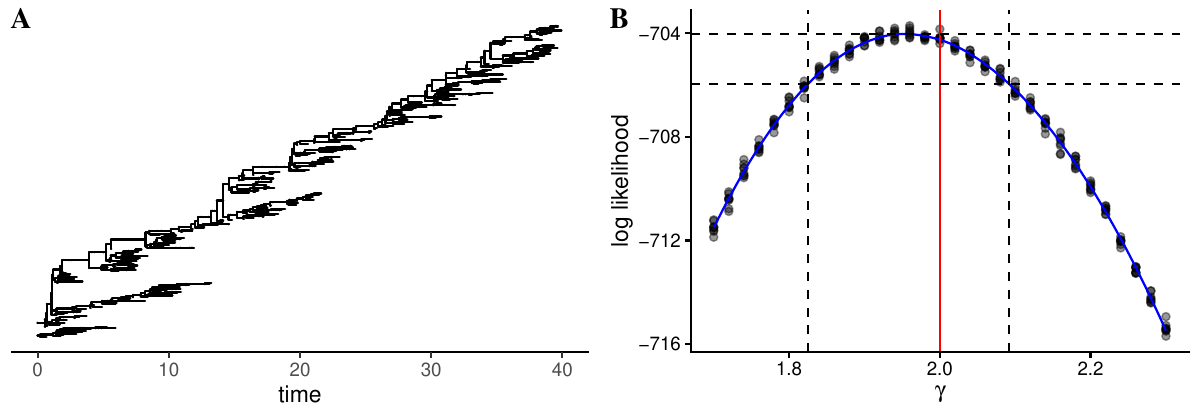} }

\end{knitrout}
  \end{center}
  \caption{
    \label{fig:sirs-example}
    Likelihood computation for the SIRS model by sequential Monte Carlo.
    \textbf{(A)}~A~simulated genealogy for $\beta=4$, $\gamma=2$, $\omega=1$, $\psi=1$, $(S_0,I_0,R_0)=(97, 3, 0)$.
    \textbf{(B)}~A~slice through the likelihood surface at the true parameters in the $\gamma$-direction.
    Each point is a distinct Monte Carlo estimate.
    The blue curve is a LOESS smooth;
    the dashed lines bound the Monte Carlo-adjusted 95\% confidence interval \citep{Ionides2017}.
  }
\end{figure}

\subsection{SEIRS model} %% SEIRS
\label{sec:seirs-example}

\def\ellE{\ell_{\deme{E}}}
\def\ellI{\ell_{\deme{I}}}

A simple model with more than one deme is the SEIRS model (\zcref{fig:example-models}A, \zcref{fig:markov-state}).
The state space is $\Zp^4$, with the state $x=(S,E,I,R)$ defined by the numbers of hosts in each of the four compartments.
The KFE for the population process is
\begin{mathsize}{9pt}{10pt}
  \begin{equation*}
    \begin{split}
      \frac{\partial{v}}{\partial{t}}(t,S,E,I,R)=
      &\frac{\beta\,(S+1)\,I}{N}\,v(t,S+1,E-1,I,R)
      -\frac{\beta\,S\,I}{N}\,v(t,S,E,I,R)\\
      &+\sigma\,(E+1)\,v(t,S,E+1,I-1,R)
      -\sigma\,E\,v(t,S,E,I,R)\\
      &+\gamma\,(I+1)\,v(t,S,E,I+1,R-1)
      -\gamma\,I\,v(t,S,E,I,R)\\
      &+\omega\,(R+1)\,v(t,S-1,E,I,R+1)
      -\omega\,R\,v(t,S,E,I,R),
    \end{split}
  \end{equation*}
\end{mathsize}%
where $N=S+E+I+R$ is the total population size.
Note that the terms associated with sampling cancel each other in the KFE, since, with the state space defined as above, sampling has no effect on the population state.

This model has two demes: $\Demes=\Set{\deme{E},\deme{I}}$.
Its deme occupancy function is $n(x)=(E,I)$.
There are five kinds of jumps:
transmission, progression, recovery, waning of immunity, and sampling:
the corresponding marks are $\Jumps=\Set{\jump{Trans},\jump{Prog},\jump{Recov},\jump{Wane},\jump{Sample}}$.
\zcref[S]{tab:seirs-model-elements} gives $\alpha_u$, $r^u$, and the event type for each of these marks.
The relevant binomial ratios are
\begin{mathsize}{9pt}{10pt}
  \begin{equation}
    \label{eq:seirs1}
    \binratio{n}{\ell}{r}{s}=
    \begin{cases}
      \tbinratio{E}{\ellE}{0}{0}\,\tbinratio{I}{\ellI}{1}{0}=
      \left(\frac{I-\ellI}{I}\right)\,\indicator{E\ge\ellE}\,\indicator{I\ge\ellI}, & u=\jump{Sample}, s=(0,0), \\[1ex]
      \tbinratio{E}{\ellE}{0}{0}\,\tbinratio{I}{\ellI}{1}{1}=
      \left(\frac{1}{I}\right)\,\indicator{E\ge\ellE}\,\indicator{I\ge\ellI}, & u=\jump{Sample}, s=(0,1), \\[1ex]
      \tbinratio{E}{\ellE}{1}{0}\,\tbinratio{I}{\ellI}{1}{0}=
      \left(\frac{E-\ellE}{E}\right)\left(\frac{I-\ellI}{I}\right)\,\indicator{E\ge\ellE}\,\indicator{I\ge\ellI}, & u=\jump{Trans}, s=(0,0), \\[1ex]
      \tbinratio{E}{\ellE}{1}{0}\,\tbinratio{I}{\ellI}{1}{1}=
      \left(\frac{E-\ellE}{E}\right)\left(\frac{1}{I}\right)\,\indicator{E\ge\ellE}\,\indicator{I\ge\ellI}, & u=\jump{Trans}, s=(0,1), \\[1ex]
      \tbinratio{E}{\ellE}{1}{1}\,\tbinratio{I}{\ellI}{1}{0}=
      \left(\frac{1}{E}\right)\left(\frac{I-\ellI}{I}\right)\,\indicator{E\ge\ellE}\,\indicator{I\ge\ellI}, & u=\jump{Trans}, s=(1,0), \\[1ex]
      \tbinratio{E}{\ellE}{1}{1}\,\tbinratio{I}{\ellI}{1}{1}=
      \left(\frac{1}{E}\right)\left(\frac{1}{I}\right)\,\indicator{E\ge\ellE}\,\indicator{I\ge\ellI}, & u=\jump{Trans}, s=(1,1), \\[1ex]
      \tbinratio{E}{\ellE}{0}{0}\,\tbinratio{I}{\ellI}{1}{0}=
      \left(\frac{I-\ellI}{I}\right)\,\indicator{E\ge\ellE}\,\indicator{I\ge\ellI}, & u=\jump{Prog}, s=(0,0), \\[1ex]
      \tbinratio{E}{\ellE}{0}{0}\,\tbinratio{I}{\ellI}{1}{1}=
      \left(\frac{1}{I}\right)\,\indicator{E\ge\ellE}\,\indicator{I\ge\ellI}, & u=\jump{Prog}, s=(0,1), \\[1ex]
      \tbinratio{E}{\ellE}{0}{0}\,\tbinratio{I}{\ellI}{0}{0}=
      \indicator{E\ge\ellE}\,\indicator{I\ge\ellI}, & u=\jump{Recov}, s=(0,0), \\[1ex]
      \tbinratio{E}{\ellE}{0}{0}\,\tbinratio{I}{\ellI}{0}{0}=
      \indicator{E\ge\ellE}\,\indicator{I\ge\ellI}, & u=\jump{Wane}, s=(0,0). \\[1ex]
    \end{cases}
  \end{equation}
\end{mathsize}
As for the compatibility indicator, it is given by
\begin{mathsize}{9pt}{10pt}
  \begin{equation}
    \label{eq:seirs2}
    Q_u(y,y') = \begin{cases}
      1, & u=\jump{Sample}, \exists{a}\ \st\ y'=\chop{\deme{I}}{a\emptyset}{y},\\
      1, & u=\jump{Sample}, \exists{a,b}\ \st\ y'=\chop{\deme{II}}{ab}{y},\\
      1, & u=\jump{Trans}, \exists{b,b'}\ \st\ y'=\fork{\deme{IIE}}{bb'}{y},\\
      1, & u=\jump{Trans}, y'=y,\\
      1, & u=\jump{Trans}, \exists{b}\ \st\ y'=\swap{\deme{II}}{b}{y},\\
      1, & u=\jump{Trans}, \exists{b}\ \st\ y'=\swap{\deme{IE}}{b}{y},\\
      1, & u=\jump{Prog}, y'=y,\\
      1, & u=\jump{Prog}, \exists{b}\ \st\ y'=\swap{\deme{EI}}{b}{y},\\
      1, & u=\jump{Recov}, y'=y,\\
      1, & u=\jump{Wane}, y'=y,\\
      0, & \text{otherwise}.\\
    \end{cases}
  \end{equation}
\end{mathsize}%

Now, writing $d=\col{y}$ and $m=\ctr{y}$ so that $y=(d,m)$ and $w=w(t,x,d,m)$, the singular portion of the filter equation (\zcref{eq:obsc-filter-reg}) has general form
\begin{mathsize}{9pt}{10pt}
  \begin{equation}
    \label{eq:seirs21}
    w(t,x,d,m)=\sum_{ud'm'}{\int{\wt(t,x',d',m')\,\alpha_u(t,x',x)\,\phi_u(t,x',x,d',d,m',m)\,\dd{x'}}},\qquad t\in\event{Z}.
  \end{equation}
\end{mathsize}%
Applying \zcref{eq:phidef} and the relevant cases of \zcref{eq:seirs1,eq:seirs2}, we expand the sums to obtain
\begin{mathsize}{9pt}{10pt}
  \begin{equation*}
    \begin{aligned}
      w(t,x,d,m)=\begin{cases}
      \displaystyle\int{\wt(t,x',\backchop{\deme{I}}{a\emptyset}{d},\backchop{}{a\emptyset}{m})\,\alpha_{\jump{Sample}}(t,x',x)\,\tbinratio{E}{\ellE}{0}{0}\,\tbinratio{I}{\ellI}{1}{0}\,\dd{x'}}, & t\in{S_0},\\
      \displaystyle\int{\wt(t,x',\backchop{\deme{II}}{ab}{d},\backchop{}{ab}{m})\,\alpha_{\jump{Sample}}(t,x',x)\,\tbinratio{E}{\ellE}{0}{0}\,\tbinratio{I}{\ellI}{1}{1}\,\dd{x'}}, & t\in{S_1},\\
      \displaystyle\int{\wt(t,x',\backfork{\deme{IIE}}{bb'}{d},\backfork{}{bb'}{m})\,\alpha_{\jump{Trans}}(t,x',x)\,\tbinratio{E}{\ellE}{1}{1}\,\tbinratio{I}{\ellI}{1}{1}\,\dd{x'}}\\
      \quad+\displaystyle\int{\wt(t,x',\backfork{\deme{IIE}}{b'b}{d},\backfork{}{b'b}{m})\,\alpha_{\jump{Trans}}(t,x',x)\,\tbinratio{E}{\ellE}{1}{1}\,\tbinratio{I}{\ellI}{1}{1}\,\dd{x'}}, & t\in{B}.\\
      \end{cases}
    \end{aligned}
  \end{equation*}
\end{mathsize}%
Note that, in \zcref{eq:seirs21}, the sample label $a$, and the branches $b$ and $b'$ are determined unambiguously by $Z$.
As in \zcref{sec:sirs-example}, we sum over $m$, writing $w(t,x,d)=\sum_m{w(t,x,d,m)}$ to obtain
\begin{mathsize}{9pt}{10pt}
  \begin{equation}
    \label{eq:seirs22}
    \begin{aligned}
      w(t,x,d)=\begin{cases}
      \displaystyle\int{\wt(t,x',\backchop{\deme{I}}{a\emptyset}{d})\,\alpha_{\jump{Sample}}(t,x',x)\,\tbinratio{E}{\ellE}{0}{0}\,\tbinratio{I}{\ellI}{1}{0}\,\dd{x'}}, & t\in{S_0},\\
      \displaystyle\int{\wt(t,x',\backchop{\deme{II}}{ab}{d})\,\alpha_{\jump{Sample}}(t,x',x)\,\tbinratio{E}{\ellE}{0}{0}\,\tbinratio{I}{\ellI}{1}{1}\,\dd{x'}}, & t\in{S_1},\\
      \displaystyle\int{\wt(t,x',\backfork{\deme{IIE}}{bb'}{d})\,\alpha_{\jump{Trans}}(t,x',x)\,\tbinratio{E}{\ellE}{1}{1}\,\tbinratio{I}{\ellI}{1}{1}\,\dd{x'}}\\[1.5ex]
      \quad+\displaystyle\int{\wt(t,x',\backfork{\deme{IIE}}{b'b}{d})\,\alpha_{\jump{Trans}}(t,x',x)\,\tbinratio{E}{\ellE}{1}{1}\,\tbinratio{I}{\ellI}{1}{1}\,\dd{x'}}, & t\in{B}.\\
      \end{cases}
    \end{aligned}
  \end{equation}
\end{mathsize}%
Finally, substituting the expressions from \zcref{eq:seirs1}, we have the following singular equation:
\begin{mathsize}{9pt}{10pt}
  \begin{equation}
    \label{eq:seirs23}
    \begin{aligned}
      w(t,x,d)=\begin{cases}
      \displaystyle\int{\wt(t,x',\backchop{\deme{I}}{a\emptyset}{d})\,\alpha_{\jump{Sample}}(t,x',x)\,\indicator{I\ge\ellI}\,\indicator{E\ge\ellE}\,\left(\tfrac{I-\ellI}{I}\right)\,\dd{x'}}, & t\in{S_0},\\
      \displaystyle\int{\wt(t,x',\backchop{\deme{II}}{ab}{d})\,\alpha_{\jump{Sample}}(t,x',x)\,\indicator{I\ge\ellI}\,\indicator{E\ge\ellE}\,\left(\tfrac{1}{I}\right)\,\dd{x'}}, & t\in{S_1},\\
      \displaystyle\int{\wt(t,x',\backfork{\deme{IIE}}{bb'}{d})\,\frac{\pi_{\jump{Trans}}^{bb'}(t,x')}{\pi_{\jump{Trans}}^{bb'}(t,x')}\,\alpha_{\jump{Trans}}(t,x',x)\,\indicator{I\ge\ellI}\,\indicator{E\ge\ellE}\,\left(\tfrac{1}{EI}\right)\,\dd{x'}}\\[1.5ex]
      \quad+\displaystyle\int{\wt(t,x',\backfork{\deme{IEI}}{bb'}{d})\,\frac{\pi_{\jump{Trans}}^{b'b}(t,x')}{\pi_{\jump{Trans}}^{b'b}(t,x')}\,\alpha_{\jump{Trans}}(t,x',x)\,\indicator{I\ge\ellI}\,\indicator{E\ge\ellE}\,\left(\tfrac{1}{EI}\right)\,\dd{x'}}, & t\in{B}.\\
      \end{cases}
    \end{aligned}
  \end{equation}
\end{mathsize}%
Here, we have introduced the importance sampling weights $\pi_{\jump{Trans}}^{bb'}$, normalized so that $\pi_{\jump{Trans}}^{bb'}(t,x)+\pi_{\jump{Trans}}^{b'b}(t,x)=1$ for all $t$ and $x$.

Turning now to the filter equation's regular part, \zcref{eq:obsc-filter-reg} is equivalent to
\begin{mathsize}{9pt}{10pt}
  \begin{equation}
    \label{eq:seirs3}
    \begin{split}
      \frac{\partial{w}}{\partial{t}}=
      &\sum_{ud'm'}{\int{w(t,x',d',m')\,\alpha_u(t,x',x)\,\phi_u(t,x',x,d',m',d,m)\,\dd{x'}}}\\
      &-\sum_{ud'm'}{\int{w(t,x,d,m)\,\alpha_u(t,x,x')\,\pi_u(t,x,x',d,m,d',m')\,\dd{x'}}}.
    \end{split}
  \end{equation}
\end{mathsize}%
Substituting the expressions from \zcref{eq:phidef,eq:seirs1,eq:seirs2} into \zcref{eq:seirs3} and evaluating the sums as in \zcref{sec:sirs-example} yields
\begin{mathsize}{9pt}{10pt}
  \begin{equation*}
    \label{eq:seirs4}
    \begin{split}
      \frac{\partial{w}}{\partial{t}}(t,x,d,m)=
      &\int{w(t,x',d,m)\,\alpha_{\jump{Trans}}(t,x',x)\,\tbinratio{E}{\ellE}{1}{0}\,\tbinratio{I}{\ellI}{1}{0}\,\dd{x'}}\\
      &+\sum_{b\in\cols{\deme{I}}{d}}{\int{w(t,x',d,\backswap{}{b}{m})\,\alpha_{\jump{Trans}}(t,x',x)\,\tbinratio{E}{\ellE}{1}{0}\,\tbinratio{I}{\ellI}{1}{1}\,\dd{x'}}}\\
      &+\sum_{b\in\cols{\deme{E}}{d}}{\int{w(t,x',\backswap{\deme{IE}}{b}{d},\backswap{}{b}{m})\,\alpha_{\jump{Trans}}(t,x',x)\,\tbinratio{E}{\ellE}{1}{1}\,\tbinratio{I}{\ellI}{1}{0}\,\dd{x'}}}\\
      &+\int{w(t,x',d,m)\,\alpha_{\jump{Prog}}(t,x',x)\,\,\tbinratio{E}{\ellE}{0}{0}\,\tbinratio{I}{\ellI}{1}{0}\,\dd{x'}}\\
      &+\sum_{b\in\cols{\deme{I}}{d}}{\int{w(t,x',\backswap{\deme{EI}}{b}{d},\backswap{}{b}{m})\,\alpha_{\jump{Prog}}(t,x',x)\,\tbinratio{E}{\ellE}{0}{0}\,\tbinratio{I}{\ellI}{1}{1}\,\dd{x'}}}\\
      &+\int{w(t,x',d,m)\,\alpha_{\jump{Recov}}(t,x',x)\,\tbinratio{E}{\ellE}{0}{0}\,\tbinratio{I}{\ellI}{0}{0}\,\dd{x'}}\\
      &+\int{w(t,x',d,m)\,\alpha_{\jump{Wane}}(t,x',x)\,\tbinratio{E}{\ellE}{0}{0}\,\tbinratio{I}{\ellI}{0}{0}\,\dd{x'}}\\
      &-\sum_{u}{\int{w(t,x,d,m)\,\alpha_u(t,x,x')\,\dd{x'}}}.
    \end{split}
  \end{equation*}
\end{mathsize}%
As in \zcref{sec:sirs-example}, we then sum over $m$, obtaining
\begin{mathsize}{9pt}{10pt}
  \begin{equation*}
    \begin{split}
      \frac{\partial{w}}{\partial{t}}(t,x,d)=
      &\int{w(t,x',d)\,\alpha_{\jump{Trans}}(t,x',x)\,\left[\tbinratio{E}{\ellE}{1}{0}\,\tbinratio{I}{\ellI}{1}{0}+\ellI\,\tbinratio{E}{\ellE}{1}{0}\,\tbinratio{I}{\ellI}{1}{1}\right]\,\dd{x'}}\\
      &+\sum_{b\in\cols{\deme{E}}{d}}{\int{w(t,x',\backswap{\deme{IE}}{b}{d})\,\alpha_{\jump{Trans}}(t,x',x)\,\tbinratio{E}{\ellE}{1}{1}\,\tbinratio{I}{\ellI}{1}{0}\,\dd{x'}}}\\
      &+\int{w(t,x',d)\,\alpha_{\jump{Prog}}(t,x',x)\,\,\tbinratio{E}{\ellE}{0}{0}\,\tbinratio{I}{\ellI}{1}{0}\,\dd{x'}}\\
      &+\sum_{b\in\cols{\deme{I}}{d}}{\int{w(t,x',\backswap{\deme{EI}}{b}{d})\,\alpha_{\jump{Prog}}(t,x',x)\,\tbinratio{E}{\ellE}{0}{0}\,\tbinratio{I}{\ellI}{1}{1}\,\dd{x'}}}\\
      &+\int{w(t,x',d)\,\alpha_{\jump{Recov}}(t,x',x)\,\tbinratio{E}{\ellE}{0}{0}\,\tbinratio{I}{\ellI}{0}{0}\,\dd{x'}}\\
      &+\int{w(t,x',d)\,\alpha_{\jump{Wane}}(t,x',x)\,\tbinratio{E}{\ellE}{0}{0}\,\tbinratio{I}{\ellI}{0}{0}\,\dd{x'}}\\
      &-\sum_{u}{\int{w(t,x,d)\,\alpha_u(t,x,x')\,\dd{x'}}}.
    \end{split}
  \end{equation*}
\end{mathsize}%
Simplifying using \zcref{eq:seirs1,eq:chu-vandermonde,eq:swapnchop} and splitting the final sum over $u$, we obtain, for $E\ge{\ellE}$, $I\ge{\ellE}$,
\begin{mathsize}{9pt}{10pt}
  \begin{equation}
    \label{eq:seirs5}
    \begin{split}
      \frac{\partial{w}}{\partial{t}}(t,x,d)=
      &\int{w(t,x',d)\,\alpha_{\jump{Trans}}(t,x',x)\,\frac{\pi_{\jump{Trans}}^{\emptyset}(t,x')}{\pi_{\jump{Trans}}^{\emptyset}(t,x')}\,\left(\tfrac{E-\ellE}{E}\right)\,\dd{x'}}
      -\int{w(t,x,d)\,\alpha_{\jump{Trans}}(t,x,x')\,\pi_{\jump{Trans}}^{\emptyset}(t,x)\,\dd{x'}}\\
      &+\sum_{bd'}\left\{\int{w(t,x',d')\,\alpha_{\jump{Trans}}(t,x',x)\,\frac{\pi_{\jump{Trans}}^b(t,x')}{\pi_{\jump{Trans}}^b(t,x')}\,\indicator{d=\swap{\deme{IE}}{b}{d'}}\,\left(\tfrac{I-\ellI}{EI}\right)\,\dd{x'}}\right.\\
      &\qquad\qquad\left.-\int{w(t,x,d)\,\alpha_{\jump{Trans}}(t,x,x')\,\pi_{\jump{Trans}}^b(t,x)\,\indicator{d'=\swap{\deme{IE}}{b}{d}}\,\dd{x'}}\right\}\\
      &+\int{w(t,x',d)\,\alpha_{\jump{Prog}}(t,x',x)\,\frac{\pi_{\jump{Prog}}^{\emptyset}(t,x')}{\pi_{\jump{Prog}}^{\emptyset}(t,x')}\,\left(\tfrac{I-\ellI}{I}\right)\,\dd{x'}}
      -\int{w(t,x,d)\,\alpha_{\jump{Prog}}(t,x,x')\,\pi_{\jump{Prog}}^{\emptyset}(t,x)\,\dd{x'}}\\
      &+\sum_{bd'}\left\{\int{w(t,x',d')\,\alpha_{\jump{Prog}}(t,x',x)\,\frac{\pi_{\jump{Prog}}^b(t,x')}{\pi_{\jump{Prog}}^b(t,x')}\,\indicator{d=\swap{\deme{EI}}{b}{d'}}\,\left(\tfrac{1}{I}\right)\,\dd{x'}}\right.\\
      &\qquad\qquad\left.-\int{w(t,x,d)\,\alpha_{\jump{Prog}}(t,x,x')\,\pi_{\jump{Prog}}^b(t,x)\,\indicator{d'=\swap{\deme{EI}}{b}{d}}\,\dd{x'}}\right\}\\
      &+\int{w(t,x',d)\,\alpha_{\jump{Recov}}(t,x',x)\,\indicator{I\ge\ellI}\,\dd{x'}}
      -\int{w(t,x,d)\,\alpha_{\jump{Recov}}(t,x,x')\,\indicator{I\ge\ellI}\,\dd{x'}}\\
      &+\int{w(t,x',d)\,\alpha_{\jump{Wane}}(t,x',x)\,\dd{x'}}
      -\int{w(t,x,d)\,\alpha_{\jump{Wane}}(t,x,x')\,\dd{x'}}\\
      &-\int{w(t,x,d)\,\alpha_{\jump{Sample}}(t,x,x')\,\dd{x'}}
      -\int{w(t,x,d)\,\alpha_{\jump{Recov}}(t,x,x')\,\indicator{I<\ellI}\,\dd{x'}}.
    \end{split}
  \end{equation}
\end{mathsize}%
Here, we have expanded the sums to range over all branches present at time $t$, using the indicator functions to nullify the extra terms.
The kernels $\pi$ that appear in \zcref{eq:seirs5} are normalized so that, for all $t,x$,
\begin{equation}
  \begin{gathered}
    \pi_{\jump{Trans}}^{\emptyset}(t,x)+\sum_b{\pi_{\jump{Trans}}^{b}(t,x)}=1,
    \qquad
    \pi_{\jump{Prog}}^{\emptyset}(t,x)+\sum_b{\pi_{\jump{Prog}}^{b}(t,x)}=1.
  \end{gathered}
\end{equation}

\begin{table}
  \caption{
    \label{tab:seirs-model-elements}
    Elements of the SEIRS model pertinent to the genealogy process.
    The population state is $x=(S,E,I,R)$ and there are four types of jumps, $x\mapsto{x'}=(S',E',I',R')$.
    For each of these, the table shows the rate ($\alpha_u$), production ($r^u=(r^u_{\deme{E}},r^u_{\deme{I}})$), and event type.
  }
  \begin{tabular}{cccl}
    \hline\hline
    $u$ & $\alpha_u(t,x,x')$ \bigstrut & $r^u$ \bigstrut & Event type \\
    \hline
    $\jump{Trans}$  & $\frac{\beta S I}{N}\,\indicator{S'=S-1,E'=E+1}$ \bigstrut & $(1,1)$ & pure birth \\
    $\jump{Prog}$   & $\sigma\,E\,\indicator{E'=E-1,I'=I+1}$           \bigstrut & $(0,1)$ & pure migration \\
    $\jump{Recov}$  & $\gamma\,I\,\indicator{I'=I-1,R'=R+1}$           \bigstrut & $(0,0)$ & pure death \\
    $\jump{Wane}$   & $\omega\,R\,\indicator{S'=S+1,R'=R-1}$           \bigstrut & $(0,0)$ & neutral \\
    $\jump{Sample}$ & $\psi\,I\,\indicator{x'=x}$                      \bigstrut & $(0,1)$ & pure sample \\
    \hline\hline
  \end{tabular}
\end{table}

Comparing \zcref{eq:seirs5,eq:seirs21} with the general form of a filter equation (\zcref{sec:filter-eqns}, \zcref{eq:filter-eq-defn-reg,eq:filter-eq-defn-sing}), one recognizes them as a filter equation on the $(x,u,d)$-space, with driver $\beta$, boost $B$, and decay $\lambda$, as follows.
\begin{align}
  \begin{split}
    \label{eq:seirs-driver}
    \beta(t,x,&x',u,u',d,d')=
    \alpha_{\jump{Trans}}(t,x,x')\,\left(\pi_{\jump{Trans}}^{\emptyset}(t,x)\,\indicator{d'=d}+\sum_{b}{\pi_{\jump{Trans}}^b(t,x)\,\indicator{d'=\swap{\deme{IE}}{b}{d}}}\right)\,\indicator{u'=\jump{Trans}}\\
    &+\alpha_{\jump{Prog}}(t,x,x')\,\left(\pi_{\jump{Prog}}^{\emptyset}(t,x)\,\indicator{d'=d}+\sum_{b}{\pi_{\jump{Prog}}^b(t,x)\,\indicator{d'=\swap{\deme{EI}}{b}{d}}}\right)\,\indicator{u'=\jump{Prog}}\\
    &+\alpha_{\jump{Recov}}(t,x,x')\,\indicator{I>\ellI}\,\indicator{d'=d}\,\indicator{u'=\jump{Recov}}
    +\alpha_{\jump{Wane}}(t,x,x')\,\indicator{d'=d}\,\indicator{u'=\jump{Wane}}\\
    &+\indicator{u'=\jump{Sample}}\,\sum_{e\in{S_0}}{\indicator{b_e\in\cols{\deme{I}}{d}}\,\indicator{d'=\chop{\deme{I}}{b_e}{d}}\,\delta(e,t)}
    +\indicator{u'=\jump{Sample}}\,\sum_{e\in{S_1}}{\indicator{b_e\in\cols{\deme{I}}{d}}\,\indicator{d'=\chop{\deme{II}}{a_eb_e}{d}}\,\delta(e,t)}\\
    &+\indicator{u'=\jump{Trans}}\,\sum_{e\in{B}}{\indicator{b_e\cup{b'_e}\in\cols{\deme{I}}{d}}\,\left[\pi_{\jump{Trans}}^{b_eb_e'}(t,x)\,\indicator{d'=\fork{\deme{IIE}}{b_eb_e'}{d}}+\pi_{\jump{Trans}}^{b_e'b_e}(t,x)\,\indicator{d'=\fork{\deme{IIE}}{b_e'b_e}{d}}\right]\,\delta(e,t)}.\\
  \end{split}\\
  \begin{split}
    \label{eq:seirs-boost}
    B(t,x,&x',u,u',d,d')=
    \indicator{u'=\jump{Trans}}\,\left(
    \tfrac{E'-\ellE'}{E'}\,\tfrac{\Indicator{d'=d}}{\pi_{\jump{Trans}}^{\emptyset}(t,x)}
    +\tfrac{I'-\ellI'}{E'I'}\,\sum_b{\tfrac{\Indicator{d'=\swap{\deme{IE}}{b}{d}}}{\pi_{\jump{Trans}}^b(t,x)}}
    \right)\\
    &+\indicator{u'=\jump{Prog}}\,\left(
    \tfrac{I'-\ellI'}{I'}\,\tfrac{\Indicator{d'=d}}{\pi_{\jump{Prog}}^{\emptyset}(t,x)}
    +\tfrac{1}{I'}\,\sum_b{\tfrac{\Indicator{d'=\swap{\deme{EI}}{b}{d}}}{\pi_{\jump{Prog}}^b(t,x)}}
    \right)
    +\indicator{u'=\jump{Recov}}
    +\indicator{u'=\jump{Wane}}\\
    &+\alpha_{\jump{Sample}}(t,x,x')\,\indicator{u'=\jump{Sample}}\,\left(\sum_{e\in{S_0}}{\indicator{b_e\in\cols{\deme{I}}{d}}\,\left(\tfrac{I'-\ellI'}{I'}\right)\,\indicator{t=e}}
    +\sum_{e\in{S_1}}{\indicator{b_e\in\cols{\deme{I}}{d}}\,\tfrac{1}{I'}\,\indicator{t=e}}\right)\\
    &+\alpha_{\jump{Trans}}(t,x,x')\,\indicator{u'=\jump{Trans}}\,\sum_{e\in{B}}{\indicator{b_e\cup{b'_e}\in\cols{\deme{I}}{d}}\,\tfrac{1}{E'I'}\,\left[\tfrac{\Indicator{d'=\fork{\deme{IIE}}{b_eb_e'}{d}}}{\pi_{\jump{Trans}}^{b_eb_e'}(t,x)}+\tfrac{\Indicator{d'=\fork{\deme{IIE}}{b_e'b_e}{d}}}{\pi_{\jump{Trans}}^{b_e'b_e}(t,x)}\right]\,\indicator{t=e}}.\\
  \end{split}\\
  \begin{split}
    \label{eq:seirs-decay}
    \lambda(t,x,&d)=
    \int{\alpha_{\jump{Sample}}(t,x,x')\,\dd{x'}}
    +\int{\alpha_{\jump{Recov}}(t,x,x')\,\indicator{I\le\ellI}\,\dd{x'}}.\\
  \end{split}
\end{align}
One sees, in particular, that the importance-sampling kernels $\pi$ represent the probabilities of changing color on a particular branch (or not doing so), contingent on the population event.
For example, $\pi_{\jump{Prog}}^{b}(t,x)$ is the probability of proposing an $\deme{E}\to\deme{I}$ change on branch $b$ at time $t$ given state $x$, given that a $\jump{Prog}$ event occurs at that time.
The complementary probability $\pi_{\jump{Prog}}^{\emptyset}(t,x)$ is that of proposing no change to any branch.
As another example $\pi_{\jump{Trans}}^{bb'}(t,x)$ is the probability of proposing, at a genealogical branch-point, that branch $b$ is put into deme $\deme{I}$ and $b'$ is put into $\deme{E}$;
$\pi_{\jump{Trans}}^{b'b}$ is the probability of the alternative.
One has wide latitude in the choice of these kernels and by tuning them, one can attempt to minimize Monte Carlo variance.
In particular, one can exploit the permitted time-dependence of the kernels $\pi$ so as to ``borrow'' information from future events.
Thus, for example, since in the SEIR model, branch points can only occur when the parent is in the $\deme{I}$ deme, one might posit $\pi_{\jump{Prog}}$ and $\pi_{\jump{Trans}}$ in such a way as to drive the color along a branch into $\deme{I}$ as it approaches a branch point.

For illustrative purposes, we here use the simpler, \emph{na\"ive} choice, which is non-anticipatory and always available.
The na\"ive choice consists of simply proposing changes in proportion to their relative abundance.
Thus, in this case, $\pi_{\jump{Prog}}^b=\tfrac{1}{E}$ for all $b\in\cols{\deme{E}}{y}$ and $\pi_{\jump{Prog}}^{\emptyset}=\tfrac{E-\ellE}{E}$.
Similarly, $\pi_{\jump{Trans}}^b=\tfrac{1}{I}$ and $\pi_{\jump{Trans}}^{\emptyset}=\tfrac{I-\ellI}{I}$.
At branch points, the two possibilities are given equal probabilities, \ie $\pi_{\jump{Trans}}^{bb'}=\pi_{\jump{Trans}}^{b'b}=\tfrac{1}{2}$.
\zcref[S]{fig:seirs-example} shows the results of a numerical calculation based on these choices.

With the na\"ive choice of importance-sampling kernels, it can happen that the coloring process puts a lineage into a deme that is inconsistent with the data.
For example, at a branch point or at a sample, the parent lineage might be in deme $\deme{E}$.
Observe that, in this case, according to \zcref{eq:seirs-boost}, the boost factor will be $0$, which is the correct penalty to apply.
Alternatively, a different choice of $\pi$ that anticipates the necessary deme-location, might avoid such extreme penalties and thereby prove more efficient.

\begin{figure}
  \begin{center}
\begin{knitrout}\small
\definecolor{shadecolor}{rgb}{0.969, 0.969, 0.969}\color{fgcolor}

{\centering \includegraphics[width=1\linewidth]{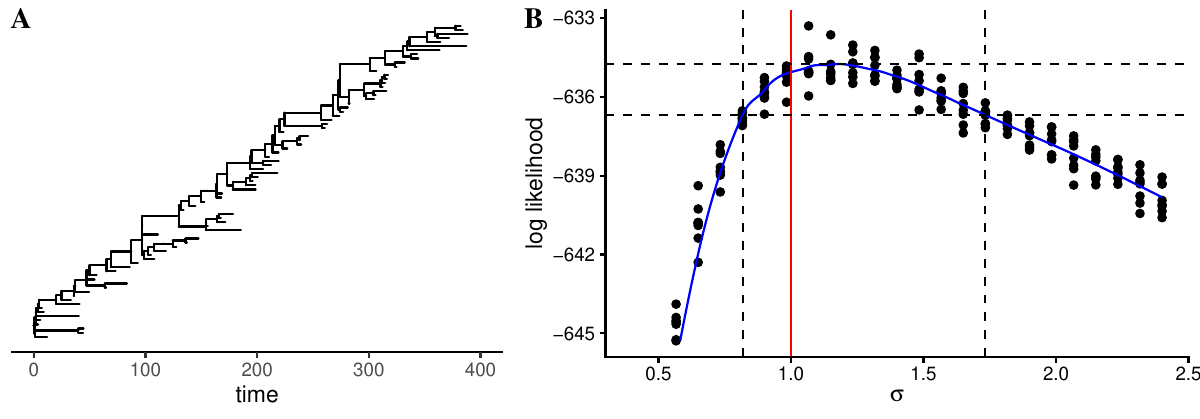} }

\end{knitrout}
  \end{center}
  \caption{
    \label{fig:seirs-example}
    Likelihood computation for the SEIRS model by sequential Monte Carlo, using the scheme of \zcref{sec:smc}.
    \textbf{(A)}~Simulated genealogy for $\beta=3$,
    $\sigma=1$,
    $\gamma=0.5$,
    $\psi=0.02$,
    $\omega=0.08$,
    $(S_0,E_0,I_0,R_0)=(70, 1, 0, 50)$.
    \textbf{(B)}~Likelihood slice in the $\sigma$-direction.
    Each point represents the estimate of an independent sequential Monte Carlo computation.
    The blue curve shows a LOESS smooth;
    the dashed vertical lines enclose the Monte Carlo-adjusted 95\% confidence interval \citep{Ionides2017}.
  }
\end{figure}

\section{Discussion}
\label{sec:discussion}

The theory presented here represents a generalization of the existing coalescent and birth-death-process approaches to phylodynamic inference.
Importantly, it allows computation of the likelihood via strictly forward-in-time computations and permits consideration of models with discrete population structure.
Moreover, inasmuch as the formulae of \zcref{thm:obsc-filter} can be efficiently computed via sequential Monte Carlo, explicit expressions for transition probabilities are not needed:
it is sufficient to be able to simulate from the population process.
This feature of the algorithms---known as the \emph{plug-and-play property} \citep{He2010}---further expands the class of population models that can be confronted with data.

Specifically, the theory gives us the freedom to choose models with many demes.
\citet{Volz2012} and \citet{Rasmussen2014a} have shown a path to phylodynamic inference in the presence of discrete population structure of the kind we consider here.
Their procedures involve a structured-coalescent approximation whereby deme-specific coalescent rates obey a set of differential equations solved backward in time.
The theory presented here avoids the need for such approximations, though in practice whether the exact results are worth their computational cost will depend on both the questions asked and the data.
In particular, even if cheaper approximations are used in data analysis, an exact calculation is invaluable in that it allows one to evaluate the accuracy of proposed approximations, which might otherwise be difficult to establish.

Some existing methods put rather severe limits on the form of the sampling model and, as \citet{Volz2014a} pointed out, misspecification of the sampling model can lead to large inferential biases.
With the theory presented here, essentially arbitrary specification of the sampling model is possible.
In particular, one can posit sampling at a rate which is an arbitrary function of time and state and include discrete sampling events as well.
It is also possible to condition on the existence of samples.

If sequential Monte Carlo algorithms are used to compute the likelihoods of \zcref{thm:obsc-filter}, then it is straightforward to simultaneously assimilate information from both time-series and genealogical data.
One can therefore supplement traditional incidence, disease, or mortality time series with genealogical data to improve inference.

A limitation of the theory is that the population models are assumed to be pure jump processes, which allows consideration of demographic stochasticity and environmental stochasticity modeled by jumps involving multiple individuals \citep{Breto2011}, but disallows stochastic processes with a diffusive component.
However, we anticipate that it will be possible to incorporate the full range of Markovian environmental stochasticity via extension of this theory to population models containing both diffusion and jump components.

Similarly, the theory presented here assumes that at most one birth event can occur at a time and that migration and sampling events involve at most one source deme at a time (though multiple migrators or samples are permitted).
These restrictions are not essential, and the proof of \zcref{thm:pruned-lik} can be adapted to accommodate relaxations of these assumptions.
This will be developed in a sequel.

The price of the theory's flexibility is primarily computational.
When sequential Monte Carlo is used to evaluate the likelihood in \zcref{thm:obsc-filter}, the computational complexity scales linearly with the number of genomic samples.
In its most straightforward implementation---using an event-driven algorithm \citep[\eg][]{Gillespie1977a}---it scales super-linearly with population size in general.
However, stochastic simulation schemes are available that scale independently of population size \citep{Higham2008}.
On the other hand, the importance sampling underlying \zcref{thm:obsc-filter} will in general require effort that is exponential in the number of demes.
For models with many demes, therefore, approaches for ameliorating or circumventing this curse of dimensionality may be necessary.
Critically, the substantial freedom one has in the choice of the importance-sampling kernel $\pi$ can be exploited for this purpose.
In particular, since it is permissible to employ an importance sampling distribution adapted to future observations, there is hope for highly efficient algorithmic computation.

\phantomsection
\addcontentsline{toc}{section}{Acknowledgments}
\section*{Acknowledgments}

We thank Erik Volz, Caroline Colijn, Cris Moore, Ethan Romero-Seversen, Castedo Ellerman, and three anonymous reviewers for their valuable suggestions.
This work was supported by grants from
the U.S. National Science Foundation (Grants \#2526827, \#1761603),
the U.S. National Institutes of Health, (Grants \#1R01AI143852, \#1U54GM111274).
QL acknowledges the support of the Michigan Institute for Data Science.

\clearpage

\phantomsection
\addcontentsline{toc}{section}{Index of symbols and acronyms}

\begin{table}[H]
  \caption{
    Index of symbols and acronyms used in the text.
    \textbf{Boldface} is used to distinguish random elements.
    \label{tab:symbols}
  }
  \vspace{1ex}
  \begin{tabular}{lp{0.6\linewidth}l}
    \hline\hline
    Symbol & Meaning & Definition\\
    \hline
    $t\in\Rp$ & time & \zcref{sec:notation} \\
    \cadlag & right-continuous with left limits & \zcref{sec:notation} \\
    \caglad & left-continuous with right limits & \zcref{sec:notation} \\
    $\indicator{P},\Indicator{P}$ & indicator ($1$ if $P$ is true, $0$ if $P$ is false) & \zcref{sec:notation} \\
    $\Xrt_t,\Yrt_t,\vt(t,x)$, etc. & left-limits of $\Xr_t,\Yr_t,v(t,x)$, etc. & \zcref{sec:notation} \\
    $\rightlim{F},\rightlim{C}$, etc. & right-limits of $F,C$, etc. & \zcref{sec:notation} \\
    $\Xspace$ & state-space of the population process & \zcref{sec:population-process} \\
    $\Xr_t,\X_t$ & population process & \zcref{sec:population-process} \\
    $\Xrh_k,\Xh_k$, $k\in\Zp$ & embedded chain of the population process & \zcref{sec:notation} \\
    $\alpha(t,x,x')$ & hazard of jump from $\Xr=x$ to $\Xr=x'$ & \zcref{sec:population-process} \\
    $p_0(x)$ & probability density of $\Xr_0$ & \zcref{sec:population-process} \\
    KFE & Kolmogorov forward equation & \zcref{sec:kolmogorov-eqns} \\
    KBE & Kolmogorov backward equation & \zcref{sec:kolmogorov-eqns} \\
    $\delta(t,t')$ & one-sided Dirac delta function & \zcref{sec:deltas} \\
    $\Jumps$ & collection of jump marks & \zcref{sec:jump-marks} \\
    $\Ur_t,\U_t$ & jump-mark process & \zcref{sec:jump-marks} \\
    $\Demes$ & collection of demes & \zcref{sec:demes} \\
    $\Hr_t,\H_t$ & history process & \zcref{sec:history-process} \\
    $\leaves$ & set of labels of samples or extant lineages & \zcref{sec:genealogy} \\
    $\partit(\leaves)$ & collection of partitions of subsets of $\leaves$ & \zcref{sec:genealogy} \\
    $\Gr_t,\G_t$ & genealogy process & \zcref{sec:genealogy,sec:genealogy-process} \\
    $T, t(G)$ & genealogy time & \zcref{sec:genealogy} \\
    $\Zr_t,\Z_t$ & tree-structure process & \zcref{sec:genealogy} \\
    $\Yr_t,\Y_t$ & coloring process & \zcref{sec:genealogy} \\
    $\fY(Z)$ & set of colorings compatible with tree-structure $Z$ & \zcref{sec:genealogy} \\
    $\fY_t(Z)$ & set of colorings compatible with tree-structure $Z$ at time $t$ & \zcref{sec:genealogy} \\
    $\col{Y}$ & color (location within the set of demes) & \zcref{sec:genealogy} \\
    $\ctr{Y}$  & event indicator & \zcref{sec:genealogy} \\
    $\G^{\lab{Z}}$ & tree-structure of a genealogy & \zcref{sec:genealogy} \\
    $\G^{\lab{Y}}$ & coloring of a genealogy & \zcref{sec:genealogy} \\
    $n$ & deme-occupancy function & \zcref{sec:demes} \\
    $r^u$ & production of a jump of mark $u\in\Jumps$ & \zcref{sec:genealogy} \\
    $\ell$ & lineage count function & \zcref{sec:ells} \\
    $s$ & saturation function & \zcref{sec:ells} \\
    $Q_u(y,y')$ & compatibility indicator
    ($1$ iff a jump of mark $u$ is compatible with the local structure as described by $y$,$y'$; $0$ else)
    & \zcref{sec:compatibility} \\
    $\tbinratio{n}{\ell}{r}{s}$ & binomial ratio & \zcref{sec:binomial-ratio} \\[1ex]
    $\cols{i}{y}$ & set of branches of a given color & \zcref{sec:swapnchop} \\
    $\swap{ij}{b}{},\backswap{ij}{b}{}$ & swap operators (associated with inline nodes) & \zcref{sec:swapnchop} \\
    $\chop{ij}{ab}{},\backchop{ij}{ab}{}$ & chop operators (associated with samples) & \zcref{sec:swapnchop} \\
    $\fork{ijk}{bb'}{},\backfork{ijk}{bb'}{}$ & fork operators (associated with branch points) & \zcref{sec:swapnchop} \\
    \hline\hline
  \end{tabular}
\end{table}

\clearpage
\appendix
\numberwithin{equation}{section}
\numberwithin{table}{section}
\numberwithin{figure}{section}

\titleformat{\section}[hang]{\large\bfseries}{Appendix \periodafter\thesection}{2ex}{\periodafter}{}
\renewcommand{\thesubsection}{\thesection\arabic{subsection}}
\renewcommand{\theequation}{\thesection\arabic{equation}}
\renewcommand{\thefigure}{\thesection\arabic{figure}}
\renewcommand{\thetable}{\thesection\arabic{table}}
\renewcommand{\thealgorithm}{\thesection\arabic{algorithm}}

\section{Proof of Theorem 1}
%% \label[appendix]{sec:proof}
\label{sec:proof}

We are given a pruned genealogy $\P=(\T,\Z,\Y)$ and a history $\H$ and wish to compute $\CondProb{\Pr_\T=\P}{\Hr_\T=\H}$.
First, note that if $\event{\H}\nsupseteq\event{\P}$, then $\H$ and $\P$ are incompatible and $\CondProb{\Pr_\T=\P}{\Hr_\T=\H}=0$.
Similarly, if any event of $\H$ is incompatible with the local structure of $\P$ in the sense of \zcref{sec:compatibility}, then necessarily $\CondProb{\Pr_\T=\P}{\Hr_\T=\H}=0$.
In either case, the conclusion of the theorem follows.
Let us therefore suppose that neither of these conditions hold.

We can always construct a time-respecting, but otherwise arbitrary ordering of the sample nodes.
Let $\Ph_j$ represent the embedded chain of the pruned genealogy process, \ie the genealogy subtended by the first $j$ samples in this ordering.
Since each $\Ph_j$ contains $\Ph_{j'}$ for all $j'<j$, $\Ph_j$ is a Markov chain.
Therefore, we have the factorization
\begin{equation*}
  \CondProb{\P}{\H}=\prod_{j}{\CondProb{\Ph_j}{\Ph_{j-1},\H}}.
\end{equation*}
Each factor, $\CondProb{\Ph_j}{\Ph_{j-1},\H}$, can itself be factorized.
In particular, suppose $\H$ and $\Ph_{j-1}$ are given and suppose that it is the $K$-th event in $\H$ that introduces the $j$-th sample.
The $j$-th lineage extends backward in time until it either coalesces with $\Ph_{j-1}$ or reaches $t=0$ (\zcref{fig:embedded-chain}).
Moreover, since $\P$ is compatible with $\H$, nothing can happen to this lineage between events of $\H$.
Therefore, we proceed backward, event by event, as follows.
Since the type of event $K$ is given, the color of the $j$-th lineage just prior to that event is known.
At each subsequent event, the lineage either emerges from the event or does not.
Moreover, if the lineage does emerge from the event, the color prior to the event is determined and the color remains unchanged if it does not.
Suppose that lineage $j$ is in deme $i$ just to the right of event $k<K$ and that the deme-$i$ occupancy and production are $n_i$, $r_i$, respectively.
Let $\ell_{ij}$, $s_{ij}$ be the lineage count and saturation at this event in $\Ph_{j-1}$.
Then the new lineage can be any one of the $n_i-\ell_i$ lineages as yet unaccounted for.
Of these, $r_i-s_{ij}$ emerge from the event.
Therefore, the probability that lineage $j$ emerges from the event is $(r_i-s_{ij})/(n_i-\ell_{ij})$ and the probability that it does not is $(n_i-r_i-\ell_{ij}+s_{ij})/(n_i-\ell_{ij})$.
Let $q_{jk}$ be the former if lineage $j$ does in fact emerge from the event and the latter if it does not.
Then $\CondProb{\Ph_j}{\Ph_{j-1},\H}=\prod_{k}{q_{jk}}$ and $\CondProb{\P}{\H}=\prod_{jk}{q_{jk}}$.

\begin{figure}[b]
\begin{knitrout}\small
\definecolor{shadecolor}{rgb}{0.969, 0.969, 0.969}\color{fgcolor}

{\centering \includegraphics[width=0.8\linewidth]{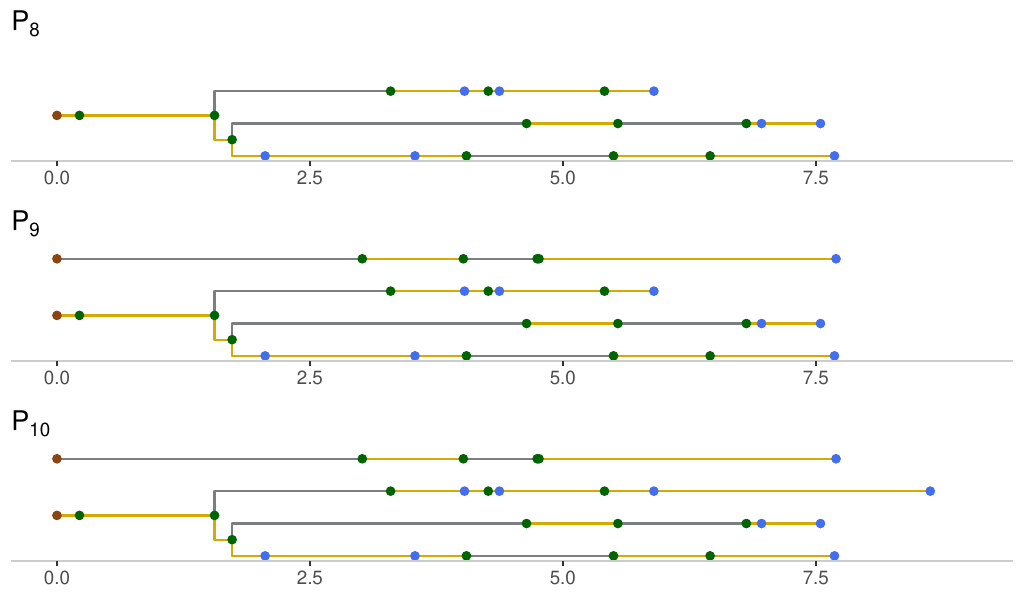}}

\end{knitrout}
  \caption{
    \textbf{Embedded chain of the pruned genealogy process.}
    Three successive states of the embedded chain of the pruned genealogy process are shown.
    At each step, a new lineage is added in a sampling event.
    This lineage extends backward until it either coalesces with the pre-existing genealogy or reaches $t=0$.
    \zcref[S]{thm:pruned-lik} is derived by computing the probability that it coalesces at each possible point, conditional on the history of the population process.
    \label{fig:embedded-chain}
  }
\end{figure}

Now fix $k$ and consider the cumulative product $\prod_{m=1}^{j}{q_{mk}}$.
Clearly $\ell_{1,k}=s_{1,k}=0$ and, if lineage $j$ has not already coalesced to the right of $k$, then $\ell_{j+1,k}=\ell_{j,k}+1$.
Moreover, if $j$ emerges from $k$, then $s_{j+1,k}=s_{j,k}+1$ and $s_{j+1,k}=s_{j,k}$ if it does not.
Therefore, as we accumulate each successive lineage $j$, we move one step to the right in the following diagram.

%% \documentclass[crop,tikz]{standalone}
%% \begin{document}
%% \usetikzlibrary{arrows}
\resizebox{\linewidth}{!}{
  \begin{tikzpicture}[scale=3.2]
    \pgfmathsetmacro{\vh}{-0.75}
    \pgfmathsetmacro{\os}{0.5}
    \tikzstyle{coordinate}=[inner sep=0pt,outer sep=0pt]
    \tikzstyle{solid}=[color=black, thin, >=stealth]
    \tikzstyle{dots}=[color=black, thin, dotted, >=stealth]
    \node (P00) at (0,0*\vh) {$(\ell=0,s=0)$};
    \node (P10) at (1,0*\vh) {$(1,0)$};
    \node (P11) at (1,1*\vh) {$(1,1)$};
    \node (P20) at (2,0*\vh) {$(2,0)$};
    \node (P21) at (2,1*\vh) {$(2,1)$};
    \node (P22) at (2,2*\vh) {$(2,2)$};
    \node (P30) at (3,0*\vh) {$(3,0)$};
    \node (P31) at (3,1*\vh) {$(3,1)$};
    \node (P32) at (3,2*\vh) {$(3,2)$};
    \node (P33) at (3,3*\vh) {$(3,3)$};
    \node (P40) at (4,0*\vh) {};
    \node (P41) at (4,1*\vh) {};
    \node (P42) at (4,2*\vh) {};
    \node (P43) at (4,3*\vh) {};
    \node (P44) at (4,4*\vh) {};
    \node (Plm2sm2) at (4+\os,2*\vh) {};
    \node (Plm2sm1) at (4+\os,3*\vh) {};
    \node (Plm2sm0) at (4+\os,4*\vh) {};
    \node (Plm1sm1) at (5+\os,3*\vh) {$(\ell-1,s-1)$};
    \node (Plm1sm0) at (5+\os,4*\vh) {$(\ell-1,s)$};
    \node (Pls)   at (6+\os,4*\vh) {$(\ell,s)$};
    \draw [solid,->] (P00) -- (P10) node[midway,above,sloped] {$\frac{n-r}{n}$};
    \draw [solid,->] (P00) -- (P11) node[midway,above,sloped] {$\frac{r}{n}$};
    \draw [solid,->] (P10) -- (P20) node[midway,above,sloped] {$\frac{n-r-1}{n-1}$};
    \draw [solid,->] (P10) -- (P21) node[midway,above,sloped] {$\frac{r}{n-1}$};
    \draw [solid,->] (P11) -- (P21) node[midway,above,sloped] {$\frac{n-r}{n-1}$};
    \draw [solid,->] (P11) -- (P22) node[midway,above,sloped] {$\frac{r-1}{n-1}$};
    \draw [solid,->] (P20) -- (P30) node[midway,above,sloped] {$\frac{n-r-2}{n-2}$};
    \draw [solid,->] (P20) -- (P31) node[midway,above,sloped] {$\frac{r}{n-2}$};
    \draw [solid,->] (P21) -- (P31) node[midway,above,sloped] {$\frac{n-r-1}{n-2}$};
    \draw [solid,->] (P21) -- (P32) node[midway,above,sloped] {$\frac{r-1}{n-2}$};
    \draw [solid,->] (P22) -- (P32) node[midway,above,sloped] {$\frac{n-r}{n-2}$};
    \draw [dots,->]  (P22) -- (P33) node[midway,above,sloped] {$\frac{r-2}{n-2}$};
    \draw [dots,->]  (P30) -- (P40) {};
    \draw [dots,->]  (P30) -- (P41) {};
    \draw [dots,->]  (P31) -- (P41) {};
    \draw [dots,->]  (P31) -- (P42) {};
    \draw [dots,->]  (P32) -- (P42) {};
    \draw [dots,->]  (P32) -- (P43) {};
    \draw [dots,->]  (P33) -- (P43) {};
    \draw [dots,->]  (P33) -- (P44) {};
    \draw [dots,->]  (Plm2sm2) -- (Plm1sm1) {};
    \draw [dots,->]  (Plm2sm1) -- (Plm1sm1) {};
    \draw [dots,->]  (Plm2sm1) -- (Plm1sm1) {};
    \draw [dots,->]  (Plm2sm1) -- (Plm1sm0) {};
    \draw [dots,->]  (Plm2sm1) -- (Plm1sm0) {};
    \draw [dots,->]  (Plm2sm0) -- (Plm1sm0) {};
    \draw [solid,->] (Plm1sm1) -- (Pls) node[midway,above,sloped] {$\frac{r-s+1}{n-\ell+1}$};
    \draw [solid,->] (Plm1sm0) -- (Pls) node[near start,above,sloped] {$\frac{n-r-\ell+s}{n-\ell+1}$};
  \end{tikzpicture}
}
%% \end{document}

\noindent
At each step, the lineage count increases by 1 and the saturation either increases by a unit or stays the same.
The probability of each option is indicated on the corresponding arrow.
Although there are multiple paths by which the process may arrive at the final lineage count and saturation, the product of the probabilities along each path is the same:
\begin{equation*}
  \frac{\tbinom{n_i-\ell_i}{r_i-s_i}}{\tbinom{n_i}{r_i}}.
\end{equation*}
Moreover, since each of the options is independent of the others,
\begin{equation*}
  \prod_{j}{q_{jk}}=\binratio{n}{\ell}{r}{s},
\end{equation*}
the latter being the binomial ratio defined in \zcref{sec:binomial-ratio}.

Alternatively, one can reason as follows.
Conditional on $\Hr_\T=\H$, at each time $t\in\event{\H}$, a jump of mark $\U_t$ occurred, with a production of $r^{\U_t}=(r_i)_{i\in\Demes}$, resulting in a deme-occupancy of $n(\X_t)=(n_i)_{i\in\Demes}$.
In $\P$, at time $t$, there are $\ell_i=\ell_i(\Y_t)$ lineages in deme $i$, of which $s_i=s_i(\Y_t)$ are emergent.
By assumption, at each genealogical event, lineages within a deme are exchangeable:
each has an identical probability of being involved.
This exchangeability implies that each lineage present in a deme at time $t$ was equally likely to have been one of the emergent lineages.
In particular, at time $t$, the probability that $s_i$ of the $\ell_i$ deme-$i$ lineages were among the $r_i$ of $n_i$ lineages emergent in the unpruned genealogy process is the same as the probability that, upon drawing $\ell_i$ balls without replacement from an urn containing $r_i$ red balls and $n_i-r_i$ black balls, exactly $s_i$ of the drawn balls are red, namely
\begin{equation*}
  \frac{\binom{n_i-\ell_i}{r_i-s_i}\,\binom{\ell_i}{s_i}}{\binom{n_i}{r_i}}.
\end{equation*}
Because the lineages are labeled, each of the $\tbinom{\ell_i}{s_i}$ equally probable sets of $s_i$ lineages is distinct;
just one of these is the one present in $\P$.
Moreover, since, again conditional on $\Hr_\T=\H$, the identities of the lineages involved in a genealogical event are random and independent of the identities selected at all other events, we have established that
\begin{equation*}
  \CondProb{\Pr_\T=\P}{\Hr_\T=\H}=\prod_{t\in\event{\H}}{\binratio{n(\X_t)}{\ell(\Y_t)}{r^{\U_t}}{s(\Y_t)}}.
\end{equation*}

Returning to the possibility that $\H$ is incompatible with $\P$, since $\Prob{\Pr_\T=\P}=0$ if either any $Q^{}_{\U_t}=0$ or $\event{\P}\nsubseteq\event{\H}$, we obtain the result.
\hfill\qedsymbol

\section{Filter equations}
%% \label[appendix]{sec:filter-eqns}
\label{sec:filter-eqns}

\subsection{Definition}

The likelihoods that appear in \zcref{thm:pruned-filter,thm:obsc-filter} are integrals over large sets of histories.
As such, explicit expressions for them are not available, and we require mathematical tools to allow us to manipulate these quantities and devise algorithms for their numerical solution.
The \emph{filter equations} we introduce here are suitable for these purposes, and we devote this appendix to exposing their essential properties.
This extremely convenient formalism has, to our knowledge, not been thoroughly exploited in the context we introduce here, though we note their resemblance to the constructions of \citet{Ogata1978}, \citet{Puri1986}, \citet{Kliemann1990}, and \citet{Giesecke2018}.

\begin{defn}
  Let $\Xr_t$ be a continuous-time Markov process with KFE
  \begin{mathsize}{9pt}{10pt}
    \begin{equation}
      \label{eq:kfe2}
      \frac{\partial{u}}{\partial{t}}(t,x)
      =\int{u(t,x')\,\beta(t,x',x)\,\dd{x'}}
      -\int{u(t,x)\,\beta(t,x,x')\,\dd{x'}}.
    \end{equation}
  \end{mathsize}%
  Suppose that $B:\Rp\times\Xspace^2\to\Rp$ and $\lambda:\Rp\times\Xspace\to\mathbb{R}$ are are given measurable functions.
  Let $S\subset\Rp$ be locally finite (\ie $S\cap{[0,t]}$ is finite for all $t>0$).
  Then the system of equations
  \begin{mathsize}{9pt}{11pt}
    \begin{align}
      \frac{\partial{w}}{\partial{t}}(t,x)
      &=\int{w(t,x')\,\beta(t,x',x)\,B(t,x',x)\,\dd{x'}}
      -\int{w(t,x)\,\beta(t,x,x')\,\dd{x'}}
      -\lambda(t,x)\,w(t,x),
      \qquad
      &t\notin{S},
      \label{eq:filter-eq-defn-reg}\\
      w(t,x)&=\int{\wt(t,x')\,\beta(t,x',x)\,B(t,x',x)\,\dd{x'}},\qquad
      &t\in{S},
      \label{eq:filter-eq-defn-sing}
    \end{align}
  \end{mathsize}%
  is called the \emph{filter equation} \emph{driven by} $\beta$, with \emph{boost} $B$, \emph{decay} $\lambda$, and \emph{observed event times} $S$.
  Without important confusion, the process $\Xr_t$ can also be said to drive the filter equation.
  \zcref[S]{eq:filter-eq-defn-reg} is the \emph{regular part} of the filter equation;
  \zcref[S]{eq:filter-eq-defn-sing} is known as the \emph{singular part}.
\end{defn}

\begin{remark}
  Trivially, a Kolmogorov forward equation is itself a filter equation with boost $1$, decay $0$, and $S=\emptyset$.
\end{remark}

\subsection{Properties}

The following results show how filter equations allow one to integrate over random histories.
First, \zcref{lemma:reg-filt} shows how one integrates over the full space of histories using a regular filter equation.
\zcref[S]{lemma:sing-filt} builds on this when the set of histories is restricted.

\begin{lemma}
  \label{lemma:reg-filt}
  Suppose that $B:\Rp\times\Xspace^2\to\Rp$ is measurable.
  Let $\Vr_t$ be an $\Rp$-valued random process satisfying $\Vr_0=1$ and
  \begin{equation*}
    \CondExpect{\Vr_t}{\Hr_t=\H_t}=\prod_{\mathclap{e\;\in\;\event{\H_t}}}{B(e,\Xt_e,\X_e)}.
  \end{equation*}
  Let the family of measures $\lambda_t$ on $\Xspace$ be defined by
  \begin{equation*}
    \lambda_t(\mathcal{E})=\Expect{\Vr_t\cdot\Indicator{\Xr_t\in{\mathcal{E}}}},
  \end{equation*}
  for measurable $\mathcal{E}$, and let $w(t,x)$ be the density of $\lambda_t$,
  \ie $\lambda_t(\dd{x})=w(t,x)\,\dd{x}$.
  In particular, $\Expect{\Vr_t}=\lambda_t(\Xspace)=\int{w(t,x)\,\dd{x}}$.
  Then $w$ satisfies the initial condition $w(0,x)=p_0(x)$ and the regular filter equation,
  \begin{equation}
    \label{eq:reg-filter-eq2}
    \frac{\partial{w}}{\partial{t}}=\int{w(t,x')\,\alpha(t,x',x)\,B(t,x',x)\,\dd{x'}}-\int{w(t,x)\,\alpha(t,x,x')\,\dd{x'}}.
  \end{equation}
\end{lemma}
\begin{proof}
  Since $\Vr_0=1$, $\lambda_0(\mathcal{E})=\Prob{\Xr_0\in\mathcal{E}}$, which implies that $w(0,x)=p_0(x)$.
  For $t>0$ and $\Delta>0$, the expectation can be broken into three terms, according to whether $\H_t$ has zero, one, or more than one event in $\halfclosed{t,t+\Delta}$.
  Accordingly, as $\Delta\downarrow{0}$,
  \begin{equation*}
    w(t+\Delta,x)=\left(1-\Delta\,\int{\alpha(t,x,x')\,\dd{x'}}\right)\,w(t,x)+\Delta\,\int{\alpha(t,x',x)\,B(t,x',x)\,w(t,x')\,\dd{x'}}+o(\Delta).
  \end{equation*}
  Forming a difference quotient and taking the limit, we obtain \zcref{eq:reg-filter-eq2}.
  Note that here we use the assumption that $\alpha$ and $B$ are \cadlag\ as functions of $t$.
\end{proof}

When events are known to have occurred at particular times, it is of interest to integrate over those histories that include an event at each of these times.
This leads to singular filter equations, as the next lemma shows.
Before we state the lemma, some terminology is needed.
Let $\mathbb{S}$ be the space of increasing, locally finite sequences in $\Rp$, with the topology induced by the Skorokhod metric and Lebesgue measure.
For $t\in\Rp$ and $S\in\mathbb{S}$, let $S^{}_t\coloneq{S\cap{[0,t]}}$.
Thus if $S\in\mathbb{S}$ and $S^{}_t=(\hat{s}_1,\dots,\hat{s}_K)$, then the infinitesimal element of Lebesgue measure at $S^{}_t$ is $\dd{S^{}_t}=\prod_{n=1}^{K}\dd{\hat{s}_n}$.

\begin{lemma}
  \label{lemma:sing-filt}
  Suppose that $B:\Rp\times\Xspace^2\to\Rp$ is measurable and
  $\Vr_t$ is an $\Rp$-valued random process satisfying
  \begin{equation*}
    \CondExpect{\Vr_t}{\Hr_t=\H_t}=\prod_{\mathclap{e\;\in\;\event{\H_{t}}}}{B(e,\Xt_e,\X_e)}.
  \end{equation*}
  Let $\lambda_t$ be a family of measures on $\Xspace\times\mathbb{S}$ defined by
  \begin{equation*}
    \lambda_t(\mathcal{E},\mathcal{S})=\Expect{\Vr_t\cdot\Indicator{\Xr_t\in{\mathcal{E}}}\cdot\Indicator{\exists{S\in\mathcal{S}}\:\text{s.t.}\:\event{\Hr_t}\supseteq{S^{}_t}}},
  \end{equation*}
  whenever $\mathcal{E}\subseteq\Xspace$ and $\mathcal{S}\subseteq\mathbb{S}$ are measurable.
  Let $w(t,x,S)$ be the density of this measure, \ie
  \begin{equation*}
    \lambda_t(\dd{x}\,\dd{S})=w(t,x,S)\,\dd{x}\,\dd{S^{}_t}.
  \end{equation*}
  Then $w$ satisfies
  \begin{align}
    \label{eq:two-part-filter-reg}
    \frac{\partial{w}}{\partial{t}}(t,x,S)
    &=\int{w(t,x',S)\,\alpha(t,x',x)\,B(t,x',x)\,\dd{x'}}
    -\int{w(t,x,S)\,\alpha(t,x,x')\,\dd{x'}},\qquad
    &t\notin{S},\\
    \label{eq:two-part-filter-sing}
    w(t,x,S)&=\int{\wt(t,x',S)\,\alpha(t,x',x)\,B(t,x',x)\,\dd{x'}},\qquad
    &t\in{S}.
  \end{align}
\end{lemma}
\begin{proof}
  The proof proceeds by induction on the cardinality of $S_t$.
  The base case, for which $S_t=\emptyset$, follows immediately from \zcref{lemma:reg-filt}.
  Assuming that it holds for $|S^{}_t|<K$, one has only to verify \zcref{eq:two-part-filter-sing}.
  This can be accomplished by integrating \zcref{eq:Hdens} directly.
\end{proof}

\begin{remark}
  In the same way that \zcref{eq:kfe-reg,eq:kfe-sing} can be represented as a single equation by means of a Dirac delta notation, the \zcref{eq:two-part-filter-reg,eq:two-part-filter-sing} can be collapsed into a more compact form if $\beta$ is allowed to have atoms at a countable set of time-points and the boost $B$ is adjusted appropriately.
\end{remark}

\subsection{Adjoint filter equation}
\label{sec:adjoint-filter}

Given driver, boost, and decay functions, a filter equation determines a family of measures on $\Xspace$.
Specifically, suppose that, for $s\le{t}$ and $S\subset{\R}$ locally finite, $w(s,x,t,x')$ satisfies the filter equation
\begin{align}
  \label{eq:gen-filt1}
  w(s,x,s,x') = &\indicator{x'=x}\\
  \label{eq:gen-filt2}
  \frac{\partial{w}}{\partial{t}}(s,x,t,x')=
  &\int{w(s,x,t,\xi)\,\beta(t,\xi,x')\,B(t,\xi,x')\,\dd{\xi}}\\
  &\qquad -\int{w(s,x,t,x')\,\beta(t,x',\xi)\,\dd{\xi}}
  -\lambda(t,x')\,w(s,x,t,x'),
  \qquad
  &t\notin{S},\notag\\
  \label{eq:gen-filt3}
  w(s,x,t,x')=&\int{\wt(s,x,t,\xi)\,\beta(t,\xi,x')\,B(t,\xi,x')\,\dd{\xi}},\qquad
  &t\in{S}.
\end{align}
Let $\omega$ be the measure whose density is $w$:
\begin{equation*}
  \omega(s,x,t,\mathcal{E})\coloneq\int_{\mathcal{E}}{w(s,x,t,x')\,\dd{x'}}.
\end{equation*}
If $f:\Xspace\to\R$ is a measurable function, then one defines its integral with respect to this measure in the usual way:
\begin{equation*}
  \omega(s,x,t,f) = \int{\omega(s,x,t,\dd{x'})\,f(x')} = \int{w(s,x,t,x')\,f(x')\,\dd{x'}}.
\end{equation*}
In particular, $\omega(s,x,t,f)$ is itself a measurable function of $x$ and we have the Chapman-Kolmogorov-like relation
\begin{equation}\label{eq:chapkolm}
  \omega(s,x,t,f) = \int{\omega(s,x,\tau,\omega(\tau,\xi,t,f))\,\dd{\xi}} = \int{\omega(s,x,\tau,\dd{\xi})\,\omega(\tau,\xi,t,f)},
\end{equation}
whenever $s\le{\tau}\le{t}$.
Let us now fix $t$ and $f$ and write $F(s,x)=\omega(s,x,t,f)$.
Let us take $\beta$, $B$, and $\lambda$ to be \cadlag\ in $t$ and $F$ to be \caglad\ in $s$.
Applying \zcref{eq:chapkolm} we obtain, for $\Delta>0$ and $s\notin{S}$,
\begin{equation*}
  \begin{split}
    F(s-\Delta,x) = &\left(1-\Delta\,\int{\beta(s-\Delta,x,\xi)\,\dd{\xi}}-\Delta\,\lambda(s-\Delta,x)\right)\,F(s,x)\\
    &\qquad +\Delta\,\int{\beta(s-\Delta,x,\xi)\,B(s-\Delta,x,\xi)\,F(s,\xi)\,\dd{\xi}}+o(\Delta).
  \end{split}
\end{equation*}
Taking $\Delta\downarrow{0}$ gives
\begin{equation}\label{eq:adjfilt1}
  \begin{split}
    -\frac{\partial{F}}{\partial{s}}
    %%    &= \int{\leftlim{\beta}(s,x,\xi)\,\leftlim{B}(s,x,\xi)\,F(s,\xi)\,\dd{\xi}}-\int{\leftlim{\beta}(s,x,\xi)\,F(s,x)\,\dd{\xi}}-\leftlim{\lambda}(s,x)\,F(s,x)\\
    &= \int{\leftlim{\beta}(s,x,\xi)\,\left[\leftlim{B}(s,x,\xi)\,F(s,\xi)-F(s,x)\right]\,\dd{\xi}}-\leftlim{\lambda}(s,x)\,F(s,x).
  \end{split}
\end{equation}
For $s\in{S}$, we have
\begin{equation}\label{eq:adjfilt2}
  F(s,x)
  =\int{\omega(s,x,s,\dd{\xi})\,\rightlim{F}(s,\xi)}
  =\int{\beta(s,x,\xi)\,B(s,x,\xi)\,\rightlim{F}(s,\xi)\,\dd{\xi}},
\end{equation}
where $\rightlim{F}$ denotes the right-limit of $F$.
\zcref[S]{eq:adjfilt1,eq:adjfilt2} constitute the \emph{adjoint form} of \zcref{eq:gen-filt2,eq:gen-filt3}.
Together with the final condition $F(t,x)=f(x)$, these determine $\omega(s,x,t,f)$ for all $s\le{t}$.
Note that, when $B=1$ and $\lambda=0$, the adjoint filter equation becomes the familiar Kolmogorov backward equation (\zcref{eq:kbe}) for the driving process.

\subsection{Solving filter equations by sequential Monte Carlo}
\label{sec:smc}

Filter equations afford a convenient means of computing expectations and likelihoods for pure jump processes.
This is facilitated by the following Lemma, the statement of which uses a one-sided Dirac delta function (\zcref{sec:deltas}).

\begin{lemma}
  \label{lemma:monte-carlo}
  \zcref[S]{eq:filter-eq-defn-reg,eq:filter-eq-defn-sing} are satisfied by $w(t,x)=\int_0^{\infty}{v\,u(t,x,v)\,\dd{v}}$, where $u(t,x,v)$ satisfies the KFE
  \begin{mathsize}{9pt}{10pt}
    \begin{equation}
      \begin{aligned}
        \zcsetup{reftype=pluralequation}
        \label{eq:filterlemma}
        \frac{\partial{u}}{\partial{t}}(t,x,v)=
        &\frac{\partial}{\partial{v}}\left[\lambda(t,x)\,v\,u(t,x,v)\right]
        +\int_0^{\infty}\int{u(t,x',v')\,\beta(t,x',x)\,\delta\!\left(B(t,x',x)\,v',v\right)\,\dd{x'}\,\dd{v'}}\\
        &\quad-\int_0^{\infty}\int{u(t,x,v)\,\beta(t,x,x')\,\delta\!\left(B(t,x,x')\,v,v'\right)\,\dd{x'}\,\dd{v'}},
        &t\notin{S},\\
        u(t,x,v)\,\dd{x}=
        &\int_0^{\infty}\int{\ut(t,x',v')\,\pi(t,x',\dd{x})\,\delta\!\left(A(t,x')\,B(t,x',x)\,v',v\right)\,\dd{x'}\,\dd{v'}},&t\in{S}.
      \end{aligned}
    \end{equation}
  \end{mathsize}%
  Here, $A(t,x)\coloneq{\int{\beta(t,x,x')\,\dd{x'}}}$
  and $\pi(t,x,\dd{x'})\coloneq\beta(t,x,x')\,\dd{x'}/A(t,x)$.
\end{lemma}
\begin{proof}
  For each $t\notin{S}$, we have
  \begin{mathsize}{9pt}{10pt}
    \begin{equation*}
      \begin{aligned}
        \frac{\partial{w}}{\partial{t}}(t,x)
        =&\int_0^{\infty}{v\,\frac{\partial{u}}{\partial{t}}(t,x,v)\,\dd{v}}\\
        =&\int_0^{\infty}\int\int_0^{\infty}{v\,u(t,x',v')\,\beta(t,x',x)\,\delta\!\left(B(t,x',x)\,v',v\right)\,\dd{v}\,\dd{x'}\,\dd{v'}}\\
        &\qquad-\int_0^{\infty}\int\int_0^{\infty}{v\,u(t,x,v)\,\beta(t,x,x')\,\delta\!\left(B(t,x,x')\,v,v'\right)\,\dd{v}\,\dd{x'}\,\dd{v'}}\\
        &\qquad+\int_0^{\infty}{v\,\tfrac{\partial}{\partial{v}}\left[\lambda(t,x)\,v\,u(t,x,v)\right]\,\dd{v}}.\\
      \end{aligned}
    \end{equation*}
  \end{mathsize}%
  Here, the non-explosivity assumption guarantees that we can differentiate under the integral sign and exchange the order of integration.
  Moreover, it ensures that $u\to{0}$ as $v\to{\infty}$.
  Hence, by evaluating the first integral with respect to $v$, the second with respect to $v'$, and the third by parts, we obtain
  \begin{mathsize}{9pt}{10pt}
    \begin{equation*}
      %% \begin{aligned}
      %%   \frac{\partial{w}}{\partial{t}}(t,x)
      %%   =&\int{v'\,u(t,x',v')\,\beta(t,x',x)\,B(t,x',x)\,\dd{v'}\,\dd{x'}}
      %%   -\int{v\,u(t,x,v)\,\beta(t,x,x')\,\dd{v}\,\dd{x'}}\\
      %%   &\qquad-\lambda(t,x)\,\int{v\,u(t,x,v)\,\dd{v}},
      %% \end{aligned}
      \frac{\partial{w}}{\partial{t}}(t,x)
      =\int{v'\,u(t,x',v')\,\beta(t,x',x)\,B(t,x',x)\,\dd{v'}\,\dd{x'}}
      -\int{v\,u(t,x,v)\,\beta(t,x,x')\,\dd{v}\,\dd{x'}}
      -\lambda(t,x)\,\int{v\,u(t,x,v)\,\dd{v}},
    \end{equation*}
  \end{mathsize}%
  which is simplified to obtain \zcref{eq:filter-eq-defn-reg}.
  Similarly, at each $t\in{S}$, we have
  \begin{mathsize}{9pt}{10pt}
    \begin{equation*}
      \begin{aligned}
        w(t,x)\,\dd{x}
        =&\int_0^{\infty}\int\int_0^{\infty}{v\,\ut(t,x',v')\,\pi(t,x',\dd{x})\,\delta\!\left(A(t,x')\,B(t,x',x)\,v',v\right)\,\dd{x'}\,\dd{v'}\,\dd{v}}\\
        =&\int_0^{\infty}\int{v'\,\ut(t,x',v')\,\pi(t,x',\dd{x})\,A(t,x')\,B(t,x',x)\,\dd{x'}\,\dd{v'}}\\
        =&\int{\wt(t,x')\,\beta(t,x',x)\,B(t,x',x)\,\dd{x'}}\,\dd{x}.
      \end{aligned}
    \end{equation*}
  \end{mathsize}%
  which is equivalent to \zcref{eq:filter-eq-defn-sing}
\end{proof}

\zcref[S]{eq:filterlemma} are recognizable as the KFE of a certain process $(\Xr_t,\Vr_t)$.
In particular, the driver $\Xr_t$ has KFE \zcref{eq:kfe2}.
$\Vr_t$ is \emph{directed} by $\Xr_t$ in the sense that $\Vr$ has jumps wherever $\Xr$ does:
when $\Xr$ jumps at time $t$ from $x$ to $x'$, $\Vr$ jumps by the multiplicative factor $B(t,x,x')\ge{0}$.
Between jumps, $\Vr_t$ decays deterministically and exponentially at rate $\lambda(t,x)$.
At the known times in $S$, $\Xr$ jumps according to the probability kernel $\pi$ and, $\Vr$ jumps by the factor $A(t,x)\,B(t,x,x')$.
If we view $\Vr_t$ as a weight, then \zcref{lemma:monte-carlo} tells us how the $\Vr_t$-weighted average of $\Xr_t$ evolves in time:
this average is simply $\int{w(t,x)\,\dd{x}}$.
Thus, \zcref{lemma:monte-carlo} gives a recipe for the integration of \zcref{eq:two-part-filter-reg,eq:two-part-filter-sing} in the Monte Carlo sense.
A pseudocode representation of this procedure is given in \zcref{alg:smc}.
For simplicity, \zcref{alg:smc} is written for the time-homogeneous case, in which neither $\beta(t,x)$ nor $\lambda(t,x)$ depend explicitly on $t$.
In the time-inhomogeneous case, the exponential waiting-time and weight updates would be replaced by the corresponding integrated-hazard expressions.
This is a simple sequential importance sampling algorithm and as such can be improved upon in many ways \citep{Doucet2001,Arulampalam2002,Chopin2020,Wills2023}.
Our purpose here is merely to illustrate the implications of \zcref{lemma:monte-carlo} by showing how the features of \zcref{eq:filter-eq-defn-reg,eq:filter-eq-defn-sing} translate into computations, in the simplest possible context.

\clearpage

\begin{algorithm}[H]
  \caption{
    \textbf{Sequential Monte Carlo filter.}
    The following is a simple Monte Carlo scheme for the filter given by \zcref{eq:filter-eq-defn-reg,eq:filter-eq-defn-sing}.
    It takes as input an initial time $t_0$, a final time $T$, and a sequence of observed event times $S=\Set{t_1,\dots,t_K}$.
    We order the latter so that $t_0\le{t_1}<t_2<\cdots<t_K<{t_{K+1}}\coloneq{T}$.
    In addition, we are given event rates $\beta$, a boost function $B$, a decay function $\lambda$, and the initial distribution of latent states $p_0$.
    As in \zcref{lemma:monte-carlo}, we let
    \begin{equation*}
      \begin{gathered}
        A(x)\coloneq{\int{\beta(x,x')\,\dd{x'}}},
        \qquad
        \pi(x,\dd{x'})\coloneq\frac{\beta(x,x')}{A(x)}\,\dd{x'}.
      \end{gathered}
    \end{equation*}
    For a given, fixed ensemble size, $J\in\Zp$, the algorithm returns an unbiased Monte Carlo estimate, $\rdm{\hat{\lik}}$, of the likelihood.
  }
  \label{alg:smc}
  \begin{algorithmic}[1]
    %% \Require{
    %%   times, $S$;
    %%   initial latent-state distribution, $p_0$;
    %%   event rates, $\beta$;
    %%   boost function, $B$;
    %%   decay function, $\lambda$;
    %%   ensemble size, $J\in\Zp$.
    %% }
    %% \Ensure{
    %%   unbiased Monte Carlo likelihood estimate, $\rdm{\lik}$
    %% }
    \Procedure{SMCFilter}{$t_0,T,S,\beta,B,\lambda,p_0,J$}
    \Statex Initialize an ensemble of latent states and weights:
    \State
    $x_j\sim{p_0(\cdot)}$, for $j=1,\dots,J$
    \State
    $w_j\gets{1}$, for $j=1,\dots,J$
    \For{$k=0,\dots,K$}
    \Comment loop over observed event times
    \For{$j=1,\dots,J$}
    \Comment loop over the ensemble
    \State $t\gets{t_{k}}$
    \If{$k>0$}
    \Comment singular part (\zcref{eq:filter-eq-defn-sing})
    \State $x'_j\sim\pi(x_j,\cdot)$
    \State $w_j\gets{w_j\,A(x_j)\,B(t,x_j,x'_j)}$
    \State $x_j\gets{x'_j}$
    \EndIf
    \While{$t<t_{k+1}$}
    \Comment regular part (\zcref{eq:filter-eq-defn-reg})
    \State $\Delta\sim{\mathrm{Exp}(A(x_j))}$
    \Comment Sojourn times are exponentially distributed.
    \If{$t+\Delta<{t_{k+1}}$}
    \State $x'_j\sim{\pi(x_j,\cdot)}$
    \State $t'\gets{t+\Delta}$
    \State $w_j\gets{w_j\,B(t',x_j,x'_j)\,\exp{(-\lambda(x_j)\,\Delta)}}$
    \State $x_j\gets{x'_j}$
    \State $t\gets{t'}$
    \Else
    \State $\Delta\gets{t_{k+1}-t}$
    \State $w_j\gets{w_j\,\exp{(-\lambda(x_j)\,\Delta)}}$
    \State $t\gets{t_{k+1}}$
    \EndIf
    \EndWhile
    \EndFor
    \EndFor
    \State $\rdm{\hat{\lik}}\gets\tfrac{1}{J}\displaystyle\sum_{j}{w_j}$
    \EndProcedure
  \end{algorithmic}
\end{algorithm}

\phantomsection

\end{document}